\documentclass[final,3p,times]{elsarticle}

\makeatletter

\@addtoreset{equation}{section}
\makeatother

\usepackage{amsmath,amsfonts,amssymb,amsthm}
\usepackage{mathrsfs}
\usepackage{graphicx}
\usepackage{braket}
\usepackage{showkeys}
\newtheorem{thm}{Theorem}[section]
\newtheorem{df}[thm]{Definition}
\newtheorem{prp}[thm]{Proposition}

\newtheorem{thmappC}{Theorem C.\!}
\newtheorem{prpappC}[thmappC]{Proposition C.\!}
\newtheorem{lemappC}[thmappC]{Lemma C.\!}

\journal{Elsevier}

\begin{document}

\begin{frontmatter}



\title{Counting rule of Nambu-Goldstone modes for internal and spacetime symmetries: Bogoliubov theory approach}

\author[cems,reckeio]{Daisuke A. Takahashi}\ead{daisuke.takahashi.ss@riken.jp}
\author[reckeio,physkeio]{Muneto Nitta}
\address[cems]{RIKEN Center for Emergent Matter Science (CEMS), Wako, Saitama 351-0198, Japan}
\address[reckeio]{Research and Education Center for Natural Sciences, Keio University, Hiyoshi 4-1-1, Yokohama, Kanagawa 223-8521, Japan}
\address[physkeio]{Department of Physics, Keio University, Hiyoshi 4-1-1, Yokohama, Kanagawa 223-8521, Japan}
\begin{abstract}
When continuous symmetry is spontaneously broken, 
there appear Nambu-Goldstone modes (NGMs) 
with linear or quadratic dispersion relation, 
which is called type-I or type-II, respectively. 
We propose a framework to count these modes 
including the coefficients of the dispersion relations 
by applying the standard Gross-Pitaevskii-Bogoliubov theory. 
Our method is mainly based on (i) zero-mode solutions of the Bogoliubov equation originated from spontaneous symmetry breaking and (ii) their generalized orthogonal relations, which naturally arise from well-known Bogoliubov transformations and are referred to as ``$\sigma$-orthogonality'' in this paper. 
Unlike previous works, 
our framework is applicable without any modification 
to the cases where there are additional zero modes, 
which do not have a symmetry origin, 
such as quasi-NGMs, and/or 
where spacetime symmetry is spontaneously broken 
in the presence of a topological soliton or a vortex. 
As a by-product of the formulation, we also give a compact summary for mathematics of bosonic Bogoliubov equations and Bogoliubov transformations, which becomes a foundation for any problem of Bogoliubov quasiparticles. 
The general results are illustrated by various examples in 
spinor Bose-Einstein condensates (BECs). In particular, the result on the spin-3 BECs includes new findings such as a type-I--type-II transition and an increase of the type-II dispersion coefficient caused by the presence of a linearly-independent pair of zero modes.
\end{abstract}


\end{frontmatter}


\section{Introduction}\label{sec:intro}
It often occurs in nature that 
a continuous symmetry of a system 
is not preserved in the ground state. 
Such a spontaneous symmetry breaking (SSB) is ubiquitous in nature 
from magnetism, superfluidity and 
superconductivity to quantum field theories such as unification of 
fundamental forces.
When such a SSB occurs,  
there must appear gapless modes 
known as Nambu-Goldstone modes (NGMs) 
and low-energy physics is solely determined by 
these degrees of freedom. 

It is generally known that dispersion relations of NGMs 
are always linear in relativistic theories.  
However, the situation is different in non-relativistic systems; 
the dispersion relation can be either linear $ (\epsilon \propto |k|) $  or quadratic  $ (\epsilon \propto k^2) $. 
Also, the number of NGMs coincides with the number 
of generators of broken symmetries 
in relativistic theories, 
but such a relation does not exist in general for 
non-relativistic cases. 
The well-known illustrative examples in condensed matter physics are the Heisenberg ferromagnets and antiferromagnets. In both cases, the Hamiltonian has the $ SO(3) $ spin-rotation symmetry, but that of ground states is reduced to $ SO(2) $, ignoring discrete symmetries. So, the number of broken continuous symmetries is given by $ \dim (SO(3)/SO(2))=2 $ for both cases. However, the type and the number of emergent NGMs are different; while we have only one spin precession mode with quadratic dispersion in the ferromagnetic case, two spin-wave excitations with linear dispersion appear in the antiferromagnetic case. Thus, the question is: How should we determine the number of NGMs having linear and quadratic dispersions in non-relativistic systems?

The first attack to the above-mentioned problem was made by Nielsen and Chadha in 1976.
They classified 
NGMs with linear and quadratic dispersion relations 
to be of type-I and type-II, respectively, 
and summarized the numbers of those modes 
in the form of 
the Nielsen-Chadha inequality \cite{Nielsen:1975hm}. 
(Strictly speaking, they defined type-I (II) by a dispersion relation with an odd (even) power of the momentum, but this classification is not essential in view of today's understanding.) 
After that, in the 21st century, following a novel remark by Nambu \cite{Nambu:2004yia}, 
Watanabe and Brauner conjectured that 
the equality of the Nielsen-Chadha inequality is saturated 
in generic situation and gave a criterion 
to the numbers of type-I and type-II NGMs 
in the form of a matrix, which we call the Watanabe-Brauner (WB) matrix, whose components are commutators of 
generators corresponding to broken symmetries,  
sandwiched by the ground state 
\cite{PhysRevD.84.125013}. 
More recently, this conjecture has been proved by the effective Lagrangian approach on a coset space \cite{PhysRevLett.108.251602} and by Mori's projection operator method \cite{Hidaka:2012ym}. 
In particular,  
the effective Lagrangian approach 
based on a coset space $ G/H $ for a symmetry $G$ spontaneously 
broken down to its subgroup $H$ in the ground state  
\cite{Coleman:1969sm,Callan:1969sn,Leutwyler:1993gf} 
is a very powerful tool to determine 
the low energy dynamics solely from symmetry arguments, 
and was extensively used in Refs.~\cite{PhysRevLett.108.251602,Watanabe:2014fva}. 
These theoretical developments are now in the stage of experimental verification, because various kinds of multicomponent Bose-Einstein condensates (BECs) are realized in ultra cold atomic gases, such as binary mixtures \cite{HoShenoy,Myatt} and spinor BECs \cite{Kawaguchi:2012ii,RevModPhys.85.1191} with spin-1 \cite{PhysRevLett.80.2027,Stenger,JPSJ.67.1822,Ho:1998zz}, spin-2 \cite{PhysRevA.61.033607,PhysRevLett.92.040402,PhysRevLett.92.140403,PhysRevA.69.063604}, and spin-3 \cite{PhysRevLett.94.160401,PhysRevLett.96.190405,PhysRevA.84.053616}. For example, the dispersion relations of the Bogoliubov phonon and the ferromagnetic spin wave are confirmed in Refs. \cite{PhysRevLett.79.553,2014arXiv1404.5631M}.

Thus far, the theory was formulated for 
internal symmetries. 
The number of NGMs
and their dispersion relations are more 
complicated when 
spacetime symmetry such as translations 
and rotations are spontaneously broken. 
See Refs.~\cite{Watanabe:2013iia,Hayata:2013vfa,Brauner:2014aha} 
for recent discussions. 
Spacetime symmetries are also spontaneously broken  
in topologically non-trivial backgrounds, 
such as 
a quantized vortex
\cite{Kobayashi01022014}, 
a domain wall in magnets \cite{Kobayashi:2014xua}, 
and in two-component BECs
\cite{PhysRevA.88.043612,Watanabe:2014zza} 
and 
a skyrmion line in magnets 
\cite{Watanabe:2014pea,Kobayashi:2014eqa}.
In these cases, there appear NGMs localized on/along topological objects. 
The equality in Refs.~\cite{Nielsen:1975hm,PhysRevD.84.125013,PhysRevLett.108.251602,Hidaka:2012ym} holds 
even in these cases, 
but a careful treatment 
of the singularities 
in the core of topological excitations 
is needed to derive non-commutative nature of generators. 

While the effective Lagrangian approach in the coset space $ G/H $ 
used in Refs.~\cite{Coleman:1969sm,Leutwyler:1993gf,PhysRevLett.108.251602,Watanabe:2014fva}
can generally find possible terms by symmetry considerations, an explicit value of the coefficient of dispersion relations can be obtained only by solving each system concretely. Furthermore, the coset space cannot describe the deviation of the order parameter from ground states, so it cannot grasp a correct physical picture for the motions of NGMs. For example, the spin-1 polar BEC, which is a non-magnetic phase of the spin-1 BEC,  has two type-I spin-wave excitations. These excitations induce a small magnetization and hence the order parameter deviates from the polar state \cite{JPSJ.67.1822,Ho:1998zz}. See also Refs.~\cite{PhysRevA.81.063632,PhysRevA.86.063614}. This effect is completely ignored if the description is closed in the coset space, because the phase of the order parameter is fixed. The similar situation also occurs in Heisenberg antiferromagnets. In order to include these effects, we must formulate the theory in a full order-parameter space. Also, the theories so far do not include  gapless modes without 
an origin of SSB. For example, the theory cannot deal with quasi-NGMs \cite{Weinberg:1972fn,Georgi:1975tz,PhysRevLett.105.230406} 
appearing when the order parameter manifold is larger than the symmetry of Lagrangian or Hamiltonian.
	
	In this paper, we formulate a theory of counting rule and dispersion relations for NGMs by the standard Gross-Pitaevskii (GP) and Bogoliubov theories \cite{Gross,Pitaevskii,Bogoliubov}, and settle the above-mentioned remaining problems. Compared to earlier formulations, our theory will be more down-to-earth and easy-to-access, since we do not need a sophisticated modern geometry. Though we illustrate our formulation by the specific multi-component GP model, our formalism can be extended to more general systems. \\ 
	\indent Here we overview the formalism of this paper. Let us consider the $ N $-component Bose-condensed systems in $ d $-dimensional spatial dimension, where the order parameter is given by $ \boldsymbol{\psi}(\boldsymbol{r})=(\psi_1(\boldsymbol{r}),\dots,\psi_N(\boldsymbol{r}))^T,\ \boldsymbol{r}\in\mathbb{R}^d$. The Bogoliubov quasiparticle wavefunctions in this system are described by a $ 2N $-component vector $ \boldsymbol{w}=(\boldsymbol{u}(\boldsymbol{r}),\boldsymbol{v}(\boldsymbol{r}))^T $ with $ \boldsymbol{u}(\boldsymbol{r})=(u_1(\boldsymbol{r}),\dots,u_N(\boldsymbol{r}))^T $ and $ \boldsymbol{v}(\boldsymbol{r})=(v_1(\boldsymbol{r}),\dots,v_N(\boldsymbol{r}))^T $
	\footnote{Here, we regard $ \boldsymbol{w} $ as the whole quasiparticle eigenvector and $ \boldsymbol{u}(\boldsymbol{r}) $ and $ \boldsymbol{v}(\boldsymbol{r}) $ as the expansion coefficients. So, we do not write the argument $ \boldsymbol{r} $ for $ \boldsymbol{w} $. More precisely speaking, it should be interpreted as $ \boldsymbol{w}=\sum_i\int\mathrm{d}\boldsymbol{r}\left[ u_i(\boldsymbol{r})\ket{i,\boldsymbol{r},u}+v_i(\boldsymbol{r})\ket{i,\boldsymbol{r},v}\right] $, where $ \big\{ \ket{i,\boldsymbol{r},\alpha} | 1\le i\le N, \boldsymbol{r}\in\mathbb{R}^d, \alpha=u,v \big\} $ is a basis for the Hilbert space of quasiparticles such that the completeness relation is given by $ 1=\sum_i\int\mathrm{d}\boldsymbol{r} \big(\ket{i,\boldsymbol{r},u}\bra{i,\boldsymbol{r},u}+\ket{i,\boldsymbol{r},v}\bra{i,\boldsymbol{r},v}\big) $. However, we shortly write it as  $ \boldsymbol{w}=(\boldsymbol{u}(\boldsymbol{r}),\boldsymbol{v}(\boldsymbol{r}))^T $, because the precise expression is lengthy.}. 
	Then, the generalized inner product between two quasiparticle wavefunctions $ \boldsymbol{w}_1=(\boldsymbol{u}_1(\boldsymbol{r}),\boldsymbol{v}_1(\boldsymbol{r}))^T $ and $ \boldsymbol{w}_2=(\boldsymbol{u}_2(\boldsymbol{r}),\boldsymbol{v}_2(\boldsymbol{r}))^T $ is defined by
	\begin{align}
		(\boldsymbol{w}_1,\boldsymbol{w}_2)_\sigma:=\int\mathrm{d}\boldsymbol{r}\left[ \boldsymbol{u}_1(\boldsymbol{r})^\dagger\boldsymbol{u}_2(\boldsymbol{r})-\boldsymbol{v}_1(\boldsymbol{r})^\dagger\boldsymbol{v}_2(\boldsymbol{r})  \right]. \label{eq:intro105}
	\end{align}
	This inner product naturally arises from the Bogoliubov transformation of bosonic field operators. It is well-known for the one-component case \cite{Fetter197267,DalfovoGiorginiPitaevskiiStringari,PethickSmith,FetterWalecka}. 
	In this paper, this product and the orthogonality based on it are called a \textit{$ \sigma $-inner product} and \textit{$ \sigma $-orthogonality}. They play a crucially important role to classify NGMs. 
	
	The classification scheme in our theory is summarized as follows. Let us suppose that the system breaks $ n $ continuous symmetries. First, we derive \textit{SSB-originated zero-mode solutions} $ \boldsymbol{w}_1,\dots,\boldsymbol{w}_n $ for zero-energy Bogoliubov equations (Subsec.~\ref{sec:symzero} for internal symmetries and Sec.~\ref{sec:spacetime} for spacetime symmetries). They have the form of
	\begin{align}
		\boldsymbol{w}_i=\begin{pmatrix} Q_i\boldsymbol{\psi}(\boldsymbol{r}) \\ -Q_i^*\boldsymbol{\psi}(\boldsymbol{r})^* \end{pmatrix},\quad i=1,\dots,n, \label{eq:intro110}
	\end{align}
	where $ Q_1,\dots,Q_n $ are generators of the Lie algebra for broken symmetries. 
	Then, whether a given zero mode becomes a ``seed'' of a type-I or type-II NGM is determined as follows:
	\begin{itemize}
		\item If a given zero mode solution  $ \boldsymbol{w}_i $ is  $ \sigma $-orthogonal to all $ \boldsymbol{w}_j $'s (including itself), it gives rise to a type-I mode.
		\item If there exists a pair of zero mode solutions $ \boldsymbol{w}_i $ and $ \boldsymbol{w}_j $ having a non-vanishing $ \sigma $-inner product, these two zero modes yield one type-II mode.
	\end{itemize} 
	On the basis of this criterion, the number of type-II modes can be counted by a Gram matrix defined as follows. Let us define an $ n\times n $  Gram matrix $ P $ with respect to the $ \sigma $-inner product, whose $ (i,j) $-component $ P_{ij} $ is given by
	\begin{align}
		P_{ij}=(\boldsymbol{w}_i,\boldsymbol{w}_j)_\sigma. \label{eq:introgram}
	\end{align}
	Then, the number $ n_{\text{II}} $ of type-II NGMs is given by
	\begin{align}
		n_{\text{II}}=\frac{1}{2}\operatorname{rank}P,
	\end{align}
	and the number of type-I NGMs is given by $ n_{\text{I}}=n -2n_{\text{II}} $.\\
	\indent While the above criterion using Eq.~(\ref{eq:intro105}) is the most general one, we can use a simplified treatment for the $ \sigma $-inner product when the order parameter $ \boldsymbol{\psi} $ has translational symmetries in some directions. In this case we can omit the integration with respect to these directions, since it only gives the factor of the system volume or the delta function. In particular, if  $ \boldsymbol{\psi} $ is spatially uniform, we need no integration (Secs.~\ref{sec:formulation}, \ref{sec:LAofsigmaprod}, and \ref{sec:perturb}). For the cases in which the order parameter leaves a translational symmetry in some direction, see Sec.~\ref{sec:spacetime}.
	
	In the case of the internal symmetry breaking (Secs.~\ref{sec:formulation}, \ref{sec:LAofsigmaprod}, and \ref{sec:perturb}), the above counting scheme based on the Gram matrix is completely equivalent to the counting rule using the WB matrix \cite{PhysRevD.84.125013,PhysRevLett.108.251602,Hidaka:2012ym}. However, we believe that our result will be more useful and general, because we can apply this method even for
	\begin{enumerate}[(i)]
		\item the case of spacetime symmetry breaking \textit{without any modification}. In particular, we do not need a sensitive mathematical treatment for cores of topological excitations to derive the central extension of a Lie algebra and non-commutativity of translation operators \cite{Kobayashi01022014,Kobayashi:2014xua,Watanabe:2014zza,Watanabe:2014pea}. The calculation of $ \sigma $-inner products is generally easier than the derivation of non-commutativity. 
		\item the case in which there exist accidental zero-energy solutions of the Bogoliubov equation $ \boldsymbol{w}_{n+1},\dots,\boldsymbol{w}_{n+m} $ which do not have an SSB origin. What we should do is only to add them in the list of zero modes and reconsider a new Gram matrix of size $ (n+m)\times(n+m) $.
	\end{enumerate}
	Thus, our formulation will give a simpler and unified method to count the number of type-II modes. The examples of (i) can be found in Sec. \ref{sec:spacetime}, in which we discuss Kelvin modes and ripplons. The general aspect of (ii) is discussed in Sec.~\ref{sec:perturb} and \ref{app:basischoice}. One fascinating example of (ii) is the quasi-NGMs in the spin-2 nematic phase (Subsec.~\ref{subsec:spin2}). 
	
	Note that our classification scheme is purely based on the dispersion relations of gapless modes, and different from the type-A,B classification proposed in Ref.~\cite{PhysRevLett.108.251602}. If the NGMs are classified on the basis of the pairing of the degrees of freedom arising from the SSB, the classification of Ref.~\cite{PhysRevLett.108.251602} is still valid even in the presence of additional zero-mode solutions. 

	We also perform the complete block-diagonalization of the WB matrix \cite{PhysRevD.84.125013} in Subsec. \ref{sec:bdwbm}. We clarify that, unlike the original assumption by Nielsen and Chadha \cite{Nielsen:1975hm,PhysRevD.84.125013}, the two zero modes yielding a type-II NGM are not necessarily linearly dependent. Furthermore, we point out that the linear independence of these two zero modes makes the coefficient of the quadratic dispersion relation larger than that of a free particle (Sec.~\ref{sec:perturb}). The illustrative example for this can be found in the F and H phases of the spin-3 BEC (Subsec.~\ref{subsec:spin3}). These findings are new and overlooked in preceding works.

	As a by-product of constructing the whole theory of NGMs, we also provide a self-contained compact summary for mathematics of finite-dimensional Bogoliubov equations and Bogoliubov transformations in Sec.~\ref{sec:LAofsigmaprod}. The \textit{Bogoliubov-hermitian and Bogoliubov-unitary matrices} defined in this section are equivalent to the Bogoliubov equations and Bogoliubov transformations in finite-dimensional systems, respectively. In particular, we would like to spotlight Colpa's results \cite{Colpa1986,Colpa1986II} for positive-semidefinite cases, which become a foundation to formulate the standard form of zero-energy Bogoliubov equations and the perturbation theory in Sec.~\ref{sec:perturb}. The proofs are a little simplified compared to Colpa's original ones. Several fundamental linear-algebraic theorems on the existence of the basis and on the diagonalizability will be useful not only in the problem of NGMs but also in any kind of problem in Bose-condensed systems. 
	
	Several remaining and related issues are discussed in Subsec.~\ref{subsec:discussion} and corresponding Appendices. In Subsec.~\ref{subsec:bigjordan}, we show that the system is unstable if a zero-wavenumber Bogoliubov matrix does not satisfy the positive-semidefinite assumption (Sec.~\ref{sec:LAofsigmaprod}) and has a large Jordan block. The corresponding perturbation theory for a large Jordan block is given in \ref{app:jordan}. We give a general treatment of ``massive'' NGMs \cite{PhysRevLett.110.011602,PhysRevLett.111.021601,Nicolis,2014arXiv1406.6271} in Bogoliubov theory in Subsec.~\ref{eq:subsecexmass} and \ref{app:exmassive}. The resulting perfect tunneling properties \cite{Kovrizhin,Kagan,DanshitaYokoshiKurihara,doi:10.1143/JPSJ.77.013602,PhysRevA.78.063611,PhysRevA.83.033627,PhysRevA.83.053624,PhysRevA.84.013616} of these NGMs are discussed in Subsec.~\ref{subsec:tunneling}. 

	This paper is organized as follows. Section~\ref{sec:formulation} is devoted to a fundamental setup. In Subsec.~\ref{subsec:GP}, we introduce the multicomponent GP and Bogoliubov equations and define the problem. 
	In Subsec.~\ref{sec:symzero}, we derive SSB-originated zero-mode solutions, which become a central object in this paper. In Subsec.~\ref{sec:bdwbm}, we derive a block-diagonalized form of the WB matrix. 
	In Sec.~\ref{sec:LAofsigmaprod}, we introduce Bogoliubov-hermitian and Bogoliubov-unitary matrices, $ \sigma $-inner products, and  $ \sigma $-orthogonality, and provide several linear-algebraic theorems. 
	Section~\ref{sec:perturb} includes one of main results of this paper; we formulate a perturbation theory for a finite momentum $ k $, and we derive the dispersion relations of type-I and type-II NGMs. 
	In Sec.~\ref{sec:examples}, we provide examples from spinor BECs to illustrate the general results. 
	In Sec.~\ref{sec:spacetime}, as an example of spacetime symmetry breaking, we treat Kelvin modes in one-component BECs and ripplons in two-component BECs. We show that these NGMs have type-II dispersion relations in finite-size systems \cite{Kobayashi01022014,PhysRevA.88.043612} and that the main criterion based on $ \sigma $-orthogonality does not change even in these cases. We also give a perspective for non-integer dispersion relations in infinite-size systems. 
	Section~\ref{sec:summary} is devoted to summary and discussions. In \ref{app:quantum}, we show that our theory is also applicable to the quantum field theory up to the Bogoliubov approximation. In \ref{app:classicalhamilton}, we show the equivalence between the Bogoliubov transformation group and the symplectic group. \ref{app:proof} and \ref{app:basischoice} provide the proofs of the theorems appearing in the main sections. \ref{app:typeI2nd} complements the perturbative calculation of Sec.~\ref{sec:perturb}, where formulae for higher-order terms of eigenvectors and eigenvalues are given. In \ref{app:jordan}, we formulate a perturbation theory when a matrix has a large Jordan block. In \ref{app:exmassive}, we give a general result on ``massive'' NGMs.
\section{Setup of the problem}\label{sec:formulation}
\subsection{Hamiltonian for multicomponent Gross-Pitaevskii field}\label{subsec:GP}
	\indent To make the story simple, we construct a theory in classical field theory. However, as shown in \ref{app:quantum}, the Bogoliubov equation for small oscillations of classical waves is equivalent to that for eigenstates of Bogoliubov quasiparticles in quantum field theory. So our result is also applicable to quantum many body systems up to the Bogoliubov approximation.\\
	\indent We start with the following $N$-component GP (or nonlinear Schr\"odinger) system whose Hamiltonian is given by $ \mathcal{H}=\int h\mathrm{d}x $, where
	\begin{align}
		h=\sum_{i=1}^N\nabla \psi_i^* \nabla \psi_i +F(\boldsymbol{\psi}^*,\boldsymbol{\psi}). \label{eq:GP00}
	\end{align}
	Here we write $ \boldsymbol{\psi}=(\psi_1,\dots,\psi_N)^T $, and $ F(\boldsymbol{\psi}^*,\boldsymbol{\psi}) $ is an abbreviation of $ F(\psi_1^*,\dots,\psi_N^*,\psi_1,\dots,\psi_N) $. The spatial dimension is arbitrary. Since the Hamiltonian must be real,  $ F=F^* $ holds. 
	Though we restrict our formulation to the model (\ref{eq:GP00}), the techniques constructed in this paper can be soon generalized to arbitrary models described by the classical Hamilton mechanics. 
	The multicomponent nonlinear Schr\"odinger equation, which is also called the GP equation in condensed matter theory, is given by
	\begin{align}
		\mathrm{i}\frac{\partial }{\partial t}\psi_i&=\frac{\delta \mathcal{H}}{\delta \psi_i^*}=-\nabla^2\psi_i+\frac{\partial F}{\partial \psi_i^*}, \label{eq:GP01}\\
		-\mathrm{i}\frac{\partial }{\partial t}\psi_i^*&=\frac{\delta \mathcal{H}}{\delta \psi_i}=-\nabla^2\psi_i^*+\frac{\partial F}{\partial \psi_i}. \label{eq:GP02}
	\end{align}
	The linearized waves in the neighbor of a solution of the above equation can be derived by letting $ \psi_i=\psi_i+\delta\psi_i $ and ignoring higher-order terms for $ \delta\psi_i $'s. Writing $ (u_i,v_i)=(\delta\psi_i,\delta\psi_i^*) $, the resultant equations are:
	\begin{align}
		\mathrm{i}\frac{\partial }{\partial t}u_i &= -\nabla^2 u_i+\frac{\partial^2F}{\partial\psi_i^*\partial\psi_j}u_j+\frac{\partial^2F}{\partial\psi_i^*\partial\psi_j^*}v_j, \label{eq:Bogo01}\\
		-\mathrm{i}\frac{\partial }{\partial t}v_i &= -\nabla^2 v_i+\frac{\partial^2F}{\partial\psi_i\partial\psi_j^*}v_j+\frac{\partial^2F}{\partial\psi_i\partial\psi_j}u_j, \label{eq:Bogo02}
	\end{align}
	where the repeated index implies the summation. 
	These equations are equivalent to the Bogoliubov equations appearing in the quantum field theory (\ref{app:quantum}). Note that the different convention  $ v_i=-\delta\psi_i^* $ is also widely used (e.g., \cite{Fetter197267,PethickSmith}). If $ F $ is an analytic function with respect to $ \psi_i $'s and  $ \psi^*_i $'s, we can show
	\begin{align}
		\left( \frac{\partial^2F}{\partial\psi_i^*\partial\psi_j} \right)^*=\frac{\partial^2F}{\partial\psi_i\partial\psi_j^*},\quad \left( \frac{\partial^2F}{\partial\psi_i^*\partial\psi_j^*} \right)^*=\frac{\partial^2F}{\partial\psi_i\partial\psi_j}.
	\end{align}
	Thus, if we write
	\begin{align}
		F_{ij}:=\frac{\partial^2F}{\partial\psi_i^*\partial\psi_j},\quad G_{ij}:=\frac{\partial^2F}{\partial\psi_i^*\partial\psi_j^*} \label{eq:FijGij}
	\end{align}
	then 
	\begin{align}
		F^\dagger=F,\quad G^T=G \label{eq:FijGij2}
	\end{align}
	holds. Introducing a vectorial notation $ \boldsymbol{u}=(u_1,\dots,u_N)^T, \boldsymbol{v}=(v_1,\dots,v_N)^T $, the Bogoliubov equations can be rewritten as
	\begin{align}
		\mathrm{i}\partial_t\boldsymbol{u}=-\nabla^2\boldsymbol{u}+F\boldsymbol{u}+G\boldsymbol{v} ,\quad -\mathrm{i}\partial_t\boldsymbol{v}=-\nabla^2 \boldsymbol{v}+F^*\boldsymbol{v}+G^*\boldsymbol{u}.
	\end{align}
	In particular, when $ \psi_i $'s are stationary and spatially  uniform, assuming the solution of the form $ (\boldsymbol{u},\boldsymbol{v})\propto \mathrm{e}^{\mathrm{i}(\boldsymbol{k}\cdot\boldsymbol{x}-\epsilon t)} $, we obtain
	\begin{align}
		\begin{pmatrix} k^2+F & G \\ -G^* & -k^2-F^* \end{pmatrix}\begin{pmatrix}\boldsymbol{u} \\ \boldsymbol{v} \end{pmatrix}=\epsilon\begin{pmatrix}\boldsymbol{u} \\ \boldsymbol{v} \end{pmatrix} \label{eq:uniformstatBogo}
	\end{align}
	with $ k=|\boldsymbol{k}| $. Thus the determination of dispersion relations of linear waves is reduced to the eigenvalue problem of this $ 2N\times 2N $ matrix. 
	However, what is difficult is that this matrix is \textit{not} hermitian and is not diagonalizable in general. Therefore, in Sec. \ref{sec:LAofsigmaprod}, we provide a self-contained summary for linear algebra necessary to treat the matrix of this type. 
\subsection{SSB-originated zero-mode solutions}\label{sec:symzero}
	Let us consider the case where the Hamiltonian density $ h $ has a symmetry of continuous group. Let  $ G $ be a subgroup of  $ N\times N $ invertible matrices  $ GL(N,\mathbb{C}) $ and assume that for every $ U\in G $, the following holds:
	\begin{align}
		h(U^*\boldsymbol{\psi}^*,U\boldsymbol{\psi})=h(\boldsymbol{\psi}^*,\boldsymbol{\psi}). \label{eq:symham}
	\end{align}
Since the kinetic term  $ \sum_i\nabla\psi_i^*\nabla\psi_i $ must be invariant under this operation, such $ G $ must be a subgroup of the unitary group $ U(N) $. We can immediately prove that
	\begin{align}
		\boldsymbol{\psi}  \text{ is a solution of Eqs. (\ref{eq:GP01}) and (\ref{eq:GP02}).} \quad\leftrightarrow\quad  \boldsymbol{\phi}=U\boldsymbol{\psi} \text{ is also a solution.} \label{eq:symsol}
	\end{align}
	The proof is as follows. We first note that $ F $ also has the same symmetry with $ h $:
	\begin{align}
		F(\boldsymbol{\phi}^*,\boldsymbol{\phi})=F(\boldsymbol{\psi}^*,\boldsymbol{\psi}),
	\end{align}
	where we write $ \boldsymbol{\phi}=U\boldsymbol{\psi},\ \boldsymbol{\phi}^*=U^*\boldsymbol{\psi}^* $. Differentiating both sides of this equation by  $ \psi_i $ or $ \psi_i^* $, and using the unitarity $ U_{ki}U_{ji}^*=\delta_{jk} $, we obtain
	\begin{align}
		U_{ki}\frac{\partial F(\boldsymbol{\psi}^*,\boldsymbol{\psi})}{\partial \psi_i^*}=\frac{\partial F(\boldsymbol{\phi}^*,\boldsymbol{\phi})}{\partial \phi_k^*},\quad U_{ki}^*\frac{\partial F(\boldsymbol{\psi}^*,\boldsymbol{\psi})}{\partial \psi_i}=\frac{\partial F(\boldsymbol{\phi}^*,\boldsymbol{\phi})}{\partial \phi_k}.
	\end{align}
	Multiplying both sides of Eq. (\ref{eq:GP01}) (resp. Eq. (\ref{eq:GP02})) by $ U_{ki} $ (resp. $  U_{ki}^* $), and using the above relations, we obtain what we wanted.\\ 
	\indent Now, let us derive SSB-originated zero mode solutions. Let $ U=U(\alpha) $ be an element of $ G $ parametrized by one real parameter $ \alpha $ such that $ U(0)=I_N $. 
	Then $ \boldsymbol{\phi}=U(\alpha)\boldsymbol{\psi} $ become a one-parameter family of solutions to Eqs. (\ref{eq:GP01}) and (\ref{eq:GP02}), Differentiating Eqs. (\ref{eq:GP01}) and (\ref{eq:GP02}) by $ \alpha $ with  $ \psi_i $'s replaced by $ \phi_i $'s, and setting $ \alpha=0 $ after differentiation, we obtain a particular solution of Bogoliubov equations (\ref{eq:Bogo01}) and (\ref{eq:Bogo02}): 
	\begin{align}
		\begin{pmatrix}\boldsymbol{u} \\ \boldsymbol{v} \end{pmatrix}=\begin{pmatrix}U_{\alpha}\boldsymbol{\psi} \\ U_{\alpha}^*\boldsymbol{\psi}^*\end{pmatrix},\quad U_\alpha:=\left.\frac{\partial U(\alpha)}{\partial \alpha}\right|_{\alpha=0}.
	\end{align}
	In particular, let $ Q $ be a generator of $ G $, and let $ U(\alpha)=\exp(\mathrm{i}\alpha Q) $. Then the solution becomes $ (\boldsymbol{u},\boldsymbol{v})=(\mathrm{i}Q\boldsymbol{\psi},-\mathrm{i}Q^*\boldsymbol{\psi}^*) $. This solution is nonvanishing if $ \boldsymbol{\psi} $ breaks the symmetry of $ Q $, i.e.,  $ \mathrm{e}^{\mathrm{i}\alpha Q}\boldsymbol{\psi}\ne\boldsymbol{\psi} $.   If $ \boldsymbol{\psi} $ does not depend on spacetime, this solution gives a solution of the stationary Bogoliubov equation (\ref{eq:uniformstatBogo}) with $ \epsilon=k=0 $. Henceforth we call this solution an \textit{SSB-originated zero-mode solution}. \\ 
	\indent \underline{\textit{Remark}}. The zero-mode solutions shown above exist if the solution space satisfy the property (\ref{eq:symsol}), even though the Hamiltonian density $ h $ does not have a group symmetry (\ref{eq:symham}). When such a symmetry is spontaneously broken, there appear gapless modes, which are called \textit{quasi-NGMs} \cite{Weinberg:1972fn,Georgi:1975tz}. The quasi-NGMs in spin-2 nematic phase due to $ SO(5) $ symmetry \cite{PhysRevLett.105.230406} can be explained by this kind of symmetry (See Sec.~\ref{sec:examples}).
\subsection{Block-diagonalization of the WB matrix}\label{sec:bdwbm}
	Let $ n $ be a dimension of the symmetry group $ G $ for the Hamiltonian density $ h $, and let $ Q_1,\dots,Q_n $ be a basis for the corresponding Lie algebra. Since $ G $ is a subgroup of the unitary group, $ Q_1,\dots,Q_n $ must be hermitian. Unlike the preceding works \cite{Nielsen:1975hm}, we do not assume $ n<N $. (Note, for example, that the spin-1 BEC \cite{JPSJ.67.1822,Ho:1998zz} has three components but the symmetry group $ U(1)\times SO(3) $ is four-dimensional.) 
	In the previous subsection, the SSB-originated zero-mode solution
	\begin{align}
		\begin{pmatrix}\boldsymbol{u} \\ \boldsymbol{v}\end{pmatrix}=\begin{pmatrix}Q_j\boldsymbol{\psi} \\ -Q_j^*\boldsymbol{\psi}^* \end{pmatrix} \quad (j=1,\dots,n) \label{eq:ssbzero}
	\end{align}
	is shown to be a solution of the Bogoliubov equation (\ref{eq:uniformstatBogo}) with  $ \epsilon=k=0 $. The counting rule for NGMs is described by the WB matrix  $ \rho $, whose  $ (i,j) $-components are defined by \cite{PhysRevD.84.125013,PhysRevLett.108.251602,Hidaka:2012ym}
	\begin{align}
		\rho_{ij}=\mathrm{i}\boldsymbol{\psi}^\dagger[Q_i,Q_j]\boldsymbol{\psi}. \label{eq:wbmatrix} 
	\end{align}
	In this subsection, we show that the block-diagonalized form of $ \rho $ is obtained with properly defined orthogonal relations between zero modes. We also give a few remarks on an assumption used in preceding works. In our theory, it is not indispensable to determine the basis which block-diagonalizes the WB matrix, but such basis will be convenient for the perturbation theory in Sec.~\ref{sec:perturb}, because we can skip perturbative calculations for degenerate eigenvalues. \\ 
	\indent The Lie algebra is a vector space over the real field $ \mathbb{R} $, and $ \{ Q_1,\dots, Q_n \}  $ gives one basis. Let us write this space as $ V $:
	\begin{align}
		V=\left\{ \sum_{i=1}^n r_iQ_i, r_i\in\mathbb{R} \right\} .
	\end{align}
	Note that the coefficient field is \textit{not} $\mathbb{C}$. 
	In what follows, it is important to discuss orthogonality and linear independence by specifying the field explicitly. For example,  $ (1,0) $ and $ (\mathrm{i},0) $ are linearly dependent over $ \mathbb{C} $ but independent over $ \mathbb{R} $. Let $ \boldsymbol{a}=(a_1,\dots,a_N)^T,\ \boldsymbol{b}=(b_1,\dots,b_N)^T, a_i,b_i\in \mathbb{C} $, be vectors. If we consider that the coefficient field is $ \mathbb{C} $, we use the hermitian inner product
	\begin{align}
		(\boldsymbol{a},\boldsymbol{b})_{\mathbb{C}}:=\sum_i a_i^*b_i=\boldsymbol{a}^\dagger\boldsymbol{b}.
	\end{align}
	If we regard the coefficient field as $ \mathbb{R} $, considering the mapping $ \boldsymbol{a}\rightarrow (\operatorname{Re}a_1,\dots,\operatorname{Re}a_N,\operatorname{Im}a_1,\dots,\operatorname{Im}a_N)^T \in \mathbb{R}^{2N} $, we define the inner product by
	\begin{align}
		(\boldsymbol{a},\boldsymbol{b})_{\mathbb{R}}:=\sum_i \left[(\operatorname{Re}a_i)(\operatorname{Re}b_i)+(\operatorname{Im}a_i)(\operatorname{Im}b_i)\right].
	\end{align}
	We can immediately see
	\begin{align}
		\operatorname{Re}(\boldsymbol{a},\boldsymbol{b})_{\mathbb{C}}=(\boldsymbol{a},\boldsymbol{b})_{\mathbb{R}}.
	\end{align}
	\indent Let $ W $ be a subspace of $ V $ whose element annihilates $ \boldsymbol{\psi} $: 
	\begin{align}
		W=\left\{ Q\in V \text{ such that } Q\boldsymbol{\psi}=\boldsymbol{0}  \right\}. \label{eq:nullspace}
	\end{align}
	$ W $ represents the unbroken symmetry of $ \boldsymbol{\psi} $, since $ Q\in W $ implies $ \mathrm{e}^{\mathrm{i}Q}\boldsymbol{\psi}=\boldsymbol{\psi} $. Let $ P $ be a natural map from $ V $ to the quotient space $ V/W $, and let $ t_1,\dots,t_m $ be a basis for $ V/W $, where $ m=n-\dim W $. Furthermore, let us take $ T_1,\dots, T_m \in V $ such that $ PT_i=t_i $. There remains an arbitrariness for each $ T_i $ to add an element of $ W $, but it does not affect the following argument. By definition, $ T_1\boldsymbol{\psi},\dots, T_m\boldsymbol{\psi} $ are linearly independent over $ \mathbb{R} $. 
	Therefore, we can choose an orthonormal basis $ \{ T_1\boldsymbol{\psi},\dots,T_m\boldsymbol{\psi} \} $ which satisfies
	\begin{align}
		(T_i\boldsymbol{\psi},T_j\boldsymbol{\psi})_{\mathbb{R}}=\operatorname{Re}(T_i\boldsymbol{\psi},T_j\boldsymbol{\psi})_{\mathbb{C}}=\delta_{ij}. \label{eq:orthoinR}
	\end{align}
	Once such basis is chosen, the orthonormality is invariant under the real orthogonal transformation, that is, if $ R $ is an $ m\times m $ real orthogonal matrix,  $ T'_i=R_{ij}T_j $ also satisfies orthonormality. \\ 
	\indent Let  $ \rho $ be an $ m\times m $ matrix whose $ (i,j) $-components are given by
	\begin{align}
		\rho_{ij}=\mathrm{i}\boldsymbol{\psi}^\dagger[T_i,T_j]\boldsymbol{\psi}. \label{eq:WB002}
	\end{align}
	Since $ \mathrm{i}[T_i,T_j] $ is hermitian,  $ \rho $ is real and skew-symmetric. Therefore, by an appropriate orthogonal transformation, it can be block-diagonalized as
	\begin{align}
		R\rho R^{-1}=M_1\oplus\dotsb\oplus M_s\oplus O_r,\quad M_i=\begin{pmatrix} 0 & -\mu_i \\ \mu_i & 0 \end{pmatrix}, \label{eq:WB003}
	\end{align}
	where  $ \mu_i>0 $,  $ 2s+r=m $, and  $ O_r $ is a zero matrix of size $ r $. \\
	\indent Since the real orthogonal matrix satisfies $ R_{ik}R_{jk}=\delta_{ij} $, each component of the above equation is given by
	\begin{align}
		(R\rho R^{-1})_{ij}=\mathrm{i}\boldsymbol{\psi}^\dagger[ R_{ik}T_k, R_{jl}T_l ]\boldsymbol{\psi}.
	\end{align}
	Thus, if we choose the basis as $ T_i'=R_{ik}T_k $, then $ \rho $ has the block-diagonalized form. In this basis, let us write the first $ 2s $  $ T_i' $'s as $ X_1^{(1)},X_1^{(2)},\dots, X_s^{(1)},X_s^{(2)} $ and the rest as $ Y_1,\dots, Y_r $. 
	Then we obtain $ \boldsymbol{\psi}^\dagger[Y_i, Y_j]\boldsymbol{\psi}=0 $ and $ \boldsymbol{\psi}^\dagger[Y_i, X_j^{(\alpha)}]\boldsymbol{\psi}=0 $, which are equivalent to  $ \operatorname{Im}(Y_i\boldsymbol{\psi},Y_j\boldsymbol{\psi})_{\mathbb{C}}=0 $ and $ \operatorname{Im}(Y_i\boldsymbol{\psi},X_j^{(\alpha)}\boldsymbol{\psi})_{\mathbb{C}}=0 $. Combining these relations and the orthonormal relation \textit{over} $ \mathbb{R} $ [Eq.~(\ref{eq:orthoinR})], 
	we obtain the orthonormal relations \textit{over} $ \mathbb{C} $: 
	\begin{align}
		(Y_i\boldsymbol{\psi},Y_j\boldsymbol{\psi})_{\mathbb{C}}&=\delta_{ij}, \label{eq:zerortho01} \\
		(Y_i\boldsymbol{\psi},X_j^{(\alpha)}\boldsymbol{\psi})_{\mathbb{C}}&=0.
	\end{align}
	By the same argument for $ X_i^{(\alpha)} $, we also obtain
	\begin{align}
		(X_i^{(\alpha)}\boldsymbol{\psi},X_j^{(\beta)}\boldsymbol{\psi})_{\mathbb{C}}&=\delta_{ij}\left( \delta_{\alpha\beta}+\frac{\mathrm{i}\mu_i}{2}\epsilon_{\alpha\beta} \right), \label{eq:zerortho05}
	\end{align}
	where we do not take a summation over $ i $ in the last expression, and $ \epsilon_{\alpha\beta} $ is defined by $ \epsilon_{11}=\epsilon_{22}=0,\ \epsilon_{12}=-\epsilon_{21}=1 $.  The nontrivial inner product appears only for the pair of  $ X_i^{(1)} $ and $  X_i^{(2)} $. 
	Thus we have obtained the block-diagonalized form of $ \rho $ with orthogonal relations (\ref{eq:zerortho01})-(\ref{eq:zerortho05}). As we will show in Sec.~\ref{sec:perturb}, $ X_i^{(\alpha)} $'s  correspond to type-II NGMs and  $ Y_i $'s correspond to type-I NGMs. So we obtain  $ r $ type-I and $ s $ type-II NGMs, consistent with Refs.~\cite{PhysRevD.84.125013,PhysRevLett.108.251602,Hidaka:2012ym}. \\
	\indent We give a few remarks on  preceding works. In Refs. \cite{Nielsen:1975hm,PhysRevD.84.125013}, the number of type-II modes is identified as the number of linearly dependent pair of zero modes. However, this is not always the case and two zero modes which become a ``seed'' of type-II NGM are linearly independent in general. As shown above, block-diagonalization of the matrix $ \rho $ is possible without using this assumption. 
	The counterexample can be found in spin-3 BECs \cite{PhysRevLett.96.190405,PhysRevA.84.053616}. For the spin-3 F phase $ \boldsymbol{\psi}=(0,1,0,0,0,0,0)^T $, two zero modes related to the type-II excitation are given by $ F_x\boldsymbol{\psi}\propto(\sqrt{6},0,\sqrt{10},0,0,0,0)^T $ and $ F_y\boldsymbol{\psi}=(-\mathrm{i}\sqrt{6},0,\mathrm{i}\sqrt{10},0,0,0,0)^T $, where $ F_x $ and $ F_y $ are matrices of spin-3. These two vectors are linearly independent. On the other hand, in the ferromagnetic spin-1 BEC \cite{JPSJ.67.1822,Ho:1998zz}, the linear dependence holds over $\mathbb{C}$, since $ F_x(1,0,0)^T\propto(0,1,0)^T $ and $ F_y(1,0,0)^T\propto(0,\mathrm{i},0)^T $. Roughly speaking, the linear dependence of two zero modes for type-II NGM is satisfied when some components of the weight vector for $ \boldsymbol{\psi} $ have the highest value. As we will show in Secs. \ref{sec:perturb} and \ref{sec:examples}, the coefficients of the dispersion relation in the case of linearly independent modes are different from those in the case of linearly dependent modes. 
\section{Mathematics of Bogoliubov equations and Bogoliubov transformations}\label{sec:LAofsigmaprod}
	In this section, in order to solve the eigenvalue problem of the bosonic-Bogoliubov type matrix [Eq.~(\ref{eq:uniformstatBogo})], we provide a few theorems from linear algebra. We write the transpose, complex conjugate, and hermitian conjugate of the matrix $ X $ as $ X^T,\ X^*, $ and $ X^\dagger $.\\
	\indent Let us write
	\begin{align}
		\sigma=\sigma_N:=\begin{pmatrix} I_N & \\ & -I_N \end{pmatrix},\quad \tau=\tau_N:=\begin{pmatrix}&I_N \\ I_N & \end{pmatrix}.
	\end{align}
	Then, let us call the  $ 2N\times 2N $ matrices  $ H $ and $ U $  \textit{Bogoliubov-hermitian (B-hermitian)}  and \textit{Bogoliubov-unitary (B-unitary)}, if they satisfy
	\begin{align}
		H=\sigma H^\dagger \sigma,\quad H=-\tau H^* \tau, \label{eq:Bhermitian} \\
		U^{-1}=\sigma U^\dagger \sigma,\quad U=\tau U^*\tau. \label{eq:Bunitary}
	\end{align}
	A B-hermitian matrix can be regarded as an infinitesimal B-unitary matrix, because $ \mathrm{e}^{\mathrm{i}H} $ is B-unitary. This relation is similar to that between hermitian and unitary matrices. There are not a few analogies between the theory of B-hermitian/B-unitary matrices and that of hermitian/unitary matrices (see, e.g., Subsec.~\ref{subsec:bhbu}). \\ 
	\indent The Bogoliubov equation for bosonic systems is generally described by a B-hermitian matrix. A B-unitary matrix defines a Bogoliubov transformation as follows. If $ \hat{a}_1,\dots,\hat{a}_N $ are annihilation operators satisfying the bosonic commutation relations $ [\hat{a}_i,\hat{a}_j^\dagger]=\delta_{ij} $ and $ [\hat{a}_i,\hat{a}_j]=0 $, the operators  $ \hat{b}_1,\dots,\hat{b}_N $ defined by
	\begin{align}
		\boldsymbol{\hat{b}}&=U\boldsymbol{\hat{a}}, \\
		\boldsymbol{\hat{a}}&:=(\hat{a}_1,\dots,\hat{a}_N,\hat{a}_1^\dagger,\dots,\hat{a}_N^\dagger)^T,\\
		\boldsymbol{\hat{b}}&:=(\hat{b}_1,\dots,\hat{b}_N,\hat{b}_1^\dagger,\dots,\hat{b}_N^\dagger)^T
	\end{align}
	also satisfy the same commutation relation. Thus, the diagonalization problem of the quadratic bosonic Hamiltonian in quantum field theory by Bogoliubov transformation is equivalent to the diagonalization of the B-hermitian matrix by B-unitary matrix. Note that B-hermitian matrices are not always diagonalizable, because of the existence of zero-norm eigenvectors, which include, for example, SSB-originated zero modes and unstable modes with complex eigenvalues (see, e.g., Ref.~\cite{MineOkumuraSUnagaYamanaka}).\\ 
	\indent We also note that B-hermitian and B-unitary matrices defined here are equivalent to ``hamiltonian'' and symplectic matrices in classical mechanics\footnote{We add the double quotation mark for ``hamiltonian'' matrices to emphasize that they are \textit{not} hermitian matrices. See also \ref{app:classicalhamilton}.}. Their relation is summarized in \ref{app:classicalhamilton}. Thus, the group of Bogoliubov transformations is equivalent to the symplectic group. Through this point of view, the classification of normal forms for ``hamiltonian'' matrices was already completed long time ago, and a compact summary by Galin\cite{Galin} based on Williamson's work \cite{Williamson} is available in the famous book by Arnold (Appendix 6 of Ref. \cite{Arnold}).\\ 
	\indent The normal forms shown in the above-mentioned book suggest that the ``hamiltonian'' matrices --- or B-hermitian matrices in this paper --- can have arbitrarily large Jordan blocks. However, as shown by Colpa \cite{Colpa1986,Colpa1986II}, if  $ \sigma H $ is positive-semidefinite, we can obtain powerful theorems on the final block-diagonal form. Since the positive-semidefiniteness means the stability of the system (Theorem~\ref{prp:colpa000} and the text below it), if we are only interested in the case where the background condensate is stable, this assumption covers sufficiently many physically relevant situations. 
\subsection{$ \sigma $-inner product and $ \sigma $-orthonormal basis}\label{subsec:sigmabasis}
	We first introduce a $ \sigma $-inner product and $ \sigma $-orthonormal basis. 
	For $ \sigma=\sigma_N $ and $ \boldsymbol{x},\boldsymbol{y} \in \mathbb{C}^{2N} $, we define a \textit{$ \sigma $-inner product} by
	\begin{align}
		(\boldsymbol{x},\boldsymbol{y})_{\sigma}=\boldsymbol{x}^\dagger\sigma\boldsymbol{y}.
	\end{align}
	If $ (\boldsymbol{x},\boldsymbol{y})_\sigma=0 $,  $ \boldsymbol{x} $ and $ \boldsymbol{y} $ are said to be \textit{$ \sigma $-orthogonal}. If $ (\boldsymbol{x},\boldsymbol{x})_\sigma $ is positive, negative, and zero,  $ \boldsymbol{x} $ is said to have \textit{positive}, \textit{negative}, and \textit{zero norm}, respectively. It is also called a \textit{positive-norm}, \textit{negative-norm}, and \textit{zero-norm vector}, respectively. The positive- and negative-norm vectors are called \textit{finite-norm} vectors. If a finite-norm vector satisfies $ (\boldsymbol{x},\boldsymbol{x})_\sigma=\pm1 $, it is said to be normalized.  \\
	\indent A set of linearly independent $ p+q+t=r(\le 2N) $ vectors $ \{ \boldsymbol{x}_1,\dots,\boldsymbol{x}_p,\boldsymbol{y}_1,\dots,\boldsymbol{y}_q,\boldsymbol{z}_1,\dots,\boldsymbol{z}_t \}$ with the following properties is said to be a \textit{$ \sigma $-orthonormal system}:
	\begin{align}
		(\boldsymbol{x}_i,\boldsymbol{x}_j)_\sigma&=\delta_{ij},\quad (\boldsymbol{y}_i,\boldsymbol{y}_j)_\sigma=-\delta_{ij},\quad (\boldsymbol{z}_i,\boldsymbol{z}_j)_\sigma=0, \nonumber \\
		(\boldsymbol{x}_i,\boldsymbol{y}_j)_\sigma&=(\boldsymbol{x}_i,\boldsymbol{z}_j)_\sigma=(\boldsymbol{y}_i,\boldsymbol{z}_j)_\sigma=0.
	\end{align}
	If a basis of an $ r $-dimensional subspace $ V $ of $ \mathbb{C}^{2N} $ satisfies the above relations, the basis is said to be a \textit{$ \sigma $-orthonormal basis of $ V $}. We can prove the following fundamental properties: 
	\begin{enumerate}[(i)]
		\item For any subspace $ V $, there exists a $ \sigma $-orthonormal basis.
		\item A subspace $ W $ of $ V $ spanned by zero-norm vectors $ \boldsymbol{z}_1,\dots,\boldsymbol{z}_t $ does not depend on a choice of basis.
		\item $ p $ and $ q $ are uniquely determined by $ V $. (However, subspaces spanned by positive- and negative-norm vectors depend on a choice of basis.)
		\item $ p,q,t\le N $.
		\item If $ r=2N $, (i.e., if $ V=\mathbb{C}^{2N} $),  $ p=q=N $ and $ t=0 $.
	\end{enumerate}
	\indent If a $ \sigma $-orthonormal system (basis) has the form $ \{\boldsymbol{x}_1,\dots,\boldsymbol{x}_p,\tau\boldsymbol{x}_1^*,\dots,\tau\boldsymbol{x}_p^* \} $ with all $ \boldsymbol{x}_i $'s having positive norm, we call it \textit{B-orthonormal system (basis)}.\\
	\indent The following proposition guarantees that any $ \sigma $-orthonormal (B-orthonormal) system without zero-norm vectors can be extended to a $ \sigma $-orthonormal (B-orthonormal) basis of $ \mathbb{C}^{2N} $.
	\begin{prp}\label{prpsigmabasis2} 
	Let $ \{ \boldsymbol{x}_1,\dots,\boldsymbol{x}_p,\boldsymbol{y}_1,\dots,\boldsymbol{y}_q \} \ (p,q\le N) $ be a $ \sigma $-orthonormal system with $ \boldsymbol{x}_i $'s having positive norm and $ \boldsymbol{y}_i $'s negative norm. By adding new  $ N-p $ positive-norm vectors and $ N-q $ negative-norm vectors to this system, one can construct a $ \sigma $-orthonormal basis for $ \mathbb{C}^{2N} $.
	 In particular, if $ p=q $ and $ \boldsymbol{y}_i=\tau\boldsymbol{x}_i^* $, i.e., the system is B-orthonormal, it can be extended to a B-orthonormal basis  $ \{\boldsymbol{x}_1,\dots,\boldsymbol{x}_N,\tau\boldsymbol{x}_1^*,\dots,\tau\boldsymbol{x}_N^* \} $  for $ \mathbb{C}^{2N} $.  
	\end{prp}
	Proofs for the properties (i)-(v) and Proposition~\ref{prpsigmabasis2} are given in \ref{app:proof}.
\subsection{Properties of Bogoliubov-hermitian and Bogoliubov-unitary matrices}\label{subsec:bhbu}
	Here we list easy-to-prove properties (i)-(xii) for B-hermitian and B-unitary matrices. Let $ U,V $ be B-unitary and $ H, H' $ be B-hermitian. 
	\begin{enumerate}[(i)]
		\item  $ \sigma H $ is hermitian: $ (\sigma H)^\dagger=\sigma H $. 
		\item $ U $ preserves $ \sigma $-inner products: $ (U\boldsymbol{x},U\boldsymbol{y})_\sigma=(\boldsymbol{x},\boldsymbol{y})_\sigma $.  $ H $ is ``$ \sigma $-self-adjoint'': $ (H\boldsymbol{x},\boldsymbol{y})_\sigma=(\boldsymbol{x},H\boldsymbol{y})_\sigma $.
		\item $ U^{-1}H U $ is B-hermitian.  $  U^{-1}VU $ is B-unitary.
		\item $ UV $ is B-unitary. Thus the whole set of B-unitary matrices is a group. (As shown in \ref{app:classicalhamilton}, it is equivalent to the symplectic group.)
		\item The commutator $ \mathrm{i}[H,H'] $ is B-hermitian. $ \mathrm{e}^{\mathrm{i}H} $ is B-unitary. These mean that the whole set of B-hermitian matrices is a Lie algebra of the group of B-unitary matrices. 
	\end{enumerate}
	Let $ \boldsymbol{w} $ and $ \boldsymbol{z} $ be  right eigenvectors of $ H $ with eigenvalues $ \lambda $ and $ \mu $, respectively. Then,
	\begin{enumerate}[(i)]
	\setcounter{enumi}{5}
		\item  $ \boldsymbol{w}^\dagger\sigma $ is a left eigenvector of $ H $ with an eigenvalue $ \lambda^* $.
		\item From (vi), if $ \lambda $ is an eigenvalue of $ H $, $ \lambda^* $ is also an eigenvalue. So there also exists a right eigenvector with an eigenvalue $ \lambda^* $. (However, we cannot express it in a closed form by using $ \boldsymbol{w} $.)
		\item If $ \boldsymbol{w} $ has finite norm $ (\boldsymbol{w},\boldsymbol{w})_\sigma\ne0 $,  $ \lambda $ is real.
		\item If $ \lambda^*\ne \mu $,  $ \boldsymbol{w} $ and $ \boldsymbol{z} $ are $ \sigma $-orthogonal to each other:  $ (\boldsymbol{w},\boldsymbol{z})_\sigma=0 $.
		\item $ \tau\boldsymbol{w}^* $ is a right eigenvector of $ H $ with an eigenvalue $ -\lambda^* $. 
	\end{enumerate}
	Let us write the B-unitary matrix $ U $ as an array of column vectors: $ U=(\boldsymbol{x}_1,\dots,\boldsymbol{x}_N,\tau\boldsymbol{x}_1^*,\dots,\tau\boldsymbol{x}_N^*) $. Then,
	\begin{enumerate}[(i)]
	\setcounter{enumi}{10}
		\item  $ \{ \boldsymbol{x}_1,\dots,\boldsymbol{x}_N,\tau\boldsymbol{x}_1^*,\dots,\tau\boldsymbol{x}_N^* \} $ is a B-orthonormal basis of $ \mathbb{C}^{2N} $. The first $N$ vectors $ \boldsymbol{x}_i $'s have positive norm and the rest $ \tau\boldsymbol{x}_i^* $'s have negative norm.
		\item $ I_{2N}=\sum_{i=1}^N\boldsymbol{x}_i\boldsymbol{x}_i^\dagger\sigma-\sum_{i=1}^N\tau\boldsymbol{x}_i^*\boldsymbol{x}_i^T\tau\sigma $. (An analog of the completeness relation.)
	\end{enumerate}
	Note that (xi) is proved by $ U^{-1}U=I_{2N} $ and (xii) is proved by $ UU^{-1}=I_{2N} $. 
\subsection{Block-diagonalization of Bogoliubov-hermitian matrix for positive-semidefinite case}\label{sec:spsigma}
	In this subsection we block-diagonalize a B-hermitian matrix $ H $ when $ \sigma H $ is positive-semidefinite, following Colpa \cite{Colpa1986} with a few refinements of proofs.\\
	\indent We first define a singular B-hermitian matrix:
	\begin{df}
	If $ K $ is a B-hermitian matrix whose all eigenvectors have zero norm, we call $ K $ \textit{a singular B-hermitian matrix}.
	\end{df} 
	Note that if $ \boldsymbol{z}_1,\boldsymbol{z}_2 $ are the eigenvectors of $ K $ with the same eigenvalue, they are  $ \sigma $-orthogonal to each other:  $ (\boldsymbol{z}_1,\boldsymbol{z}_2)_\sigma=0. $ We can show it by noting that $ \alpha\boldsymbol{z}_1+\beta\boldsymbol{z}_2 $ is also an eigenvector for any $ \alpha,\beta\in \mathbb{C} $. It means that if we construct a  $ \sigma $-orthonormal basis for an eigenspace of some eigenvalue of $ K $, it consists only of zero-norm vectors. 

	\indent The following theorem shows that a B-hermitian matrix can be uniquely decomposed into a real diagonal part and a singular B-hermitian part:
	\begin{thm}\label{prplambdaKdecomp} Let $ H $ be a B-hermitian matrix of size $ 2N $. There exists a B-unitary matrix $ U $ such that
	\begin{align}
		U^{-1}HU=\begin{pmatrix} \Lambda &&& \\ & K_{11} && K_{12} \\ && -\Lambda & \\  & K_{21} && K_{22}  \end{pmatrix},\label{eq:prplambdaKdecomp00}
	\end{align}
	where  $ \Lambda $ is a real diagonal matrix of size $ r \ (0 \le r\le N)$ and 
	\begin{align}
		K=\begin{pmatrix} K_{11} & K_{12} \\ K_{21} & K_{22} \end{pmatrix}\label{eq:prplambdaKdecomp}
	\end{align}
	is a singular B-hermitian matrix of size $ 2(N-r) $. The block $ K_{ij} $'s are of size $ N-r $. If $ r=0 $, there is no diagonal part $ \Lambda $. If $ r=N  $, there is no singular part $ K $.  $ \Lambda $ is unique up to rearrangement of eigenvalues and  $ K $ is unique up to transformation $ K\rightarrow V^{-1}KV $, where $ V $ is a B-unitary matrix. 
	\end{thm}
	The proof is given in \ref{app:proof}. The proof of this theorem is very similar to that of diagonalizability of hermitian matrices by unitary matrices. Therefore, the most crucial difference between mathematics of B-hermitian/B-unitary matrices and that of hermitian/unitary matrices lies in the theory of \textit{zero-norm eigenvectors}. 

	If we consider all kinds of B-hermitian matrices, the singular part $ K $ can generally have a complicated Jordan-block structure. However, when  $ \sigma H $ is positive-semidefinite, the problem becomes very simple. 
	The positive-definite case, which is much easier than the positive-semidefinite one, was first considered by Thouless \cite{Thouless} and revisited by Colpa \cite{Colpa1978}:
	\begin{thm}\label{prp:colpa009} Let $ H $ be a B-hermitian matrix such that $ \sigma H $ is positive-definite. The following (i)-(iii) hold.
	\begin{enumerate}[(i)]
		\item All eigenvalues of  $ H $ are real and nonzero.
		\item Every eigenvector of $ H $ with a positive (negative) eigenvalue has positive (negative) norm.
		\item The singular part $ K $ determined by Theorem~\ref{prplambdaKdecomp} does not exist, i.e., $ H $ is diagonalizable.
	\end{enumerate}
	\end{thm}
	\begin{proof} (i) and (ii): Since $ \sigma H $ is positive-definite,  if $ \boldsymbol{w} $ is an eigenvector of $ H $ with an eigenvalue $ \lambda $,  $ (\boldsymbol{w},\sigma H\boldsymbol{w})_{\mathbb{C}}=(\boldsymbol{w},H\boldsymbol{w})_\sigma=\lambda(\boldsymbol{w},\boldsymbol{w})_\sigma>0 $. Thus $ \lambda\ne 0 $ follows, and since $ (\boldsymbol{w},\boldsymbol{w})_\sigma $ is real,  $ \lambda $ is also real and the sign of $ \lambda $ and $ (\boldsymbol{w},\boldsymbol{w})_\sigma $ are the same. (iii): By (i) and (ii), all eigenvectors have finite norm.
	\end{proof}
	However, Theorem~\ref{prp:colpa009} is not enough for practical use, since the SSB-originated zero-mode solution derived in Subsec.~\ref{sec:symzero} is just a zero-norm eigenvector with zero eigenvalue! A desired theorem suitable for the current purpose can be obtained, if we weaken the assumption of Theorem~\ref{prp:colpa009} and only assume the positive-\textit{semi}definiteness:
	\begin{thm}[Colpa \cite{Colpa1986}]\label{prp:colpa000} Let $ H $ be a B-hermitian matrix such that $ \sigma H $ is positive-semidefinite. The following (i)-(iii) hold.
	\begin{enumerate}[(i)]
		\item All eigenvalues of  $ H $ are real.
		\item Every eigenvector of $ H $ with a positive (negative) eigenvalue has positive (negative) norm.
		\item The singular part $ K $ determined by Theorem~\ref{prplambdaKdecomp} has only zero eigenvalue.
	\end{enumerate}
	\end{thm}
	Note that an eigenvector with zero eigenvalue can have either zero or finite norm, and therefore, the diagonal part $ \Lambda $ in Eq.~(\ref{eq:prplambdaKdecomp00}) can contain zero. \\
	\indent In condensed matter physics, the absence of complex eigenvalue implies the absence of the dynamical instability, and the coincidence of signs between the eigenvalues and the norms means the absence of the Landau instability.
	\begin{proof}[Proof of Theorem \ref{prp:colpa000}] (i) and (ii): We first note that $ H $ and $ \sigma H $ share the same eigenvectors with zero eigenvalue. Since $ \sigma H $ is a positive-semidefinite hermitian matrix,  $ (\boldsymbol{w},\sigma H\boldsymbol{w})_{\mathbb{C}}=(\boldsymbol{w},H\boldsymbol{w})_\sigma=0  $ holds if and only if $ \boldsymbol{w} $ is an eigenvector with zero eigenvalue.  Therefore, if  $ \boldsymbol{w} $ is an eigenvector of $ H $ with an eigenvalue $ \lambda\ne0 $,  $ (\boldsymbol{w},\sigma H\boldsymbol{w})_{\mathbb{C}}=(\boldsymbol{w},H\boldsymbol{w})_\sigma=\lambda(\boldsymbol{w},\boldsymbol{w})_\sigma>0 $. Since $ (\boldsymbol{w},\boldsymbol{w})_\sigma $ is real,  $ \lambda $ is also real and the sign of $ \lambda $ and $ (\boldsymbol{w},\boldsymbol{w})_\sigma $ are the same. (iii) From (ii), any eigenvector with nonzero eigenvalue has finite norm, so it cannot be an eigenvector of a singular B-hermitian matrix.
	\end{proof}
	By Theorem~\ref{prp:colpa000}, the remaining work is to obtain a ``good'' standard form for the singular part $ K $ for the positive-semidefinite case. After a few mathematical constructions, we arrive at the following theorem: 
	\begin{thm}[Colpa \cite{Colpa1986}]\label{prp:colpa001} Let $ H $ be a B-hermitian matrix such that $ \sigma H $ is positive-semidefinite. There exists a B-unitary matrix $ U $ such that
	\begin{align}
		U^{-1}HU=\begin{pmatrix} \Lambda &&& \\ & \tilde{K} && \tilde{K} \\ && -\Lambda & \\  & -\tilde{K} && -\tilde{K} \end{pmatrix},\label{eq:colpathm002}
	\end{align}
	where $ \Lambda $ is a real and non-negative diagonal matrix and $ \tilde{K} $ is a real and positive diagonal matrix.
	\end{thm}
	This theorem is a goal of this section. This theorem also becomes a starting point of perturbation theory in the next section. Since the proof is a little long and technical, we prove it in \ref{app:proof} with detailed mathematical techniques.
\section{Derivation of dispersion relation for type-I and type-II Nambu-Goldstone modes}\label{sec:perturb}
	Now, let us go back to the Bogoliubov equation (\ref{eq:uniformstatBogo}). We write
	\begin{align}
		H=H_0+\sigma k^2,\quad H_0=\begin{pmatrix} F & G \\ -G^* & -F^* \end{pmatrix}.
	\end{align}
	We solve the eigenvalue problem of $ H $ perturbatively, regarding $ H_0 $ as an unperturbed part and  $ \sigma k^2 $ as a perturbation term. We assume  $ \sigma H_0 $ is positive-semidefinite and hence the standard form of Theorem~\ref{prp:colpa001} can be used. Then,  $ \sigma H=\sigma H_0+k^2 I_{2N} $ is positive-definite if $ k>0 $. Thus, by Theorems~\ref{prp:colpa009} and \ref{prp:colpa000}, the system has no Landau and dynamical instability.
\subsection{Block-diagonalization for $ k=0 $}\label{subsec:H0k0}
	\indent We have derived the SSB-originated zero-mode solutions in Subsec.~\ref{sec:bdwbm}. Let us write them as follows:
	\begin{align}
		\boldsymbol{x}_j^{(\alpha)}&:=\begin{pmatrix}X_j^{(\alpha)}\boldsymbol{\psi} \\ -X_j^{(\alpha)*}\boldsymbol{\psi}^* \end{pmatrix},\quad j=1,\dots,s,\quad \alpha=1,2. \label{eq:xjalphas} \\
		\boldsymbol{y}_j&:=\begin{pmatrix}Y_j\boldsymbol{\psi} \\ -Y_j^*\boldsymbol{\psi}^* \end{pmatrix},\quad j=1,\dots.r. \label{eq:yjs}
	\end{align}
	All of them are zero-norm vectors and have the symmetry $ \boldsymbol{x}_j^{(\alpha)}=-\tau\boldsymbol{x}_j^{(\alpha)*} $ and $ \boldsymbol{y}_j=-\tau\boldsymbol{y}_j^* $. Since $ (\boldsymbol{x}_j^{(1)},\boldsymbol{x}_j^{(2)})_\sigma\ne0 $ by Eq.~(\ref{eq:zerortho05}), we can construct a \textit{finite-norm} eigenvector of $ H_0 $ from their linear combination:
	\begin{align}
		\boldsymbol{x}_i=\frac{\boldsymbol{x}_i^{(1)}-\mathrm{i}\boldsymbol{x}_i^{(2)}}{\sqrt{2\mu_i}},\ \tau\boldsymbol{x}_i^*&=-\frac{\boldsymbol{x}_i^{(1)}+\mathrm{i}\boldsymbol{x}_i^{(2)}}{\sqrt{2\mu_i}}. \label{eq:defxi}
	\end{align}
	We can check that $ \boldsymbol{x}_i $ and $ \tau\boldsymbol{x}_i^* $ have positive and negative norm, respectively. 
	As shown in Subsec.~\ref{sec:symzero}, they are eigenvectors of $ H_0 $ with zero eigenvalue:
	\begin{align}
		H_0\boldsymbol{x}_j=\boldsymbol{0},\quad H_0\tau\boldsymbol{x}_j^*=\boldsymbol{0},\quad H_0\boldsymbol{y}_j=\boldsymbol{0}.
	\end{align}
	From Eqs. (\ref{eq:zerortho01})-(\ref{eq:zerortho05}), they satisfy the $ \sigma $-orthogonal relations
	\begin{align}
		&(\boldsymbol{x}_i,\boldsymbol{x}_j)_\sigma=-(\tau\boldsymbol{x}_i^*,\tau\boldsymbol{x}_j^*)_\sigma=\delta_{ij}, \label{eq:orthozero11} \\
		&(\boldsymbol{y}_i,\boldsymbol{y}_j)_\sigma=(\boldsymbol{y}_i,\boldsymbol{x}_j)_\sigma=(\boldsymbol{y}_i,\tau\boldsymbol{x}_j^*)_\sigma=(\boldsymbol{x}_i,\tau\boldsymbol{x}_j^*)_\sigma=0, \label{eq:orthozero111}
	\end{align}
	and the orthogonal relations for the ordinary inner product  $ (\boldsymbol{a},\boldsymbol{b})_{\mathbb{C}}:=(\boldsymbol{a},\sigma\boldsymbol{b})_\sigma=(\sigma\boldsymbol{a},\boldsymbol{b})_\sigma=\boldsymbol{a}^\dagger\boldsymbol{b} $: 
	\begin{align}
		&(\boldsymbol{x}_i,\sigma\boldsymbol{x}_j)_\sigma=(\tau\boldsymbol{x}_i^*,\sigma\tau\boldsymbol{x}_j^*)_\sigma=\frac{1}{\mu_i}\delta_{ij},\quad (\boldsymbol{y}_i,\sigma\boldsymbol{y}_j)_\sigma=2\delta_{ij}, \label{eq:orthozero119} \\
		&(\boldsymbol{x}_i,\sigma\tau\boldsymbol{x}_j^*)_\sigma=(\boldsymbol{y}_i,\sigma\boldsymbol{x}_j)_\sigma=(\boldsymbol{y}_i,\sigma\tau\boldsymbol{x}_j^*)_\sigma=0. \label{eq:orthozero12} 
	\end{align}
	In view of the application to the perturbation theory, it is useful and favorable to write down all orthogonal relations only using $ (\cdot,\cdot)_\sigma $ and without using $ (\cdot,\cdot)_{\mathbb{C}} $. 
	Equation (\ref{eq:orthozero111}) implies that  $ \boldsymbol{y}_i $'s are $ \sigma $-orthogonal to all other zero-mode solutions. On the other hand, the pair $ \boldsymbol{x}_i^{(1)} $ and $ \boldsymbol{x}_i^{(2)} $ has a nonzero $ \sigma $-inner product, and hence, we can construct the finite-norm eigenvectors $ \boldsymbol{x}_i $ and $ \tau\boldsymbol{x}_i^* $. As already mentioned in the Introduction (Sec.~\ref{sec:intro}), and as we will see in Subsecs. \ref{subsec:gram} and \ref{subsec:finitek}, these $ \sigma $-orthogonal relations are directly related to the dispersion relations of NGMs, and the number of type-II modes is in fact a half of the number of finite-norm zero-energy eigenvectors. \\
	\indent Let us assume that $ H_0 $ has no other eigenvectors with zero eigenvalue. Then, $ \boldsymbol{y}_i $'s are $ \sigma $-orthogonal to all other eigenvectors and become a constituent of the non-diagonalizable singular part in Theorem~\ref{prplambdaKdecomp}.  By Theorem~\ref{prp:colpa001}, $ H_0 $ can be block-diagonalized as
	\begin{align}
		U^{-1}H_0U&=\begin{pmatrix} K &&&K&& \\ & O_s &&&& \\ &&\Lambda&&& \\ -K&&& -K && \\  &&&&O_s & \\ &&&&& -\Lambda \end{pmatrix}, \label{eq:UiHU0U} \\
		K&=\operatorname{diag}(\kappa_1,\dots,\kappa_r),\quad \kappa_i>0, \\
		\Lambda&=\operatorname{diag}(\lambda_1,\dots,\lambda_m),\quad \lambda_i>0,
	\end{align}
	where $ O_s $ is a zero matrix of size $ s $, $ m+s+r=N $, and $ U $ has the form of 
	\begin{align}
		U=&\left(\tfrac{\boldsymbol{y}_1+\boldsymbol{z}_1}{2},\dots,\tfrac{\boldsymbol{y}_r+\boldsymbol{z}_r}{2}, \boldsymbol{x}_1,\dots,\boldsymbol{x}_s,\boldsymbol{w}_1,\dots,\boldsymbol{w}_m,\tfrac{-\boldsymbol{y}_1+\boldsymbol{z}_1}{2},\dots,\tfrac{-\boldsymbol{y}_r+\boldsymbol{z}_r}{2},\tau\boldsymbol{x}_1^*,\dots,\tau\boldsymbol{x}_s^*,\tau\boldsymbol{w}_1^*,\dots,\tau\boldsymbol{w}_m^* \right). \label{eq:Ucolumns}
	\end{align}
	Here,  $ \boldsymbol{w}_i $ is a positive-norm eigenvector with a positive eigenvalue $ \lambda_i $:
	\begin{align}
		&H_0\boldsymbol{w}_i=\lambda_i\boldsymbol{w}_i,\quad H_0\tau\boldsymbol{w}_i^*=-\lambda_i\tau\boldsymbol{w}_i^*, \\
		&(\boldsymbol{w}_i,\boldsymbol{w}_j)_\sigma=-(\tau\boldsymbol{w}_i^*,\tau\boldsymbol{w}_j^*)_\sigma=\delta_{ij},
	\end{align}
	and $ \boldsymbol{z}_i $ is a generalized eigenvector satisfying
	\begin{align}
		&H_0\boldsymbol{z}_i=2\kappa_i\boldsymbol{y}_i,\quad \boldsymbol{z}_i=\tau\boldsymbol{z}_i^*, \\
		&(\boldsymbol{z}_i,\boldsymbol{z}_j)_\sigma=0,\quad (\boldsymbol{y}_i,\boldsymbol{z}_j)_\sigma=2\delta_{ij}.
	\end{align}
	By Theorem~ C.\ref{prp:K2}(ii), if $ H_0 $ has only zero eigenvalue, 
	\begin{align}
		\boldsymbol{z}_i&=\sigma\boldsymbol{y}_i,\\
		\sigma H_0\boldsymbol{z}_i&=2\kappa_i\boldsymbol{z}_i
	\end{align}
	hold, and practically we often encounter such case (see Sec.~\ref{sec:examples}). However, at a general level,  $ \boldsymbol{z}_j $ does not have a closed-form expression. Note that the values of $ \mu_i $ and $ \kappa_i $ are uniquely fixed by normalization conditions $ (\boldsymbol{x}_i,\boldsymbol{x}_i)_\sigma=1 $ and $ (\boldsymbol{y}_i,\sigma\boldsymbol{y}_i)_\sigma=(\boldsymbol{y}_i,\boldsymbol{z}_i)_\sigma=2 $.  All other  $ \sigma $-inner products not shown vanish because of B-unitarity of $ U $.\\ 
	\indent Using the notations defined so far,  $ H_0 $ can be written as
	\begin{align}
		H_0=\sum_{i=1}^m \lambda_i\boldsymbol{w}_i\boldsymbol{w}_i^\dagger\sigma+\sum_{i=1}^m \lambda_i\tau\boldsymbol{w}_i^*\boldsymbol{w}_i^T\tau\sigma+\sum_{i=1}^r\kappa_i\boldsymbol{y}_i\boldsymbol{y}_i^\dagger\sigma.
	\end{align}
	It is obtained by multiplying Eq.~(\ref{eq:UiHU0U}) by $U$ from left and $U^{-1}$ from right. An analog of completeness relation [Subsec.~\ref{subsec:bhbu}, (xii)] is given by
	\begin{align}
		I_{2N}=\sum_{i=1}^m\boldsymbol{w}_i\boldsymbol{w}_i^\dagger\sigma-\sum_{i=1}^m\tau\boldsymbol{w}_i^*\boldsymbol{w}_i^T\tau\sigma+\sum_{i=1}^s\boldsymbol{x}_i\boldsymbol{x}_i^\dagger\sigma-\sum_{i=1}^s\tau\boldsymbol{x}_i^*\boldsymbol{x}_i^T\tau\sigma+\sum_{i=1}^r\frac{\boldsymbol{y}_i\boldsymbol{z}_i^\dagger+\boldsymbol{z}_i\boldsymbol{y}_i^\dagger}{2}\sigma. \label{eq:closure1}
	\end{align}
	Since $ U $ [Eq.~(\ref{eq:Ucolumns})] is B-unitary, the set of column vectors
	\begin{align}
		\left\{\tfrac{\boldsymbol{y}_1+\boldsymbol{z}_1}{2},\dots,\tfrac{\boldsymbol{y}_r+\boldsymbol{z}_r}{2}, \boldsymbol{x}_1,\dots,\boldsymbol{x}_s,\boldsymbol{w}_1,\dots,\boldsymbol{w}_m,\tfrac{-\boldsymbol{y}_1+\boldsymbol{z}_1}{2},\dots,\tfrac{-\boldsymbol{y}_r+\boldsymbol{z}_r}{2},\tau\boldsymbol{x}_1^*,\dots,\tau\boldsymbol{x}_s^*,\tau\boldsymbol{w}_1^*,\dots,\tau\boldsymbol{w}_m^* \right\} \label{H0basis}
	\end{align}
	is a B-orthonormal basis, i.e., they are  $ \sigma $-orthogonal to each other, and every vector is normalized and the first $N$ vectors have positive norm and the rest have negative norm. \\ 
	\indent \underline{\textit{Remark}}: Here we have derived the $ \sigma $-orthogonal relations (\ref{eq:orthozero11})-(\ref{eq:orthozero12}) from the properties of SSB-originated zero modes Eqs. (\ref{eq:zerortho01})-(\ref{eq:zerortho05}). However, we can prove that for any B-hermitian matrix $ H_0 $ such that $ \sigma H_0 $ is positive-semidefinite, we can always take a $ \sigma $-orthonormal basis for an eigenspace of zero eigenvalue such that Eqs.~(\ref{eq:orthozero11})-(\ref{eq:orthozero12}) hold. (See \ref{app:basischoice}.) Thus, our theory shown here is applicable even for cases where there are accidental zero-energy eigenvectors which do not have an origin in SSB. \\ 
	\indent For example, if  $ \kappa_1=0 $ occurs by a fine-tuning of a system parameter,  $ \boldsymbol{z}_1 $ also becomes an eigenvector with zero eigenvalue, and $ \boldsymbol{x}_{s+1}=\frac{\boldsymbol{y}_1+\boldsymbol{z}_1}{2} $ and $ \tau\boldsymbol{x}_{s+1}^*=\frac{-\boldsymbol{y}_1+\boldsymbol{z}_1}{2} $ become new finite-norm eigenvectors with zero eigenvalue. Also, if $ \lambda_1=0 $ occurs,  $ \boldsymbol{x}_{s+1}=\boldsymbol{w}_1 $ becomes a new positive-norm eigenvector with zero eigenvalue. These eigenvectors are not originated from an SSB, but yield a gapless mode with type-II dispersion. 
\subsection{Gram matrix}\label{subsec:gram}
	Let us consider the Gram matrix with respect to  $ \sigma $-inner products for zero-mode solutions Eqs. (\ref{eq:xjalphas}) and (\ref{eq:yjs}). We define $ (2N)\times(2s+r) $ matrix by an array of zero-energy eigenvectors (\ref{eq:xjalphas}) and (\ref{eq:yjs}):
	\begin{align}
		A=(\boldsymbol{x}_1^{(1)},\boldsymbol{x}_1^{(2)},\dots,\boldsymbol{x}_s^{(1)},\boldsymbol{x}_s^{(2)},\boldsymbol{y}_1,\dots,\boldsymbol{y}_r). \label{eq:gramA}
	\end{align}
	Then, the Gram matrix of size $ (2s+r)\times(2s+r) $ for these zero-mode solutions can be defined as
	\begin{align}
		P=A^\dagger\sigma A, \label{eq:graminternal}
	\end{align}
	whose components provide the list of values of $ \sigma $-inner products between SSB-originated zero-mode solutions. By definition, it is equal to the WB matrix (\ref{eq:WB002}) up to a constant factor:
	\begin{align}
		P=-\mathrm{i}\rho.
	\end{align}
	From $ \sigma $-orthogonal relations (\ref{eq:orthozero11}) and (\ref{eq:orthozero111}), we immediately find
	\begin{align}
		P=\tilde{M}_1\oplus\dotsb\oplus \tilde{M}_s\oplus O_r,\quad \tilde{M}_i=\begin{pmatrix} 0 & \mathrm{i}\mu_i \\ -\mathrm{i}\mu_i & 0 \end{pmatrix},
	\end{align}
	which is the same with Eq.~(\ref{eq:WB003}). Then, $ s=\frac{1}{2}\operatorname{rank}P $ gives a number of pairs of zero-mode solutions having nonvanishing  $ \sigma $-inner products. As shown later, they give type-II modes. \\ 
	\indent It is obvious that the rank is independent of a choice of the basis, because the rank of $ P $ and $ Q^\dagger P Q $ are the same, where $ Q $ is an invertible matrix. For example, instead of $ \boldsymbol{x}_j^{(\alpha)} $'s, we can use finite-norm vectors, i.e., $ \boldsymbol{x}_j $'s and $ \tau\boldsymbol{x}_j^* $'s:
	\begin{align}
		\tilde{A}&=(\boldsymbol{x}_1,\dots,\boldsymbol{x}_s,\tau\boldsymbol{x}_1^*,\dots,\tau\boldsymbol{x}_s^*,\boldsymbol{y}_1,\dots,\boldsymbol{y}_r), \\
		\tilde{P}&=\tilde{A}^\dagger\sigma \tilde{A}=I_s\oplus(-I_s)\oplus O_r
	\end{align}
	So, we can say that $ 2s $ is the number of finite-norm eigenvectors with zero eigenvalue. \\ 
	\indent Because of the existence of the basis shown in \ref{app:basischoice}, this counting method also remains valid even when there exist accidental zero-energy solutions as stated in the preceding remark. We also mention that the Gram matrix plays a fundamental role in proving fundamental theorems. See its usage in \ref{app:proof} and \ref{app:basischoice}. 
\subsection{Perturbation theory for finite $ k $}\label{subsec:finitek}
	\indent In what follows, we calculate approximate eigenvalues and eigenvectors of $ H=H_0+k^2\sigma  $ by perturbation theory. 
	As we will see,
	\begin{itemize}
		\item The block of $ K $ in Eq.~(\ref{eq:UiHU0U}) gives type-I modes.
		\item The block of $ O_s $ in Eq.~(\ref{eq:UiHU0U}) gives type-II modes.
		\item The block of $ \Lambda $ in Eq.~(\ref{eq:UiHU0U}) gives gapful modes.
	\end{itemize}
	We emphasize that once we have arrived at the standard form of $ H_0 $ [Eq.~(\ref{eq:UiHU0U})], we can completely ``forget'' the physical origin of each of zero modes in the following perturbative calculation. Namely, whether a given zero-energy eigenvector of $ H_0 $ is originated from SSB or is only an accidental solution due to some fine-tuning of system parameters does not have an influence on the following calculation. \\ 
	\indent Let us write the perturbation expansion of eigenvector $ \boldsymbol{\xi} $  by parameter $ k $ as  $ \boldsymbol{\xi}=\boldsymbol{\xi}_0+k\boldsymbol{\xi}_1+k^2\boldsymbol{\xi}_2+\dotsb $, and the expansion of eigenvalue as $ \epsilon=\epsilon_0+\epsilon_1 k+\epsilon_2k^2+\dotsb $. The equations up to second order are given by
	\begin{align}
		H_0\boldsymbol{\xi}_0&=\epsilon_0\boldsymbol{\xi}_0, \\
		H_0\boldsymbol{\xi}_1&=\epsilon_1\boldsymbol{\xi}_0+\epsilon_0\boldsymbol{\xi}_1, \\
		\sigma\boldsymbol{\xi}_0+H_0\boldsymbol{\xi}_2&=\epsilon_2\boldsymbol{\xi}_0+\epsilon_1\boldsymbol{\xi}_1+\epsilon_0\boldsymbol{\xi}_2.
	\end{align}
	If we are interested in the case where  $ \boldsymbol{\xi}_0 $ is an eigenvector of $ H_0 $ with zero eigenvalue, we can set $ H_0\boldsymbol{\xi}_0=0 $ and $ \epsilon_0=0 $. The zeroth order then becomes an identity, and the first and the second order equations become
	\begin{align}
		H_0\boldsymbol{\xi}_1&=\epsilon_1\boldsymbol{\xi}_0, \\
		\sigma\boldsymbol{\xi}_0+H_0\boldsymbol{\xi}_2&=\epsilon_2\boldsymbol{\xi}_0+\epsilon_1\boldsymbol{\xi}_1. \label{eq:prtbX2}
	\end{align}
	In the well-known perturbation theory of hermitian matrices, the expansion is made by the power of a perturbation parameter, which is $ k^2 $ in the present case. However, since we now take a non-diagonalizable and non-hermitian matrix as $ H_0 $, we need to modify the theory. In the present case, the perturbative expansion works well if we expand eigenvectors and eigenvalues by the square root of the perturbation parameter, i.e., $ k=\sqrt{k^2} $. As we will see below, if we do not consider the term of $ O(k^1) $, the coefficient for a zeroth order solution $ \boldsymbol{y}_i $ vanishes.\\ 
	\indent As a zeroth order solution, $ \boldsymbol{\xi}_0 $ can take eigenvectors of zero eigenvalue, i.e., $ \boldsymbol{x}_i,\ \tau\boldsymbol{x}_i^*, $ and $ \boldsymbol{y}_i $. So let us write
	\begin{align}
		\boldsymbol{\xi}_0&=\sum_{i=1}^sa_i^{(0)}\boldsymbol{x}_i+\sum_{i=1}^sb_i^{(0)}\tau\boldsymbol{x}_i^*+\sum_{i=1}^rc_i^{(0)}\boldsymbol{y}_i
	\end{align}
	The higher order terms $ \boldsymbol{\xi}_1,\boldsymbol{\xi}_2,\dots $ can contain all kinds of vectors in the basis (\ref{H0basis}), but we can always eliminate the component of zeroth order solutions by using the arbitrariness such that we can add them to the higher order terms. So, we set
	\begin{align}
		\boldsymbol{\xi}_j&=\sum_{i=1}^rd_i^{(j)}\boldsymbol{z}_i+\sum_{i=1}^{m}\alpha_i^{(j)}\boldsymbol{w}_i+\sum_{i=1}^{m}\beta_i^{(j)}\tau\boldsymbol{w}_i^*,\quad (j\ge 1). \label{eq:prtrbxij}
	\end{align}
	Let us begin to solve the perturbation equation. In general $ H_0\boldsymbol{\xi}_j $ is given by
	\begin{align}
		H_0\boldsymbol{\xi}_j=\sum_{i=1}^r2d_i^{(j)}\kappa_i\boldsymbol{y}_i+\sum_{i=1}^{m}\alpha_i^{(j)}\lambda_i\boldsymbol{w}_i-\sum_{i=1}^{m}\beta_i^{(j)}\lambda_i\tau\boldsymbol{w}_i^*. \label{eq:prtrbxij2}
	\end{align}
	Using this, the first order equation $ H_0\boldsymbol{\xi}_1-\epsilon_1\boldsymbol{\xi}_0=0 $ can be written as
	\begin{align}
		\sum_{i=1}^r(2d_i^{(1)}\kappa_i-\epsilon_1c_i^{(0)})\boldsymbol{y}_i+\sum_{i=1}^{m}\lambda_i(\alpha_i^{(1)}\boldsymbol{w}_i-\beta_i^{(1)}\tau\boldsymbol{w}_i^*)-\epsilon_1\sum_{i=1}^s(a_i^{(0)}\boldsymbol{x}_i+b_i^{(0)}\tau\boldsymbol{x}_i^*)=0.
	\end{align}
	Since the vectors in Eq.~(\ref{H0basis}) are linearly independent, all coefficients of this equation vanish. Thus,
	\begin{align}
		2\kappa_id_i^{(1)}-\epsilon_1c_i^{(0)}&=0, \label{eq:prtrb1storder} \\
		\epsilon_1a_i^{(0)}= \epsilon_1b_i^{(0)}&=0, \label{eq:prtrb1storder2} \\
		\alpha_i^{(1)}=\beta_i^{(1)}&=0. \label{eq:prtrb1storder3}
	\end{align}
	This means $ \boldsymbol{\xi}_1 $ is proportional to $ \epsilon_1 $ and only contains $ \boldsymbol{z}_i $'s:
	\begin{align}
		\boldsymbol{\xi}_1=\epsilon_1\sum_{i=1}^r\frac{c_i^{(0)}}{2\kappa_i}\boldsymbol{z}_i. \label{eq:prtrb1storder4}
	\end{align}
	Henceforth, we consider two cases: $ \epsilon_1\ne0 $ and $ \epsilon_1=0 $. \\ 
	\indent First, let us consider the case $ \epsilon_1\ne0 $. Then $ a_i^{(0)}=b_i^{(0)}=0 $ from Eq.~(\ref{eq:prtrb1storder2}). To determine $ \epsilon_1 $ and $  d_i^{(1)} $, we need one more relation. To derive this, let us take the $ \sigma $-inner product between $ \boldsymbol{y}_i $ and the second order equation (\ref{eq:prtbX2}):
	\begin{align}
		(\boldsymbol{y}_i,\sigma\boldsymbol{\xi}_0)_\sigma=\epsilon_2(\boldsymbol{y}_i,\boldsymbol{\xi}_0)_\sigma+\epsilon_1(\boldsymbol{y}_i,\boldsymbol{\xi}_1)_\sigma
	\end{align}
	Using the $ \sigma $-orthogonal relations of the basis vectors in Eq.~(\ref{H0basis}) and Eqs.~(\ref{eq:orthozero119}) and (\ref{eq:orthozero12}), it reduces to
	\begin{align}
		c_i^{(0)}=\epsilon_1d_i^{(1)}. \label{eq:prtrb1storder02}
	\end{align}
	Therefore, if there is no degeneracy in $ \kappa_i $'s, Eqs. (\ref{eq:prtrb1storder}) and (\ref{eq:prtrb1storder02}) have a solution only when one $ c_j^{(0)} $ is nonzero and all other $ c_i^{(0)} $'s $(i\ne j)$ are zero, and the solution is given by
	\begin{align}
		\epsilon_1=\pm\sqrt{2\kappa_j},\ d_j^{(1)}=\pm\frac{a_j^{(0)}}{\sqrt{2\kappa_j}}.
	\end{align}
	Therefore, if we use $ \boldsymbol{\xi}_0=\boldsymbol{y}_j $ as a zeroth order ``seed'' solution, the eigenvalue and the eigenvector up to first order become
	\begin{align}
		\epsilon&=\pm\sqrt{2\kappa_j}k+O(k^2), \label{eq:typeIvalue} \\
		\boldsymbol{\xi}&=\boldsymbol{y}_j\pm\frac{k}{\sqrt{2\kappa_j}}\boldsymbol{z}_j+O(k^2).\label{eq:typeIvector}
	\end{align}
	We thus obtain the linear dispersion of the NGM originated from $ \boldsymbol{y}_j $. The detailed calculation in \ref{app:typeI2nd} shows that the second order energy vanishes $ (\epsilon_2=0) $ and the expression for  $ \boldsymbol{\xi}_2 $ is given by Eq.~(\ref{eq:appthrid0I}). Thus, 
	\begin{align}
		\epsilon&=\pm\sqrt{2\kappa_j}k+O(k^3), \label{eq:prtrb2storder01} \\
		\boldsymbol{\xi}&=\boldsymbol{y}_j\pm\frac{k}{\sqrt{2\kappa_j}}\boldsymbol{z}_j-k^2\left[ \sum_{i=1}^r\frac{\boldsymbol{z}_i\boldsymbol{z}_i^\dagger}{4\kappa_i}+\sum_{i=1}^{m}\frac{\boldsymbol{w}_i\boldsymbol{w}_i^\dagger}{\lambda_i}+\sum_{i=1}^{m}\frac{\tau\boldsymbol{w}_i^*\boldsymbol{w}_i^T\tau}{\lambda_i} \right]\boldsymbol{y}_j+O(k^3). \label{eq:prtrb2storder02}
	\end{align}
	\indent Next, let us consider the case $ \epsilon_1=0 $. In this case $ \boldsymbol{\xi}_1=\boldsymbol{0} $ by Eq.~(\ref{eq:prtrb1storder4}). 
	Thus the perturbation expansion is given by the power of $ k^2 $, as similar to the conventional perturbation theory. In this case the second order equation (\ref{eq:prtbX2}) becomes 
	\begin{align}
		\sigma \boldsymbol{\xi}_0+H_0\boldsymbol{\xi}_2=\epsilon_2\boldsymbol{\xi}_0. \label{eq:prtbX2-2}
	\end{align}
	Taking the $ \sigma $-inner product between  $ \boldsymbol{y}_i $ and Eq.~(\ref{eq:prtbX2-2}) and using Eqs.~(\ref{eq:orthozero119}) and (\ref{eq:orthozero12}), we first obtain $ c_i^{(0)}=0 $. This result means that if we do not consider the term of $ O(k^1) $, the coefficient of $ \boldsymbol{y}_i $ in zeroth order vanishes. The $ \sigma $-inner product between $ \boldsymbol{x}_i $ and Eq.~(\ref{eq:prtbX2-2}) with using Eqs.~(\ref{eq:orthozero119}) and (\ref{eq:orthozero12}) yields
	\begin{align}
		\epsilon_2=\frac{(\boldsymbol{x}_i,\sigma\boldsymbol{x}_i)_\sigma}{(\boldsymbol{x}_i,\boldsymbol{x}_i)_\sigma}=\frac{1}{\mu_i}
	\end{align}
	Similarly, if we take the $ \sigma $-inner product between $ \tau\boldsymbol{x}_i^* $ and Eq.~(\ref{eq:prtbX2-2}), we obtain
	\begin{align}
		\epsilon_2=\frac{(\tau\boldsymbol{x}_i^*,\sigma\tau\boldsymbol{x}_i^*)_\sigma}{(\tau\boldsymbol{x}_i^*,\tau\boldsymbol{x}_i^*)_\sigma}=-\frac{1}{\mu_i}
	\end{align}
	So we obtain $ \epsilon_2=\pm \mu_i^{-1} $. Thus, if we begin with $ \boldsymbol{\xi}_0=\boldsymbol{x}_i $ or  $ \tau\boldsymbol{x}_i^* $, the eigenvalue up to second order becomes
	\begin{align}
		\epsilon=\pm \frac{1}{\mu_i}k^2+O(k^4). \label{eq:typeiidsprsn}
	\end{align}
	We thus obtain the quadratic dispersion relation for type-II modes. Here we have used $ \epsilon_3=0 $, which is shown in \ref{app:typeI2nd}. We also obtain the eigenvector by a conventional procedure:
	\begin{align}
		\boldsymbol{\xi}=\boldsymbol{\xi}_0-\left[ \sum_{i=1}^r\frac{\boldsymbol{z}_i\boldsymbol{z}_i^\dagger}{4\kappa_i}+\sum_{i=1}^{m}\frac{\boldsymbol{w}_i\boldsymbol{w}_i^\dagger}{\lambda_i}+\sum_{i=1}^{m}\frac{\tau\boldsymbol{w}_i^*\boldsymbol{w}_i^T\tau}{\lambda_i} \right]\boldsymbol{\xi}_0k^2+O(k^4), \label{eq:typeiidsprsn2}
	\end{align}
	where  $ \boldsymbol{\xi}_0=\boldsymbol{x}_i  $ or $  \tau\boldsymbol{x}_i^* $. Here $ \boldsymbol{\xi}_3=\boldsymbol{0} $ is also shown in \ref{app:typeI2nd}. 
	We can indeed check that these eigenvectors and eigenvalues become a solution of $ H\boldsymbol{\xi}=\epsilon\boldsymbol{\xi} $ up to second order by direct substitution and using the completeness relation (\ref{eq:closure1}). \\
	\indent Before closing this section, let us consider the lower bound of the coefficient $ \epsilon_2 $ in detail. Let us write $ \boldsymbol{x}_i $ as $ \boldsymbol{x}_i=(\boldsymbol{u},\boldsymbol{v})^T,\ \boldsymbol{u},\boldsymbol{v}\in \mathbb{C}^N $. Since $ \boldsymbol{x}_i $ is a normalized positive-norm eigenvector,  $ (\boldsymbol{x}_i,\boldsymbol{x}_i)_\sigma=\boldsymbol{u}^\dagger\boldsymbol{u}-\boldsymbol{v}^\dagger\boldsymbol{v}=1 $. Using this, we obtain the inequality
	\begin{gather}
		\epsilon_2^2=\mu_i^{-2}=(\boldsymbol{u}^\dagger\boldsymbol{u}+\boldsymbol{v}^\dagger\boldsymbol{v})^2=1+4(\boldsymbol{u}^\dagger\boldsymbol{u})(\boldsymbol{v}^\dagger\boldsymbol{v})\ge 1. \label{eq:typeiidpndnt0}
	\end{gather}
	The equality holds if and only if $ \boldsymbol{v}=\boldsymbol{0} $. From the definition of $ \boldsymbol{x}_i $ [Eq.~(\ref{eq:defxi})],
	\begin{align}
		\boldsymbol{v}=\boldsymbol{0} \quad\leftrightarrow \quad X_i^{(1)}\boldsymbol{\psi}+\mathrm{i}X_i^{(2)}\boldsymbol{\psi}=\boldsymbol{0}. \label{eq:typeiidpndnt}
	\end{align}
	Namely, \textit{$ |\epsilon_2|=1 $ holds only when two zero modes $ X_i^{(1)}\boldsymbol{\psi} $ and $ X_i^{(2)}\boldsymbol{\psi} $ are linearly dependent.} We can indeed find an example of $ \epsilon_2>1 $ as follows. In the spin-3 BEC F phase, $ \boldsymbol{\psi}=(0,1,0,0,0,0,0)^T $, and two zero modes are given by $ F_x\boldsymbol{\psi}=\frac{1}{2}(\sqrt{6},0,\sqrt{10},0,0,0,0)^T $ and $  F_y\boldsymbol{\psi}=\frac{1}{2}(-i\sqrt{6},0,i\sqrt{10},0,0,0,0)^T $. The normalized eigenvector  $ \boldsymbol{x} $ is then constructed as $ \boldsymbol{x}=(\boldsymbol{u},\boldsymbol{v})^T $ with $ \boldsymbol{u}=\frac{1}{2}(0,0,\sqrt{10},0,0,0,0)^T,\ \boldsymbol{v}=\frac{1}{2}(-\sqrt{6},0,0,0,0,0,0)^T $. We then obtain $ \epsilon_2=\frac{(\boldsymbol{x},\sigma\boldsymbol{x})_\sigma}{(\boldsymbol{x},\boldsymbol{x})_\sigma}=4 $ and the dispersion relation becomes $ \epsilon=4k^2+O(k^4). $ This result is consistent with the exact solution in Subsec.~\ref{subsec:spin3}. Note that this $ \epsilon_2 $ is determined only by the form of $ \boldsymbol{\psi} $ and does not depend on the system parameters, e.g., the coupling constants. 
\section{Examples in spinor Bose-Einstein condensates}\label{sec:examples}
	\indent In this section, we illustrate the general results shown in Sec.~\ref{sec:perturb} by examples of spinor BECs. For spin-$F$ BECs ($F\le 2$), we treat all phases appearing in the phase diagram with zero magnetic field. Probably the spin-1 ferromagnetic phase is a helpful example to understand the standard form of $ H_0 $ [Eq.~(\ref{eq:UiHU0U})], because it has one type-I, one type-II, and one gapful excitations. The spin-2 nematic phase is an interesting example since it has quasi-NGMs. We also consider a few phases of spin-3 BECs, since they show a few new behaviors which are absent in spin-$F$ BECs with $ F\le2 $. See the beginning of Subsec.~\ref{subsec:spin3} for more detail. \\ 
	\indent When $ H_0 $ or some block of $ H_0 $ has only type-I modes, we can use Theorem~C.\ref{prp:K2} to determine the coefficient of the type-I dispersion relation. This is demonstrated in the spin-2 nematic and spin-3 H phases. The spin-0 and the spin-1 polar BEC are also the case, but we do not need to use this technique because the equation is simple.
\subsection{Scalar (spin-0) BEC}
	This is the simplest example such that $ H_0 $ becomes a non-diagonalizable matrix and the type-I NGM appears. The Hamiltonian density with a chemical potential term is given by
	\begin{align}
		h=|\nabla\psi|^2-\mu|\psi|^2+c_0|\psi|^4,
	\end{align}
	where $c_0(>0)$ is a two-body interaction parameter and assumed to be positive in order to stabilize a spatially uniform condensate. The nonlinear Schr\"odinger or the GP equation is given by $ \mathrm{i}\partial_t\psi=-\nabla^2\psi-\mu\psi+2c_0|\psi|^2\psi $, and a uniform solution is given by $ \psi=\sqrt{\rho_0} $ with $ \mu=2c_0\rho_0 $. The Bogoliubov equation is given by
	\begin{align}
		&H\begin{pmatrix}u \\ v \end{pmatrix}=\epsilon\begin{pmatrix}u \\ v \end{pmatrix}, \label{eq:scalarbogo} \\
		&H=H_0+\sigma k^2,\quad H_0=2c_0\rho_0\begin{pmatrix} 1 & 1 \\ -1 & -1 \end{pmatrix}.
	\end{align}
	Thus, $ H_0 $ has $ 1\times 1 \ K $-part and no $ O_s $- and $ \Lambda $-part in the standard form (\ref{eq:UiHU0U}). The system has a $U(1)$-gauge symmetry $ \psi \rightarrow \mathrm{e}^{\mathrm{i}\alpha}\psi $, and the SSB-originated zero mode solution from this symmetry and the corresponding generalized eigenvector are given by
	\begin{align}
		\boldsymbol{y}_1=\begin{pmatrix}1 \\ -1 \end{pmatrix},\quad \boldsymbol{z}_1=\begin{pmatrix}1 \\ 1\end{pmatrix}.
	\end{align}
	They satisfy
	\begin{align}
		&H_0\boldsymbol{y}_1=\boldsymbol{0},\quad H_0\boldsymbol{z}_1=2\kappa_1\boldsymbol{y}_1,\quad \kappa_1=2c_0\rho_0, \\
		&(\boldsymbol{y}_1,\sigma\boldsymbol{y}_1)_\sigma=2,\quad (\boldsymbol{y}_1,\boldsymbol{z}_1)_\sigma=2.
	\end{align}
	The perturbative expansions of them [Eqs. (\ref{eq:typeIvalue}) and (\ref{eq:typeIvector})] are
	\begin{align}
		\epsilon&=\pm 2\sqrt{c_0\rho_0}k+\dotsb, \label{eq:scalarphonon} \\
		\boldsymbol{\xi}&=\boldsymbol{y}_1\pm\frac{k}{2\sqrt{c_0\rho_0}}\boldsymbol{z}_1+\dotsb,
	\end{align}
	which are consistent with the dispersion relation $ \epsilon=\pm\sqrt{4c_0\rho_0k^2+k^4} $ obtained by directly solving Eq.~(\ref{eq:scalarbogo}).
\subsection{Spinor BECs: general}\label{sec:spinorgen}
	Before going to a variety of phases in spin-1, 2, and 3 BECs, we summarize a common aspect of spinor BECs. \\
	\indent The order parameter of the spin-$ F $ BEC consists of $ (2F+1) $-components: $ \boldsymbol{\psi}=(\psi_F,\dots,\psi_{-F})^T $. Let $ F_x,\ F_y, $ and $ F_z $ be  $ (2F+1)\times(2F+1) $ spin-$ F $ matrices. We also use the notation $ F_\pm=F_x\pm\mathrm{i}F_y $. Then, a particle density, magnetization vector, and quadrupole (or nematic) tensor are defined by
	\begin{align}
		\rho=\boldsymbol{\psi}^\dagger\boldsymbol{\psi},\quad M_i=\boldsymbol{\psi}^\dagger F_i\boldsymbol{\psi},\quad N_{ij}=\boldsymbol{\psi}^\dagger \frac{F_iF_j+F_jF_i}{2}\boldsymbol{\psi}, \label{eq:rhoMNdef}
	\end{align}
	respectively, where the indices can take either $ i,j=x,y,z $ or $ i,j=z,+,- $. They behave as rank 0, 1, and 2 tensors under $ SO(3) $-rotation and are invariant under $ U(1) $-gauge transformation, $ \boldsymbol{\psi}'=\mathrm{e}^{\mathrm{i}\varphi}\boldsymbol{\psi} $. We can similarly define octupole and more general $ 2^n $-pole tensors as  $ O_{ijk}=\boldsymbol{\psi}^\dagger\mathcal{S}(F_iF_jF_k)\boldsymbol{\psi} $ and $ O_{i_1i_2,\dotsc,i_n}=\boldsymbol{\psi}^\dagger\mathcal{S}(F_{i_1}F_{i_2}\dotsm F_{i_n})\boldsymbol{\psi} $, where $ \mathcal{S} $ is a symmetrization operator. They behave as rank 3 and $ n $ tensors, respectively\footnote{In view of irreducibility, we should define them as a totally-symmetric traceless tensor, but we use this definition according to the convention.}. 
	By a well-known expression for $ F_i $'s, the components of the above are obtained as
	\begin{align}
		M_z&=\sum_{j=-F}^Fj|\psi_j|^2, \label{eq:FMz} \\
		M_{\pm}&=\sum_{j=-F}^F\sqrt{(F\pm j)(F\mp j+1)}\psi_j^*\psi^{}_{j\mp 1}, \\
		N_{zz}&=\sum_{j=-F}^Fj^2|\psi_j|^2, \label{eq:Nzz}\\
		N_{+-}&=\sum_{j=-F}^F(F(F+1)-j^2)|\psi_j|^2=F(F+1)\rho-N_{zz}, \label{eq:Npm}\\
		N_{z\pm}&=\sum_{j=-F}^F\frac{2j\mp 1}{2}\sqrt{(F\pm j)(F\mp j+1)}\psi_j^*\psi^{}_{j\mp 1}, \\
		N_{\pm\pm}&= \sum_{j=-F}^F\sqrt{(F\pm j)(F\mp j+1)(F\pm j-1)(F\mp j+2)}\psi_j^*\psi_{j\mp 2}^{}. \label{eq:Npmpm}
	\end{align}
	Here and hereafter,  $ \psi_j $'s with $ |j|>F $ are all ignored. 
	Using these quantities, the magnitude of the magnetization vector and the nematic tensor are given by
	\begin{align}
		\boldsymbol{M}^2=M_z^2+M_+M_-,\quad \operatorname{tr}\mathcal{N}^2=N_{zz}^2+\frac{1}{2}N_{+-}^2+2N_{z+}N_{z-}+\frac{1}{2}N_{++}N_{--}.
	\end{align}
	In addition to $ \rho $, we can consider another scalar, i.e., an inner product between $ \boldsymbol{\psi} $ and its time-reversed state:
	\begin{align}
		\Theta=\sum_{j=-F}^F(-1)^j\psi_j\psi_{-j}, \label{eq:FTheta}
	\end{align}
	which is called a singlet pair amplitude. Note that the time-reversed state is given by replacement $ \psi_j \rightarrow (-1)^j\psi_{-j}^* $. This $ \Theta $ is invariant under $ SO(3) $-spin rotation but not invariant under  $ U(1) $-gauge transformation. $ |\Theta|^2 $ is invariant under both operations. \\
	\indent The Hamiltonian density of the spin-$ F $ BEC without magnetic field generally allows scalars which are invariant under $ U(1)\times SO(3) $ transformation, i.e., the overall phase multiplication and the spin rotation: $ \boldsymbol{\psi}'=\mathrm{e}^{\mathrm{i}(\theta+\alpha F_x+\beta F_y+\gamma F_z)}\boldsymbol{\psi} $. As a one-body operator, only the density $ \rho $ is allowed, giving a term of chemical potential $ -\mu N $. We again emphasize that $ \Theta $ is not invariant under the $ U(1) $-gauge transformation. As a two-body interaction term, the candidates of invariants are $ \rho^2, |\Theta|^2, \boldsymbol{M}^2, \operatorname{tr}\mathcal{N}^2, $ and the magnitudes of higher-rank tensors, e.g., $ \sum_{i,j,k}O_{ijk}O_{ijk} $. However, we can generally check that there are only $ F+1 $ linearly independent two-body operators in the spin-$ F $ BEC. For example, we can indeed show the following relations among the invariants:
	\begin{alignat}{2}
		\boldsymbol{M}^2&=\rho^2-|\Theta|^2,\quad \operatorname{tr}\mathcal{N}^2=\frac{3}{2}\rho^2+\frac{1}{2}|\Theta|^2, &\qquad \text{(valid only for spin-1)}, \label{eq:spinor2bdyid01} \\
		\operatorname{tr}\mathcal{N}^2&=12\rho^2+6|\Theta|^2+\frac{3}{2}\boldsymbol{M}^2, &\qquad \text{(valid only for spin-2)}.
	\end{alignat}
	Thus, the Hamiltonian density of spin-$F$ BECs with $ F=1,2, $ and 3 up to two-body interaction terms are given by
	\begin{align}
		h=\sum_{j=-F}^F|\nabla\psi_j|^2-\mu\rho+h_{\text{int}} \label{eq:spinorgen00}
	\end{align}
	with
	\begin{alignat}{2}
		h_{\text{int}}&=c_0\rho^2+c_1|\Theta|^2 &\qquad (\text{spin-1}), \label{eq:spinorgen01} \\
		h_{\text{int}}&=c_0\rho^2+c_1\boldsymbol{M}^2+c_2|\Theta|^2 &\qquad (\text{spin-2}), \label{eq:spinorgen02} \\
		h_{\text{int}}&=\tilde{c}_0\rho^2+\tilde{c}_1\boldsymbol{M}^2+\frac{\tilde{c}_2}{7}|\Theta|^2+\tilde{c}_3\operatorname{tr}\mathcal{N}^2  &\qquad (\text{spin-3}). \label{eq:spinorgen03}
	\end{alignat}
	Here, the coefficients $ \tilde{c}_1,\tilde{c}_2, $ and $ \tilde{c}_3 $ in the spin-3 are taken to be the same with Fig.~8 of Ref.~\cite{PhysRevA.84.053616} for convenience of comparison.\\
	\indent The GP equation is given by
	\begin{align}
		\mathrm{i}\partial_t\psi_j=-\nabla^2\psi_j-\mu\psi_j+\frac{\partial h_{\text{int}}}{\partial\psi_j^*}. \label{eq:spinorgenGP00}
	\end{align}
	To write down the last term explicitly, we need a derivative of two-body interaction terms. From Eqs.~(\ref{eq:FMz})-(\ref{eq:FTheta}), we obtain
	\begin{align}
		\frac{\partial \rho^2}{\partial\psi_j^*}&=2\psi_j\rho, \\
		\frac{\partial |\Theta|^2}{\partial\psi_j^*}&=2(-1)^j\psi^*_{-j}\Theta, \label{eq:spinorgenGP01} \\
		\frac{\partial \boldsymbol{M}^2}{\partial \psi_j^*}&=2j\psi_jM_z+\sum_{s=\pm}\sqrt{(F-sj)(F+sj+1)}\psi_{j+s1}M_s, \label{eq:spinorgenGP02} \\
		\frac{\partial\operatorname{tr}\mathcal{N}^2}{\partial \psi_j^*}&=\left[ 2j^2N_{zz}+(F^2+F-j^2)N_{+-} \right]\psi_j+\sum_{s=\pm}(2j+s1)\sqrt{(F-sj)(F+sj+1)}\psi_{j+s1}N_{zs} \nonumber\\
		&\quad+\sum_{s=\pm}\frac{1}{2}\sqrt{(F-sj)(F+sj+1)(F-sj-1)(F+sj+2)}\psi_{j+s2}N_{ss}, \label{eq:spinorgenGP03}
	\end{align}
	where $ \psi_i $'s with $ |i|>F $ should be all ignored. We can thus write down the GP equation explicitly. Assuming a spatially uniform ground state, the GP equation reduces to $ \mu\psi_j=\frac{\partial h_{\text{int}}}{\partial \psi_j^*} $. The chemical potential is determined if we fix the density. Multiplying the GP equation by $ \psi_j^* $ and taking a summation over $ j $, we obtain 
	\begin{align}
		\mu=\frac{2h_{\text{int}}}{\rho}. \label{eq:spinormu}
	\end{align}
	Here, Euler's theorem for homogeneous functions is used\footnote{When the Hamiltonian contains $n$-body interactions with $ n\ge 3 $, the above result is modified as follows. If the interaction term in the Hamiltonian density is written as $ h_{\text{int}}=\sum_n h_{\text{int}}^{(n)} $, where $h_{\text{int}}^{(n)}$ represents the  $ n $-body interaction such as $ h_{\text{int}}^{(n)}=\sum c_{i_1\dotsm i_nj_1\dotsm j_n}\psi_{i_1}^*\dotsm\psi_{i_n}^*\psi^{}_{j_1}\dotsm\psi^{}_{j_n} $, we obtain $ \mu=\left[\sum_n nh_{\text{int}}^{(n)} \right]/\rho $ by Euler's theorem.}. \\
	\indent In order to derive the Bogoliubov equation, it is convenient to introduce a notation for linearizations of $ \rho, \Theta, M_i, $ and $ N_{ij} $. In the same way as the derivation of the Bogoliubov equation, we set $ (\delta\psi_j,\delta\psi_j^*)=(u_j,v_j) $ after linearization. Then we obtain
	\begin{align}
		\delta\rho&=\sum_{j=-F}^F(\psi_j^*u_j+\psi_jv_j), \label{eq:spinorgenlin01} \\
		\delta M_z&=\sum_{j=-F}^Fj(\psi_j^*u_j+\psi_jv_j), \label{eq:spinorgenlin03} \\
		\delta M_{\pm}&=\sum_{j=-F}^F\sqrt{(F\pm j)(F\mp j+1)}(\psi_{j\mp 1}v_j+\psi_j^*u_{j\mp 1}),  \label{eq:spinorgenlin05} \\
		\delta\Theta&=2\sum_{j=-F}^F(-1)^j\psi_{-j}u_j,\quad \delta\Theta^*=2\sum_{j=-F}^F(-1)^j\psi^*_{-j}v_j. \label{eq:spinorgenlin02}
	\end{align}
	Here, $ \psi_j, u_j, $ and $ v_j $ with $ |j|>F $ should be considered to be zero. We also define $ \delta M_x=(\delta M_++\delta M_-)/2 $ and $ \delta M_y=(M_+-M_-)/(2\mathrm{i}) $. The linearized nematic tensor $ \delta N_{ij} $ can be also written down in the same way, but we omit it. These linearized quantities can be used to \textit{characterize} each mode by what kind of physical quantity is excited. For example, the Bogoliubov phonon originated from the $ U(1) $-gauge symmetry breaking has a finite $ \delta \rho $ and vanishing $ \delta M_i $'s. The spin wave originated from $ SO(3) $-rotation symmetry breaking has no $ \delta \rho $ but finite $ \delta M_i $'s. The gapful modes generally have only fluctuations of higher-rank tensors, e.g.,  $ \delta N_{ij},\ \delta O_{ijk}, $ and so on. These features will be illustrated by the examples in the rest of this section. \\
	\indent Finally, we provide the SSB-originated zero-mode solutions and the WB matrix, which are discussed in Sec.~\ref{sec:formulation} in detail. As mentioned above, the symmetry of the Hamiltonian density of spinor BECs without external magnetic field is $ U(1)\times SO(3) $; if $ \boldsymbol{\psi} $ is a solution of the GP equation,  $ \boldsymbol{\psi}'=\mathrm{e}^{\mathrm{i}(\theta+\alpha F_x+\beta F_y+\gamma F_z)}\boldsymbol{\psi} $ also becomes a solution. (Recall the property (\ref{eq:symsol}).) Differentiating the GP equation by $ \theta,\ \alpha,\ \beta, $ and $ \gamma $, we obtain the following four SSB-originated zero-mode solutions for the Bogoliubov equation:
	\begin{align}
		\begin{pmatrix}\boldsymbol{u} \\ \boldsymbol{v} \end{pmatrix}=\begin{pmatrix} \boldsymbol{\psi} \\ -\boldsymbol{\psi}^* \end{pmatrix},\ \begin{pmatrix} F_z\boldsymbol{\psi} \\ -F_z^*\boldsymbol{\psi}^* \end{pmatrix},\ \begin{pmatrix} F_x\boldsymbol{\psi} \\ -F_x^*\boldsymbol{\psi}^* \end{pmatrix},\ \begin{pmatrix} F_y\boldsymbol{\psi} \\ -F_y^*\boldsymbol{\psi}^* \end{pmatrix}. \label{eq:spinorzeromode}
	\end{align}
	The following choice of the basis is also convenient for discussing a type-II NGM:
	\begin{align}
		\begin{pmatrix}\boldsymbol{u} \\ \boldsymbol{v} \end{pmatrix}=\begin{pmatrix} \boldsymbol{\psi} \\ -\boldsymbol{\psi}^* \end{pmatrix},\ \begin{pmatrix} F_z\boldsymbol{\psi} \\ -F_z^*\boldsymbol{\psi}^* \end{pmatrix},\ \begin{pmatrix} F_+\boldsymbol{\psi} \\ -F_-^*\boldsymbol{\psi}^* \end{pmatrix},\ \begin{pmatrix} F_-\boldsymbol{\psi} \\ -F_+^*\boldsymbol{\psi}^* \end{pmatrix}. \label{eq:spinorzeromode02}
	\end{align}
	Though the above expressions become simpler since $ F_x,\ \mathrm{i}F_y,\ F_z, $ and $ F_\pm $ are real matrices, we keep them as they are, because these expressions remind us of the general formulae [Eq.~(\ref{eq:ssbzero})].
	The WB matrix [Eq.~(\ref{eq:wbmatrix})], or equivalently, the Gram matrix [Eq.~(\ref{eq:introgram})] in the present system is given by
	\begin{align}
		\rho&=\mathrm{i}\boldsymbol{\psi}^\dagger\begin{pmatrix} [F_x,F_x] & [F_x,F_y] & [F_x,F_z] & [F_x,I] \\ [F_y,F_x] & [F_y,F_y] & [F_y,F_z] & [F_y,I] \\ [F_z,F_x] & [F_z,F_y] & [F_z,F_z] & [F_z,I] \\ [I,F_x] & [I,F_y] & [I,F_z] & [I,I] \end{pmatrix}\boldsymbol{\psi} \ = \ \begin{pmatrix} 0&-M_z&M_y& 0 \\ M_z & 0 & -M_x & 0 \\ -M_y & M_x & 0&0 \\ 0&0&0&0  \end{pmatrix}. \label{eq:wbmatspinor}
	\end{align}
	So, the rank of $ \rho $ is given by
	\begin{align}
		n_{\text{II}}=\frac{1}{2}\operatorname{rank}\rho=\begin{cases}1 & (\boldsymbol{M}^2\ne0) \\ 0 & (\boldsymbol{M}^2=0). \end{cases} \label{eq:spinorWB}
	\end{align}
	Thus, in the current case, the criterion for the emergence of the type-II mode is very simple; if the order parameter has a finite magnetization, there is one type-II mode. If it has no magnetization, there is no type-II mode. 
\subsection{Spin-1 BECs}
	The spin-1 BEC model is first investigated by Ref.~\cite{JPSJ.67.1822,Ho:1998zz}. As already mentioned in the previous subsection, the Hamiltonian density for a three-component condensate $ \boldsymbol{\psi}=(\psi_1,\psi_0,\psi_{-1})^T $ is given by Eqs.~(\ref{eq:spinorgen00}) and (\ref{eq:spinorgen01}):
	\begin{align}
		h=\sum_{j=-1}^1 |\nabla\psi_j|^2-\mu\rho+c_0\rho^2+c_1|\Theta|^2.
	\end{align}
	Since the identity (\ref{eq:spinor2bdyid01}) holds, the interaction part can be rewritten as 
	\begin{align}
		c_0\rho^2+c_1|\Theta|^2=\tilde{c}_0\rho^2+\tilde{c}_1\boldsymbol{M}^2,\qquad (\tilde{c}_0,\tilde{c}_1)=(c_0+c_1,-c_1),
	\end{align}
	which may be more familiar. 
	The GP equation is given by [c.f.: Eqs. ~(\ref{eq:spinorgenGP00})-(\ref{eq:spinorgenGP01})]
	\begin{align}
		\mathrm{i}\partial_t\psi_j=-\nabla^2\psi_j-\mu\psi_j+2c_0\rho\psi_j+2c_1\Theta(-1)^j\psi_{-j}^*. \label{eq:spin1gp}
	\end{align}
	The Bogoliubov equation is obtained by linearization of Eq.~(\ref{eq:spin1gp}) and its complex conjugate:
	\begin{align}
		\mathrm{i}\partial_tu_j&=-\nabla^2u_j-\mu u_j+2c_0(\delta\rho \psi_j+\rho u_j)+2c_1(-1)^j(\delta\Theta\psi_{-j}^*+\Theta v_{-j}), \\ 
		\mathrm{i}\partial_tv_j&=\nabla^2v_j+\mu v_j-2c_0(\delta\rho\psi_j^*+\rho v_j)-2c_1(-1)^j(\delta\Theta^*\psi_{-j}+\Theta^*u_{-j}), 
	\end{align}
	where $ \delta\rho,\ \delta\Theta $ and $ \delta\Theta^* $ are defined by Eqs.~(\ref{eq:spinorgenlin01}) and (\ref{eq:spinorgenlin02}). The equation has the SSB-originated zero-mode solutions, Eq.~(\ref{eq:spinorzeromode}) or Eq.~(\ref{eq:spinorzeromode02}). \\ 
	\indent Two ground states appear in this system, depending on the sign of $ c_1 $. The one is the ferromagnetic state for $ c_1>0 $
	\begin{align}
		\boldsymbol{\psi}=(\sqrt{\rho_0},0,0)^T,\quad \boldsymbol{M}^2=\rho_0^2,\quad \mu=2c_0\rho_0, \label{eq:spin1ferro}
	\end{align}
	and the other is the polar state for $ c_1<0 $ 
	\begin{align}
		\boldsymbol{\psi}=(0,\sqrt{\rho_0},0)^T,\quad \boldsymbol{M}^2=0,\quad \mu=2(c_0+c_1)\rho_0. \label{eq:spin1polar}
	\end{align}
\subsubsection{Spin-1 Ferromagnetic Phase}
	\indent Let us first example the excitations in the ferromagnetic phase (\ref{eq:spin1ferro}). Since this phase has a magnetization, we expect from Eq.~(\ref{eq:spinorWB}) that one type-II NGM appears. The phase has a residual continuous $ SO(2) $-symmetry $ \mathrm{e}^{\mathrm{i}\alpha(I-F_z)}\boldsymbol{\psi}=\boldsymbol{\psi} $, so the number of broken symmetry is three and  $ (\boldsymbol{\psi},-\boldsymbol{\psi}^*)^T $ and $ (F_z\boldsymbol{\psi},-F_z^*\boldsymbol{\psi}^*)^T $ in the zero-mode solutions (\ref{eq:spinorzeromode02}) are degenerate. The modes $ (F^+\boldsymbol{\psi},-F_-^*\boldsymbol{\psi}^*)^T $ and $ (F_-\boldsymbol{\psi},-F_+^*\boldsymbol{\psi}^*)^T $ become finite-norm eigenvectors, which become a seed of type-II NGM.\\ 
	\indent The stationary Bogoliubov equation for this phase becomes
	\begin{gather}
		H\begin{pmatrix}\boldsymbol{u} \\ \boldsymbol{v} \end{pmatrix}=\epsilon\begin{pmatrix}\boldsymbol{u} \\ \boldsymbol{v} \end{pmatrix},\quad H=H_0+\sigma k^2, \\
		H_0=2\rho_0\begin{pmatrix} c_0 &&& c_0 && \\ &0&&&& \\ &&2c_1 &&& \\ -c_0 &&& -c_0 && \\ &&&&0 & \\ &&&&&-2c_1  \end{pmatrix}. \label{eq:spin1FH0}
	\end{gather}
	This $ H_0 $ already has the standard form (\ref{eq:UiHU0U}), so we have one type-I, one type-II and one gapful mode. 
	Let us give a notation for the corresponding eigenvectors in accordance with Subsec.~\ref{subsec:H0k0}:
	\begin{align}
		\boldsymbol{y}_1&:=\frac{1}{\sqrt{\rho_0}}\begin{pmatrix}\boldsymbol{\psi} \\ -\boldsymbol{\psi}^* \end{pmatrix}=(1,0,0,-1,0,0)^T, \\ 
		\boldsymbol{z}_1&:=\sigma\boldsymbol{y}_1=(1,0,0,1,0,0)^T,\\
		\boldsymbol{x}_1&:=\frac{1}{\sqrt{2\rho_0}}\begin{pmatrix} F_-\boldsymbol{\psi} \\ -F_+^*\boldsymbol{\psi}^* \end{pmatrix}=(0,1,0,0,0,0)^T, \\ 
		\boldsymbol{w}_1&:=(0,0,1,0,0,0)^T.
	\end{align}
	They satisfy
	\begin{align}
		&H_0\boldsymbol{y}_1=\boldsymbol{0},\ H_0\boldsymbol{z}_1=4c_0\rho_0\boldsymbol{y}_1,\\
		&H_0\boldsymbol{x}_1=H_0\tau\boldsymbol{x}_1^*=0,\\
		&H_0\boldsymbol{w}_1=4c_1\rho_0\boldsymbol{w}_1,\ H_0\tau\boldsymbol{w}_1^*=-4c_1\rho_0\tau\boldsymbol{w}_1^*
	\end{align}
	The block-diagonalizing matrix $ U $ is just an identity matrix: $ U=I_6=\left( \tfrac{\boldsymbol{y}_1+\boldsymbol{z}_1}{2},\boldsymbol{x}_1,\boldsymbol{w}_1,\tfrac{-\boldsymbol{y}_1+\boldsymbol{z}_1}{2},\tau\boldsymbol{x}_1^*,\tau\boldsymbol{w}_1^* \right) $, and $ U^{-1}H_0U $ is given by Eq.~(\ref{eq:spin1FH0}).
	The perturbative expansions of NGMs for finite $ k $ are given by
	\begin{align}
		&\epsilon=\pm2\sqrt{c_0\rho_0}k+\dotsb,\ \boldsymbol{\xi}=\boldsymbol{y}_1\pm\frac{k}{2\sqrt{c_0\rho_0}}\boldsymbol{z}_1+\dotsb, \label{eq:spin1ferro21} \\
		&\epsilon=\pm k^2,\ \boldsymbol{\xi}=\boldsymbol{x}_1 \text{ or } \tau\boldsymbol{x}_1^*, \label{eq:spin1ferro22}
	\end{align}
	These are consistent with the exact eigenvalues of $ H $:
	\begin{align}
		\epsilon=\pm\sqrt{4c_0\rho_0k^2+k^4},\ \pm k^2,\ \pm(4c_1\rho_0+k^2).
	\end{align}
	We note that in this phase  $ F_x\boldsymbol{\psi} $ and $ F_y\boldsymbol{\psi} $ are linearly dependent over $ \mathbb{C} $ (recall the discussion in Subsec.~\ref{sec:bdwbm}). This is the case of Eq.~(\ref{eq:typeiidpndnt}), and the equality of Eq. (\ref{eq:typeiidpndnt0}) holds, i.e., the coefficient of quadratic dispersion becomes $ \epsilon_2=1 $. \\ 
	\indent Let us check what kind of physical quantity is excited in each mode. Using $ \boldsymbol{\psi}=(\sqrt{\rho_0},0,0)^T $, the fluctuations defined in Eqs. (\ref{eq:spinorgenlin01})-(\ref{eq:spinorgenlin05}) are given by
	\begin{align}
		\delta\rho=\delta M_z=\sqrt{\rho_0}(u_1+v_1),\quad \delta M_+=\sqrt{2\rho_0}u_0,\quad \delta M_-=\sqrt{2\rho_0}v_0.
	\end{align}
	The type-I NGM (\ref{eq:spin1ferro21}) has finite $ \delta \rho=\delta M_z $ and no $ \delta M_\pm $. By this excitation, the magnetization per density does not change, because $ \delta(\frac{M_z}{\rho})=\frac{\rho\delta M_z-M_z\delta\rho}{\rho^2}=0 $. Thus it can be interpreted as a sound wave. The type-II NGM (\ref{eq:spin1ferro21}) has finite $ \delta M_x $ and $ \delta M_y $. We can check, however, that the fluctuation of the magnitude of the magnetization, defined by $ \delta\boldsymbol{M}^2=2 \sum_i M_i\delta M_i $, is zero. Since this type-II mode is physically interpreted as a spin precession, the mode only changes the angle of the spin from $ z $-axis, and does not change the total magnitude. The mode $ \boldsymbol{w}_1 $ only has a quadrupolar fluctuation $ \delta N_{ij} $. 
\subsubsection{Spin-1 Polar Phase}
	\indent The stationary Bogoliubov equation for the polar phase (\ref{eq:spin1polar}) is given by
	\begin{gather}
		(H_0+\sigma k^2)\begin{pmatrix}\boldsymbol{u} \\ \boldsymbol{v} \end{pmatrix}=\epsilon\begin{pmatrix}\boldsymbol{u} \\ \boldsymbol{v} \end{pmatrix},\quad H_0=\begin{pmatrix} F & G \\ -G^* & -F^* \end{pmatrix}, \\
		F=2\rho_0\operatorname{diag}(-c_1,c_0+c_1,-c_1),\quad G=2\rho_0\begin{pmatrix} &&-c_1 \\ &c_0+c_1& \\ -c_1&& \end{pmatrix}.
	\end{gather}
	Since the polar phase has no magnetization, the WB matrix (\ref{eq:wbmatspinor}) vanishes and only type-I NGMs appear. Since the polar phase has a $ SO(2) $-symmetry $ \mathrm{e}^{\mathrm{i}\alpha F_z}\boldsymbol{\psi}=\boldsymbol{\psi} $, we obtain three SSB-originated zero modes solutions corresponding to the three broken symmetries: in Eq.~(\ref{eq:spinorzeromode}), the mode $ (F_z\boldsymbol{\psi},-F_z^*\boldsymbol{\psi}^*)^T $ is zero because of this $SO(2)$-symmetry. Since the nonvanishing three modes are $\sigma$-orthogonal to each other, we cannot make a finite-norm eigenvector by their linear combination. Thus, all NGMs are of type-I. Let us introduce the notation for eigenvectors in the same manner with Subsec.~\ref{subsec:H0k0}:
	\begin{align}
		\boldsymbol{y}_1&=\frac{1}{\sqrt{\rho_0}}\begin{pmatrix}\boldsymbol{\psi} \\ -\boldsymbol{\psi}^*\end{pmatrix}=(0,1,0,0,-1,0)^T, \\
		\boldsymbol{y}_2&=\frac{1}{\sqrt{\rho_0}}\begin{pmatrix}F_x\boldsymbol{\psi} \\ -F_x^*\boldsymbol{\psi}^*\end{pmatrix}=\frac{1}{\sqrt{2}}(1,0,1,-1,0,-1)^T,\\
		\boldsymbol{y}_3&=\frac{1}{\sqrt{\rho_0}}\begin{pmatrix}F_y\boldsymbol{\psi} \\ -F_y^*\boldsymbol{\psi}^*\end{pmatrix}=\frac{1}{\sqrt{2}}(-\mathrm{i},0,\mathrm{i},-\mathrm{i},0,\mathrm{i})^T,\\ 
		\boldsymbol{z}_i&=\sigma\boldsymbol{y}_i \quad (i=1,2,3).
	\end{align}
	They satisfy
	\begin{align}
		&H_0\boldsymbol{y}_j=\boldsymbol{0},\quad H_0\boldsymbol{z}_j=2\kappa_j\boldsymbol{y}_j,\\
		&\kappa_1=2(c_0+c_1)\rho_0,\quad \kappa_2=\kappa_3=-2c_1\rho_0.
	\end{align}
	If we define the diagonalizing matrix $ U=\left( \tfrac{\boldsymbol{z}_1+\boldsymbol{y}_1}{2},\tfrac{\boldsymbol{z}_2+\boldsymbol{y}_2}{2},\tfrac{\boldsymbol{z}_3+\boldsymbol{y}_3}{2},\tfrac{\boldsymbol{z}_1-\boldsymbol{y}_1}{2},\tfrac{\boldsymbol{z}_2-\boldsymbol{y}_2}{2},\tfrac{\boldsymbol{z}_3-\boldsymbol{y}_3}{2} \right) $, the standard form of $ H_0 $ is given by
	\begin{align}
		U^{-1}H_0U=2\rho_0\begin{pmatrix} K & K \\ -K & -K  \end{pmatrix},\quad K=2\rho_0\operatorname{diag}(c_0+c_1,-c_1,-c_1).
	\end{align}
	The perturbative expansions of eigenvectors are given by [Eqs.~(\ref{eq:typeIvalue}) and (\ref{eq:typeIvector})]: 
	\begin{align}
		\epsilon = \pm \sqrt{2\kappa_j}k+O(k^2),\quad \boldsymbol{\xi}=\boldsymbol{y}_j\pm\frac{k}{\sqrt{2\kappa_j}}\boldsymbol{z}_j+O(k^2),\quad j=1,2,3,
	\end{align}
	which is consistent with the exact eigenvalues of $ H_0+\sigma k^2 $:
	\begin{align}
		\epsilon=&\pm\sqrt{4(c_0+c_1)\rho_0k^2+k^4}, \\
		& \pm\sqrt{-4c_1\rho_0k^2+k^4} \text{ (doubly degenerate)}.
	\end{align}
	\indent Let us check the fluctuation of physical quantities for each mode. Using $ \boldsymbol{\psi}=(0,\sqrt{\rho_0},0)^T $, the fluctuations defined in Eqs. (\ref{eq:spinorgenlin01})-(\ref{eq:spinorgenlin05}) become
	\begin{align}
		\delta\rho=\sqrt{\rho_0}(u_0+v_0),\quad \delta M_z=0,\quad \delta M_+=\sqrt{2\rho_0}(u_{-1}+v_1),\quad \delta M_-=\sqrt{2\rho_0}(v_{-1}+u_1).
	\end{align}
	The mode  $ \boldsymbol{y}_1+\frac{k}{\sqrt{2\kappa_1}}\boldsymbol{z}_1 $ has only $ \delta\rho $, so it is interpreted as a sound wave. The mode $ \boldsymbol{y}_2+\frac{k}{\sqrt{2\kappa_2}}\boldsymbol{z}_2 $ has finite $ \delta M_x $ and others are zero: $ \delta\rho=\delta M_y=0 $. Thus, it represents a spin wave in the $x$-direction. Similarly, the mode $ \boldsymbol{y}_3+\frac{k}{\sqrt{2\kappa_3}}\boldsymbol{z}_3 $ is a spin wave in the $ y $-direction. As mentioned in the Introduction, these fluctuations of magnetizations are ignored if we describe all phenomena in a coset space, because the phase is fixed to be polar.
\subsection{Spin-2 BECs}\label{subsec:spin2}
	The model of the spin-2 BEC was first introduced in Ref.~\cite{PhysRevA.61.033607} and the phase diagram was given. The Hamiltonian density for a five-component condensate  $ \boldsymbol{\psi}=(\psi_2,\dots,\psi_{-2})^T $  is given by Eqs.~(\ref{eq:spinorgen00}) and (\ref{eq:spinorgen02}):
	\begin{align}
		h=\sum_{j=-2}^2 |\nabla\psi_j|^2-\mu\rho+c_0\rho^2+c_1\boldsymbol{M}^2+c_2|\Theta|^2,
	\end{align}
	and the GP equation is given by [c.f.: Eqs. ~(\ref{eq:spinorgenGP00})-(\ref{eq:spinorgenGP01})]
	\begin{align}
		\mathrm{i}\partial_t\psi_j&=-\nabla^2\psi_j-\mu\psi_j+2c_0\rho\psi_j+2c_2\Theta(-1)^j\psi_{-j}^* \nonumber \\
		&\quad +c_1\left[ 2j\psi_jM_z+\sqrt{(2+j)(3-j)}\psi_{j-1}M_-+\sqrt{(2-j)(3+j)}\psi_{j+1}M_+ \right], \label{eq:spin2GP}
	\end{align}
	where $ \psi_j $ with $ |j|>2 $ are interpreted as zero. The Bogoliubov equation is obtained by linearization of the GP equation and its complex conjugate: 
	\begin{align}
		\mathrm{i}\partial_tu_j&=-\nabla^2\delta u_j-\mu u_j+2c_0(\delta\rho\psi_j+\rho u_j)+2c_2(-1)^j(\delta\Theta\psi_{-j}^*+\Theta v_{-j}) \nonumber \\
		&\quad +c_1\left[ 2j(u_jM_z+\psi_j\delta M_z)+\sqrt{(2+j)(3-j)}(u_{j-1}M_-+\psi_{j-1}\delta M_-)+\sqrt{(2-j)(3+j)}(u_{j+1}M_++\psi_{j+1}\delta M_+) \right], \label{eq:spin-2bogo} \\
		\mathrm{i}\partial_tv_j&=\nabla^2v_j+\mu v_j-2c_0(\delta\rho\psi_j^*+\rho v_j)-2c_2(-1)^j(\delta\Theta^*\psi_{-j}+\Theta^*u_{-j}) \nonumber \\
		&\quad -c_1\left[ 2j(v_jM_z+\psi_j^*\delta M_z)+\sqrt{(2+j)(3-j)}(v_{j-1}M_++\psi_{j-1}^*\delta M_+)+\sqrt{(2-j)(3+j)}(v_{j+1}M_-+\psi_{j+1}^*\delta M_-) \right], 
	\end{align}
	where $ \delta\rho,\ \delta\Theta,\ \delta\Theta^* $, and $ \delta M_i $'s are defined by Eqs.~(\ref{eq:spinorgenlin01})-(\ref{eq:spinorgenlin02}) with setting $ F=2 $.  	This Bogoliubov equation has the SSB-originated zero-mode solutions Eq.~(\ref{eq:spinorzeromode}) or equivalently Eq.~(\ref{eq:spinorzeromode02}). \\ 
	\indent The spin-2 BEC admits three kinds of ground states \cite{PhysRevA.61.033607}. When $ c_1<0  $ and $  4c_1<c_2 $, the ground state is ferromagnetic:
	\begin{align}
		&\boldsymbol{\psi}=(\sqrt{\rho_0},0,0,0,0)^T,\quad \boldsymbol{M}^2=4\rho_0^2,\quad |\Theta|^2=0,\quad \mu=2(c_0+4c_1)\rho_0. \label{eq:spin2ferro}
	\end{align}
	When $ c_1>0  $ and $  c_2>0 $, the ground state is cyclic:
	\begin{align}
		&\boldsymbol{\psi}=\sqrt{\rho_0}(\tfrac{\mathrm{i}}{2},0,\tfrac{1}{\sqrt{2}},0,\tfrac{\mathrm{i}}{2}),\quad \boldsymbol{M}^2=|\Theta|^2=0,\quad \mu=2c_0\rho_0. \label{eq:spin2cyclic}
	\end{align}
	When $ c_2<0  $ and $  4c_1>c_2 $, the ground state is nematic:
	\begin{align}
		&\boldsymbol{\psi}=\sqrt{\rho_0}(\tfrac{\sin\eta}{\sqrt{2}},0,\cos\eta,0,\tfrac{\sin\eta}{\sqrt{2}}),\quad \boldsymbol{M}^2=0,\quad |\Theta|^2=\rho_0^2,\quad \mu=2(c_0+c_2)\rho_0. \label{eq:spin2nematic}
	\end{align}
	Here we note that the value of $ \mu $ is easily determined by Eq.~(\ref{eq:spinormu}). The parameter $ \eta $ in the nematic phase is real and it shows a large continuous degeneracy of this phase. The states with different $ \eta $ are not equivalent under $ U(1)\times SO(3) $ transformation, and therefore this phase contains distinct spinor states characterized by one parameter $ \eta $. It is known that this degeneracy is resolved when the quantum correction is added, and either $ \eta=0 $ (uniaxial nematic) or $ \frac{\pi}{2} $ (biaxial nematic) is favored \cite{PhysRevLett.98.160408,PhysRevLett.98.190404,PhysRevA.81.063632}.  However, at a classical-field level, these states are completely degenerate. 
\subsubsection{Spin-2 Ferromagnetic Phase}
	\indent Let us first consider the NGMs in ferromagnetic phase (\ref{eq:spin2ferro}). The stationary Bogoliubov equation for $ (\boldsymbol{u},\boldsymbol{v})^T=(u_2,\dots,u_{-2},v_2,\dots,v_{-2})^T $  becomes
	\begin{align}
		&(H_0+\sigma k^2)\begin{pmatrix}\boldsymbol{u} \\ \boldsymbol{v}\end{pmatrix}=\epsilon\begin{pmatrix}\boldsymbol{u} \\ \boldsymbol{v}\end{pmatrix},\quad H_0=\begin{pmatrix}F & G \\ -G^* & -F^*\end{pmatrix}, \label{eq:spin2bogo01} \\
		F&=2\rho_0 \operatorname{diag}(c_0+4c_1,0,-4c_1,-6c_1,-8c_1+2c_2), \\
		G&=2\rho_0 \operatorname{diag}(c_0+4c_1,0,0,0,0).
	\end{align}
	As with the spin-1 ferromagnetic case, this $ H_0 $ already has the standard form (\ref{eq:UiHU0U}), and we obtain one type-I, one type-II, and three gapful modes. We do not discuss this phase any more in detail, because it is almost the same with the spin-1 ferromagnetic phase. Since the ferromagnetic state has a residual $ SO(2) $-symmetry $ \boldsymbol{\psi}\rightarrow\mathrm{e}^{\mathrm{i}\alpha(F_z-2I)}\boldsymbol{\psi} $, there are three independent SSB-originated zero-mode solutions in Eq.~(\ref{eq:spinorzeromode02}). The exact eigenvalues of $ H=H_0+\sigma k^2 $ are given by
	\begin{align}
		\epsilon=&\pm\sqrt{4(c_0+4c_1)\rho_0k^2+k^4},\ \pm k^2, \nonumber \\
		& \pm(k^2-8c_1\rho_0),\ \pm(k^2-12c_1\rho_0),\ \pm(k^2-16c_1\rho_0+4c_2\rho_0).
	\end{align}
	The first two modes are gapless NGMs of type-I and type-II. The latter three are gapful. 
\subsubsection{Spin-2 Cyclic Phase}
	\indent Next let us see the cyclic phase (\ref{eq:spin2cyclic}). The Bogoliubov equation is given by
	\begin{align}
		(H_0+\sigma k^2)\begin{pmatrix}\boldsymbol{u} \\ \boldsymbol{v}\end{pmatrix}=\epsilon\begin{pmatrix}\boldsymbol{u} \\ \boldsymbol{v}\end{pmatrix},\quad H_0=\begin{pmatrix}F & G \\ -G^* & -F^*\end{pmatrix}
	\end{align}
	with
	\begin{align}
		F&=2\rho_0\begin{pmatrix}\frac{c_0+4c_1+2c_2}{4} & 0 & \frac{\mathrm{i}(c_0-2c_2)}{2\sqrt{2}}&0&\frac{c_0-4c_1+2c_2}{4} \\ 0&2c_1&0&0&0 \\ \frac{-\mathrm{i}(c_0-2c_2)}{2\sqrt{2}} &0&\frac{c_0+2c_2}{2}&0&\frac{-\mathrm{i}(c_0-2c_2)}{2\sqrt{2}}\\ 0&0&0&2c_1&0 \\ \frac{c_0-4c_1+2c_2}{4} & 0 & \frac{\mathrm{i}(c_0-2c_2)}{2\sqrt{2}}&0&\frac{c_0+4c_1+2c_2}{4} \end{pmatrix}, \\
		G&=2\rho_0\begin{pmatrix} \frac{-c_0-4c_1}{4}&0&\frac{\mathrm{i}c_0}{2\sqrt{2}}&0&\frac{-c_0+4c_1}{4} \\ 0&\mathrm{i}\sqrt{3}c_1 & 0 & c_1 & 0 \\ \frac{\mathrm{i}c_0}{2\sqrt{2}} &0& \frac{c_0}{2} & 0 & \frac{\mathrm{i}c_0}{2\sqrt{2}} \\  0 & c_1 & 0 & \mathrm{i}\sqrt{3}c_1 & 0 \\ \frac{-c_0+4c_1}{4}&0&\frac{\mathrm{i}c_0}{2\sqrt{2}}&0&\frac{-c_0-4c_1}{4} \end{pmatrix}.
	\end{align}
	The equation is decoupled for $ (u_2,u_0,u_{-2},v_2,v_0,v_{-2}) $ and $ (u_1,u_{-1},v_1,v_{-1}) $. \\ 
	\indent Since the cyclic phase breaks all four symmetries of $ U(1)\times SO(3) $, we have four linearly independent SSB-originated zero mode solutions (\ref{eq:spinorzeromode}). Furthermore, since the cyclic phase has no magnetization, the WB matrix (\ref{eq:wbmatspinor}) vanishes. Thus, we obtain four type-I NGMs. Let us introduce a notation for eigenvectors in the same way with Subsec.~\ref{subsec:H0k0}: 
	\begin{align}
		\boldsymbol{y}_1&:=\frac{1}{\sqrt{\rho_0}}\begin{pmatrix}\boldsymbol{\psi} \\ -\boldsymbol{\psi}^* \end{pmatrix}=(\tfrac{\mathrm{i}}{2},0,\tfrac{1}{\sqrt{2}},0,\tfrac{\mathrm{i}}{2},\tfrac{\mathrm{i}}{2},0,\tfrac{-1}{\sqrt{2}},0,\tfrac{\mathrm{i}}{2})^T,\\
		\boldsymbol{y}_2&:=\frac{1}{\sqrt{2\rho_0}}\begin{pmatrix}F_z\boldsymbol{\psi} \\ -F_z^*\boldsymbol{\psi}^* \end{pmatrix}=\frac{1}{\sqrt{2}}(\mathrm{i},0,0,0,-\mathrm{i},\mathrm{i},0,0,0,-\mathrm{i})^T, \\
		\boldsymbol{y}_3&:=\frac{1}{\sqrt{2\rho_0}}\begin{pmatrix}F_x\boldsymbol{\psi} \\ -F_x^*\boldsymbol{\psi}^* \end{pmatrix}=\frac{1}{\sqrt{2}}(0,\mathrm{e}^{\mathrm{i}\pi/6},0,\mathrm{e}^{\mathrm{i}\pi/6},0,0,-\mathrm{e}^{-\mathrm{i}\pi/6},0,-\mathrm{e}^{-\mathrm{i}\pi/6},0)^T, \\
		\boldsymbol{y}_4&:=\frac{1}{\sqrt{2\rho_0}}\begin{pmatrix}F_y\boldsymbol{\psi} \\ -F_y^*\boldsymbol{\psi}^* \end{pmatrix}=\frac{1}{\sqrt{2}}(0,-\mathrm{e}^{\mathrm{i}\pi/3},0,\mathrm{e}^{\mathrm{i}\pi/3},0,0,\mathrm{e}^{-\mathrm{i}\pi/3},0,-\mathrm{e}^{-\mathrm{i}\pi/3},0)^T, \\
		\boldsymbol{z}_i&:=\sigma \boldsymbol{y}_i \quad (i=1,2,3,4).
	\end{align}
	The last one mode describes a gapful mode, given by
	\begin{align}
		\boldsymbol{w}_1&:=\begin{pmatrix}\boldsymbol{\psi}^* \\ \boldsymbol{0} \end{pmatrix}=(\tfrac{-\mathrm{i}}{2},0,\tfrac{1}{\sqrt{2}},0,\tfrac{-\mathrm{i}}{2},0,0,0,0,0)^T.
	\end{align}
	They satisfy
	\begin{align}
		&H_0\boldsymbol{y}_i=0,\quad H_0\boldsymbol{z}_i=2\kappa_i\boldsymbol{y}_i, \quad (i=1,2,3,4) \\
		&\kappa_1=2c_0\rho_0,\ \kappa_2=\kappa_3=\kappa_4=4c_1\rho_0, \\
		&H_0\boldsymbol{w}_1=4c_2\rho_0\boldsymbol{w}_1,\quad H_0\tau\boldsymbol{w}_1^*=-4c_2\rho_0\tau\boldsymbol{w}_1^*
	\end{align}
	If we define the block-diagonalizing matrix
	\begin{align}
		U=(\tfrac{\boldsymbol{y}_1+\boldsymbol{z}_1}{2},\dots,\tfrac{\boldsymbol{y}_4+\boldsymbol{z}_4}{2},\boldsymbol{w}_1,\tfrac{-\boldsymbol{y}_1+\boldsymbol{z}_1}{2},\dots,\tfrac{-\boldsymbol{y}_4+\boldsymbol{z}_4}{2},\tau\boldsymbol{w}_1^*),
	\end{align}
	then we obtain the standard form  of $ H_0 $ (\ref{eq:UiHU0U}):
	\begin{align}
		&U^{-1}H_0U=\begin{pmatrix} \tilde{F} & \tilde{G} \\ -\tilde{G}^* & -\tilde{F}^* \end{pmatrix},\\
		&\tilde{F}=2\rho_0\operatorname{diag}(c_0,2c_1,2c_1,2c_1,2c_2),\quad \tilde{G}=2\rho_0\operatorname{diag}(c_0,2c_1,2c_1,2c_1,0).
	\end{align}
	This standard form clearly shows that there are four type-I and one gapful modes. The perturbative expansions for $ \boldsymbol{y}_i $'s are given by the general formulae (\ref{eq:typeIvalue}) and (\ref{eq:typeIvector}). The result is consistent with the exact eigenvalues of $ H=H_0+\sigma k^2 $:
	\begin{align}
		\epsilon=&\pm\sqrt{4c_0\rho_0k^2+k^4},\\
		&\pm\sqrt{8c_1\rho_0k^2+k^4} \quad(\text{triply degenerate}), \\
		& \pm(k^2+4c_2\rho_0).
	\end{align}
	The fluctuation of physical quantities (\ref{eq:spinorgenlin01})-(\ref{eq:spinorgenlin05}) are
	\begin{align}
		&\delta\rho=\sqrt{\rho_0}\left[ \frac{\mathrm{i}(v_2+v_{-2}-u_2-u_{-2})}{2}+\frac{u_0+v_0}{\sqrt{2}} \right],\quad \delta M_z=\mathrm{i}\sqrt{\rho_0}(v_2-v_{-2}-u_2+u_{-2}), \\
		&\delta M_+=\sqrt{\rho_0}\left[ \mathrm{i}(v_{-1}-u_1)+\sqrt{3}(v_1+u_{-1}) \right],\quad \delta M_-=\sqrt{\rho_0}\left[ \mathrm{i}(v_1-u_{-1})+\sqrt{3}(u_1+v_{-1}) \right].
	\end{align}
	We can check that the type-I mode  $ \boldsymbol{y}_1+\frac{k}{\sqrt{2\kappa_1}}\boldsymbol{z}_1 $ has finite $ \delta\rho $, so it represents a sound wave. The mode $ \boldsymbol{y}_i+\frac{k}{\sqrt{2\kappa_i}}\boldsymbol{z}_i $ with $ i=2,3, $ and $ 4 $ has finite $ \delta M_z,\ \delta M_x, $ and $ \delta M_y $, respectively. Thus, they represent a spin wave in the $z$-, $x$-, and $y$-direction, respectively. The mode $ \boldsymbol{w}_1 $ has neither $ \delta\rho $ nor $ \delta M_i $'s, so it is a fluctuation of a higher-rank tensor. 
\subsubsection{Spin-2 Nematic Phase}\label{sec:spin2nematic}
	\indent Finally, let us consider the nematic phase (\ref{eq:spin2nematic}). This phase is well-known for possessing quasi-NGMs \cite{PhysRevLett.105.230406,PhysRevA.81.063632}. When $ \boldsymbol{\psi} $ is given by Eq.~(\ref{eq:spin2nematic}), the Bogoliubov equation becomes
	\begin{align}
		(H_0+\sigma k^2)\begin{pmatrix}\boldsymbol{u} \\ \boldsymbol{v}\end{pmatrix}=\epsilon\begin{pmatrix}\boldsymbol{u} \\ \boldsymbol{v}\end{pmatrix},\quad  H_0=\begin{pmatrix}F & G \\ -G^* & -F^*\end{pmatrix},
	\end{align}
	where $ F $ and $ G $ are given by 
	\begin{align}
		F&=2\rho_0\begin{pmatrix} \frac{(c_0+4c_1)\tilde{s}^2-2c_2\tilde{c}^2}{2} & 0 & \frac{(c_0+2c_2)\tilde{c}\tilde{s}}{\sqrt{2}} & 0 & \frac{(c_0-4c_1+2c_2)\tilde{s}^2}{2} \\ 0&c_1(1+2\tilde{c}^2)-c_2 & 0 & 2\sqrt{3}c_1\tilde{c}\tilde{s} & 0 \\ \frac{(c_0+2c_2)\tilde{c}\tilde{s}}{\sqrt{2}} & 0 & (c_0+2c_2)\tilde{c}^2-c_2 & 0 & \frac{(c_0+2c_2)\tilde{c}\tilde{s}}{\sqrt{2}} \\ 0 & 2\sqrt{3}c_1\tilde{c}\tilde{s} & 0 & c_1(1+2\tilde{c}^2)-c_2 & 0 & \\ \frac{(c_0-4c_1+2c_2)\tilde{s}^2}{2} & 0 & \frac{(c_0+2c_2)\tilde{c}\tilde{s}}{\sqrt{2}} & 0 & \frac{(c_0+4c_1)\tilde{s}^2-2c_2\tilde{c}^2}{2} \end{pmatrix}, \\
		G&=2\rho_0\begin{pmatrix} \frac{(c_0+4c_1)\tilde{s}^2}{2} & 0 & \frac{c_0\tilde{c}\tilde{s}}{\sqrt{2}} & 0 & \frac{(c_0-4c_1)\tilde{s}^2+2c_2}{2} \\ 0 & 2\sqrt{3}c_1\tilde{c}\tilde{s} & 0 & c_1(2\tilde{c}^2+1)-c_2 & 0 \\ \frac{c_0\tilde{c}\tilde{s}}{\sqrt{2}} & 0 & c_0\tilde{c}^2+c_2 & 0 & \frac{c_0\tilde{c}\tilde{s}}{\sqrt{2}} \\ 0 & c_1(2\tilde{c}^2+1)-c_2 & 0 & 2\sqrt{3}c_1\tilde{c}\tilde{s} & 0 \\ \frac{(c_0-4c_1)\tilde{s}^2+2c_2}{2} & 0 & \frac{c_0\tilde{c}\tilde{s}}{\sqrt{2}} & 0 & \frac{(c_0+4c_1)\tilde{s}^2}{2} \end{pmatrix},
	\end{align}
	with $ \tilde{c}=\cos\eta,\ \tilde{s}=\sin\eta $. The equation is decoupled for $ (u_2,u_0,u_{-2},v_2,v_0,v_{-2}) $ and $ (u_1,u_{-1},v_1,v_{-1}) $. We can check that $ H_0^2=0 $, and therefore $ H_0 $ has only zero eigenvalue, and Theorem~C.\ref{prp:K2} can be applied. \\ 
	\indent Before solving the above Bogoliubov equation directly, we first clarify the $ U(1)\times SO(5) $-symmetric nature of the set of nematic states. The nematic phase with an arbitrary angle can be generally written as
	\begin{align}
		\boldsymbol{\psi}=\mathrm{e}^{\mathrm{i}\theta}(\psi_2,\psi_1,r_0,-\psi_1^*,\psi_2^*)^T,\quad \theta,r_0\in\mathbb{R},\ \psi_1,\psi_2\in\mathbb{C}. \label{eq:spin2nematic0}
	\end{align}
	In other words, the above state can be always transformed into the form of Eq.~(\ref{eq:spin2nematic}) by phase multiplication and rotation. To see the  $ SO(5) $-symmetry, let us consider the parametrization
	\begin{align}
		\psi_0=a_1,\quad \psi_{\pm 2}=\frac{a_2\pm\mathrm{i}a_3}{\sqrt{2}},\quad \psi_{\pm1}=\frac{\pm a_4+\mathrm{i}a_5}{\sqrt{2}},
	\end{align}
	or equivalently, 
	\begin{align}
		\boldsymbol{\psi}=U_0\boldsymbol{a},\quad U_0:=\begin{pmatrix} 0& \frac{1}{\sqrt{2}} & \frac{\mathrm{i}}{\sqrt{2}} &0&0 \\ 0&0&0&\frac{1}{\sqrt{2}}&\frac{\mathrm{i}}{\sqrt{2}} \\ 1 &0&0&0&0 \\ 0&0&0&\frac{-1}{\sqrt{2}}&\frac{\mathrm{i}}{\sqrt{2}} \\ 0&\frac{1}{\sqrt{2}}&\frac{-\mathrm{i}}{\sqrt{2}}&0&0 \end{pmatrix}
	\end{align}
	with writing $ \boldsymbol{a}=(a_1,\dots,a_5)^T $. Then, the density and singlet pair amplitude are written as
	\begin{align}
		\rho=\boldsymbol{a}^*\cdot\boldsymbol{a},\quad \Theta=\boldsymbol{a}\cdot\boldsymbol{a}.
	\end{align}
	Clearly these two scalars are invariant under real orthogonal transformation $ \boldsymbol{a}'=R\boldsymbol{a} $ with a $ 5\times 5 $ real orthogonal matrix $ R $. On the other hand, the magnetization vector
	\begin{align}
		M_z&=2\mathrm{i}(a_2^*a_3-a_3^*a_2)+\mathrm{i}(a_4^*a_5-a_5^*a_4),\\
		M_+=M_-^*&=\mathrm{i}(a_2^*a_5-a_5^*a_2+a_4^*a_3-a_3^*a_4)+(a_2^*a_4-a_4^*a_2+a_3^*a_5-a_5^*a_3)+\mathrm{i}\sqrt{3}(a_1^*a_5-a_5^*a_1)+\sqrt{3}(a_4^*a_1-a_1^*a_4)
	\end{align}
	does not have such invariance for general $ \boldsymbol{a} $. However, as we see below, if the state is nematic, the magnetization vanishes and it also becomes invariant. \\
	\indent In terms of $ \boldsymbol{a} $, the nematic state (\ref{eq:spin2nematic0}) can be represented as a real vector up to overall factor:
	\begin{align}
		\boldsymbol{a}=\mathrm{e}^{\mathrm{i}\theta}(r_1,r_2,r_3,r_4,r_5)^T,\quad r_i,\theta\in\mathbb{R}. \label{eq:realvectora}
	\end{align}
	In particular, the state Eq.~(\ref{eq:spin2nematic}) can be written as
	\begin{align}
		\boldsymbol{a}=\sqrt{\rho_0}(\cos\eta,\sin\eta,0,0,0)^T. \label{eq:realvectorb}
	\end{align}
	We can easily check that if $ \boldsymbol{a} $ is given by Eq.~(\ref{eq:realvectora}) or (\ref{eq:realvectorb}) $, M_z,M_\pm $ vanishes. Moreover, this vanishing property is preserved under the $ 5\times 5 $ real orthogonal transformation $ \boldsymbol{a}'=R\boldsymbol{a} $. Thus, if $ \boldsymbol{a} $ has the form (\ref{eq:realvectora}), the following holds:
	\begin{align}
		&\boldsymbol{a} \text{ is a solution of the GP equation.} \nonumber \\
		\leftrightarrow \quad & \boldsymbol{a}'=\mathrm{e}^{\mathrm{i}\varphi}R\boldsymbol{a} \text{ is also a solution.}
	\end{align}
	Or, in terms of $ \boldsymbol{\psi} $, if $ \boldsymbol{\psi} $ is a nematic state [Eq.~(\ref{eq:spin2nematic0})],
	\begin{align}
		&\boldsymbol{\psi} \text{ is a solution of the GP equation [Eq.~(\ref{eq:spin2GP})].} \nonumber \\
		\leftrightarrow\quad & \boldsymbol{\psi}'=\mathrm{e}^{\mathrm{i}\varphi}U_0RU_0^{-1}\boldsymbol{\psi} \text{ is also a solution.} \label{eq:nematicqsym}
	\end{align}
	Now let us recall the discussion in Subsec.~\ref{sec:symzero}. Even if the Hamiltonian density does not have a group symmetry Eq.~(\ref{eq:symham}), as far as a set of solutions satisfy the property Eq.~(\ref{eq:symsol}), we can derive the corresponding zero-mode solutions  $ (\boldsymbol{u},\boldsymbol{v})=(U_{\varphi}\boldsymbol{\psi},U_{\varphi}^*\boldsymbol{\psi}^*) $. In the present case, Eq.~(\ref{eq:nematicqsym}) suggests that the set of solutions has a  $ U(1)\times SO(5) $-symmetry, though the symmetry of Hamiltonian is $ U(1)\times SO(3) $. 
	Since $ SO(5) $ is generated by ten operators $ T_{ab} $, where $ T_{ab} \ (1\le a<b\le 5) $ is a matrix such that  $ (a,b) $-component is $ -\mathrm{i} $ and $ (b,a) $-component is $ \mathrm{i} $ and all other components are zero, we obtain eleven SSB-originated zero-mode solutions corresponding to the infinitesimal transformations of $ U(1)\times SO(5) $:
	\begin{align}
		\begin{pmatrix} \boldsymbol{u} \\ \boldsymbol{v} \end{pmatrix}=\begin{pmatrix} \boldsymbol{\psi} \\ -\boldsymbol{\psi}^* \end{pmatrix},\quad \begin{pmatrix} U_0T_{ab}U_0^{-1}\boldsymbol{\psi} \\ -U_0^*T_{ab}^*(U_0^{-1})^*\boldsymbol{\psi}^* \end{pmatrix} \quad(1\le a<b\le 5). \label{eq:elevenSO5}
	\end{align}
	Note, however, that we only obtain five linearly independent solutions from Eq.~(\ref{eq:elevenSO5}), because any real vector $ \boldsymbol{a} $ has a six-dimensional $ SO(4) $-symmetry.  In particular, if we use  $ \boldsymbol{a} $ of Eq.~(\ref{eq:realvectorb}), the null space $ W $ defined by Eq. (\ref{eq:nullspace}) is given by 
	\begin{align}
		W=&\operatorname{span}\{ \sin\eta T_{13}-\cos\eta T_{23},\, \sin\eta T_{14}-\cos\eta T_{24},\,\sin\eta T_{15}-\cos\eta T_{25},\, T_{34},\,T_{35},\,T_{45} \},
	\end{align} 
	which describes the unbroken $ SO(4) $ algebra. Note that the above eleven solutions (\ref{eq:elevenSO5}) also include the conventional SSB-originated zero mode solutions originated from the ordinary $ U(1)\times SO(3) $-symmetry. The relation between generators of $ SO(3) $ and those of $ SO(5) $ is as follows:
	\begin{align}
		F_z&=-U_0(2T_{23}+T_{45})U_0^{-1}, \label{eq:FztoT}\\
		F_x&=U_0(-\sqrt{3}T_{15}-T_{25}+T_{34})U_0^{-1},\\
		F_y&=U_0(-\sqrt{3}T_{14}+T_{24}+T_{35})U_0^{-1}. \label{eq:FytoT}
	\end{align}
	Now let us introduce the notation for five zero-mode solutions $ \boldsymbol{y}_1,\dots,\boldsymbol{y}_5 $ in the same way with the standard form of Subsec.~\ref{subsec:H0k0}. Let $ \boldsymbol{\psi} $ be Eq.~(\ref{eq:spin2nematic}), and
	\begin{align}
		\boldsymbol{y}_1&=\frac{1}{\sqrt{\rho_0}}\begin{pmatrix}\boldsymbol{\psi} \\ -\boldsymbol{\psi}^* \end{pmatrix}=(\tfrac{\sin\eta}{\sqrt{2}},0,\cos\eta,0,\tfrac{\sin\eta}{\sqrt{2}},-\tfrac{\sin\eta}{\sqrt{2}},0,-\cos\eta,0,-\tfrac{\sin\eta}{\sqrt{2}})^T, \\
		\boldsymbol{y}_2&=\frac{1}{\sqrt{\rho_0}}\begin{pmatrix}U_0T_{12}U_0^{-1}\boldsymbol{\psi} \\ -U_0^*T_{12}^*(U_0^{-1})^*\boldsymbol{\psi}^* \end{pmatrix}=(\tfrac{\mathrm{i}\cos\eta}{\sqrt{2}},0,-\mathrm{i}\sin\eta,0,\tfrac{\mathrm{i}\cos\eta}{\sqrt{2}},\tfrac{\mathrm{i}\cos\eta}{\sqrt{2}},0,-\mathrm{i}\sin\eta,0,\tfrac{\mathrm{i}\cos\eta}{\sqrt{2}})^T,\\
		\boldsymbol{y}_3&=\frac{1}{\sqrt{\rho_0}}\begin{pmatrix}U_0(T_{13}\cos\eta+T_{23}\sin\eta)U_0^{-1}\boldsymbol{\psi} \\ -U_0^*(T_{13}^*\cos\eta+T_{23}^*\sin\eta)(U_0^{-1})^*\boldsymbol{\psi}^* \end{pmatrix}=(\tfrac{-1}{\sqrt{2}},0,0,0,\tfrac{1}{\sqrt{2}},\tfrac{1}{\sqrt{2}},0,0,0,\tfrac{-1}{\sqrt{2}})^T, \\
		\boldsymbol{y}_4&=\frac{1}{\sqrt{\rho_0}}\begin{pmatrix}U_0(T_{14}\cos\eta+T_{24}\sin\eta)U_0^{-1}\boldsymbol{\psi} \\ -U_0^*(T_{14}^*\cos\eta+T_{24}^*\sin\eta)(U_0^{-1})^*\boldsymbol{\psi}^* \end{pmatrix}=(0,\tfrac{\mathrm{i}}{\sqrt{2}},0,\tfrac{-\mathrm{i}}{\sqrt{2}},0,0,\tfrac{\mathrm{i}}{\sqrt{2}},0,\tfrac{-\mathrm{i}}{\sqrt{2}},0)^T, \\
		\boldsymbol{y}_5&=\frac{1}{\sqrt{\rho_0}}\begin{pmatrix}U_0(T_{15}\cos\eta+T_{25}\sin\eta)U_0^{-1}\boldsymbol{\psi} \\ -U_0^*(T_{15}^*\cos\eta+T_{25}^*\sin\eta)(U_0^{-1})^*\boldsymbol{\psi}^* \end{pmatrix}=(0,\tfrac{-1}{\sqrt{2}},0,\tfrac{-1}{\sqrt{2}},0,0,\tfrac{1}{\sqrt{2}},0,\tfrac{1}{\sqrt{2}},0)^T.
	\end{align}
	We also define $ \boldsymbol{z}_i=\sigma\boldsymbol{y}_i $ for $ i=1,\dots,5 $. Then, they satisfy
	\begin{align}
		&H_0\boldsymbol{y}_i=\boldsymbol{0},\quad H_0\boldsymbol{z}_i=2\kappa_i\boldsymbol{y}_i,\\
		&\kappa_1=2(c_0+c_2)\rho_0,\quad \kappa_2=-2c_2\rho_0,\\
		&\kappa_3=2(4c_1\sin^2\eta-c_2)\rho_0,\\
		&\kappa_4=2(4c_1\sin^2(\eta-\tfrac{\pi}{3})-c_2)\rho_0, \\ 
		&\kappa_5=2(4c_1\sin^2(\eta+\tfrac{\pi}{3})-c_2)\rho_0.
	\end{align}
	All $ \boldsymbol{y}_i $'s are $ \sigma $-orthogonal to each other, and hence only type-I (quasi-)NGMs emerge. Because $ H_0^2=0 $, the eigenvalues and eigenvectors of $ \sigma H_0 $ are given by $ 2\kappa_i $'s and $ \boldsymbol{z}_i $'s by Theorem~C.\ref{prp:K2}. This provides an easy way to determine the values of $ \kappa_i $'s and the eigenvectors $ \boldsymbol{y}_i $'s satisfying the orthogonal relations (\ref{eq:orthozero119}). \\
	\indent If we define the block-diagonalizing matrix by
	\begin{align}
		U=\left( \tfrac{\boldsymbol{y}_1+\boldsymbol{z}_1}{2},\dots,\tfrac{\boldsymbol{y}_5+\boldsymbol{z}_5}{2},\tfrac{-\boldsymbol{y}_1+\boldsymbol{z}_1}{2},\dots,\tfrac{-\boldsymbol{y}_5+\boldsymbol{z}_5}{2} \right),
	\end{align}
	then we obtain the standard form of $ H_0 $ [Eq.~(\ref{eq:UiHU0U})]:
	\begin{align}
		&U^{-1}H_0U=\begin{pmatrix} K & K \\ -K & -K \end{pmatrix},\\
		&K=\operatorname{diag}(\kappa_1,\kappa_2,\kappa_3,\kappa_4,\kappa_5),
	\end{align}
	which shows that there are five type-I modes. The perturbative expansion of  $ \boldsymbol{y}_i $'s and corresponding dispersion relations are given by [Eqs. (\ref{eq:typeIvalue}) and (\ref{eq:typeIvector})]
	\begin{align}
		\epsilon=\sqrt{2\kappa_i}k+O(k^2),\quad \boldsymbol{\xi}=\boldsymbol{y}_i\pm\frac{k}{\sqrt{2\kappa_i}}\boldsymbol{z}_i+O(k^2),\quad i=1,\dots,5.
	\end{align}
	The exact dispersion relations are given by 
	\begin{align}
		\epsilon=\pm\sqrt{2\kappa_i k^2+k^4},\quad i=1,\dots,5,
	\end{align}
	which are consistent with the perturbation result.\\ 
	\indent In this phase, the fluctuations of physical quantities (\ref{eq:spinorgenlin01})-(\ref{eq:spinorgenlin05}) are
	\begin{align}
		&\delta\rho=\sqrt{\rho_0}\left[ \frac{(u_2+u_{-2}+v_2+v_{-2})\sin\eta}{\sqrt{2}}+(u_0+v_0)\cos\eta \right],\quad \delta M_z=\sqrt{2\rho_0}(u_2-u_{-2}+v_2-v_{-2})\sin\eta, \\
		&\delta M_+=\sqrt{2\rho_0}\left[ (u_1+v_{-1})\sin\eta+\sqrt{3}(v_1+u_{-1})\cos\eta \right],\quad \delta M_-=\sqrt{2\rho_0}\left[ (v_1+u_{-1})\sin\eta+\sqrt{3}(u_1+v_{-1})\cos\eta \right].
	\end{align}
	We can verify that the mode $ \boldsymbol{y}_1+\frac{k}{\sqrt{2\kappa_1}}\boldsymbol{z}_1 $ has finite $ \delta\rho $ and it is a sound wave. The mode $ \boldsymbol{y}_i+\frac{k}{\sqrt{2\kappa_i}}\boldsymbol{z}_i $ with $ i=3,4, $ and 5 has finite $ \delta M_z,\ \delta M_x, $ and $ \delta M_y $. So they are a spin wave with the $z$-, $x$-, and $y$-direction. Exceptionally, if $ \eta=0 $, i.e., if the phase is uniaxial nematic, $ \delta M_z $ vanishes and $ \boldsymbol{y}_3+\frac{k}{\sqrt{2\kappa_3}}\boldsymbol{z}_3 $ has only a fluctuation of higher-rank tensors. Regardless of the value of $ \eta $, the mode $ \boldsymbol{y}_2+\frac{k}{\sqrt{2\kappa_2}}\boldsymbol{z}_2 $ is always a fluctuation of higher-rank tensors. \\
	\indent The above discussion on fluctuations of physical quantities is closely related to whether a given mode is a NGM or a quasi-NGM. We note that $ \boldsymbol{y}_1, \boldsymbol{y}_3,\,\boldsymbol{y}_4, $ and $ \boldsymbol{y}_5 $ are regarded as conventional NGMs, because
	\begin{align}
		\boldsymbol{y}_3 \propto \begin{pmatrix} F_z\boldsymbol{\psi} \\ -F_z^*\boldsymbol{\psi}^* \end{pmatrix},\quad \boldsymbol{y}_4 \propto\begin{pmatrix} F_y\boldsymbol{\psi} \\ -F_y^*\boldsymbol{\psi}^* \end{pmatrix},\quad \boldsymbol{y}_5 \propto\begin{pmatrix} F_x\boldsymbol{\psi} \\ -F_x^*\boldsymbol{\psi}^* \end{pmatrix}
	\end{align}
	hold if $ \eta\ne0 $, thus they reduce to Eq.~(\ref{eq:spinorzeromode}). Therefore, only $ \boldsymbol{y}_2 $ is a quasi-NGM. This quasi-NGM is also simply obtained by differentiation of the GP equation by $ \eta $. If $ \eta=0 $, the state becomes uniaxial nematic $ \boldsymbol{\psi}=(0,0,\sqrt{\rho_0},0,0)^T $ and has an $ SO(2) $-symmetry with respect to the $ z $-axis rotation. In this case $ \boldsymbol{y}_3 $ also becomes a quasi-NGM, because $ F_z\boldsymbol{\psi}=\boldsymbol{0} $. As a consequence of the fact that $ \boldsymbol{y}_3 $ changes from a NGM to a quasi-NGM at $ \eta=0 $, the fluctuation $ \delta M_z $ of the mode $ \boldsymbol{y}_3+\frac{k}{\sqrt{2\kappa_3}}\boldsymbol{z}_3 $ vanishes at $ \eta=0 $. We mention that Ref.~\cite{2014arXiv1404.5685} has shown that both modes corresponding to $ \boldsymbol{y}_2 $ and $ \boldsymbol{y}_3 $ acquire an energy gap in the uniaxial nematic phase if the quantum fluctuation is included by the spinor Beliaev theory \cite{Phuc2013158}.
\subsection{Spin-3 BECs}\label{subsec:spin3}
	We also consider a few phases in spin-3 BECs. Even though the spin-3 BEC model is complicated, it is worth analyzing because it contains the following examples:
	\begin{itemize}
		\item The coefficient $ \epsilon_2 $ of the type-II dispersion relation $ \epsilon=\epsilon_2k^2+O(k^4) $ deviates from unity, as stated in Subsec. \ref{sec:bdwbm} and \ref{subsec:finitek}.
		\item The sound-spin composite wave excitation appears. Due to this, we need to make a nontrivial linear combination of zero modes $ (\boldsymbol{\psi},-\boldsymbol{\psi}^*) $ and $ (F_z\boldsymbol{\psi},-F_z\boldsymbol{\psi}^*) $ to obtain the standard form (\ref{eq:UiHU0U}). 
		\item The block-diagonalizing B-unitary matrix $ U $ can have a non-zero off-diagonal block. In all the previous examples of spin-$F$ condensates ($F\le2$) which we have seen so far, $ U $ has the form of $ U=\left( \begin{smallmatrix} U_0 & \\ & U_0^* \end{smallmatrix} \right) $, and the kinetic term  $ \sigma k^2 $ is invariant under the transformation by $ U $:  $ U^{-1}\sigma U k^2 =\sigma k^2 $. In the present case,  $ U^{-1}\sigma U k^2 $ may change to a different form.
	\end{itemize}
	Since spin-3 BECs have too many phases \cite{PhysRevLett.96.190405,PhysRevA.84.053616}, here we only focus on the following phases:
	\begin{itemize}
		\item F phase:  $ \boldsymbol{\psi}=\sqrt{\rho_0}(0,1,0,0,0,0,0)^T $.
		\item H phase:  $ \boldsymbol{\psi}=\sqrt{\rho_0}(\cos\eta,0,0,0,0,\sin\eta,0)^T $. 
	\end{itemize}
	\indent As already mentioned in Subsec.~\ref{sec:spinorgen}, the Hamiltonian density of the spin-3 BEC \cite{PhysRevLett.96.190405,PhysRevA.84.053616,Kawaguchi:2012ii} is given by Eqs.~(\ref{eq:spinorgen00}) and (\ref{eq:spinorgen03}):
	\begin{align}
		h=\sum_{j=-3}^3 |\nabla\psi_j|^2-\mu\rho+\tilde{c}_0\rho^2+\tilde{c}_1\boldsymbol{M}^2+\frac{\tilde{c}_2}{7}|\Theta|^2+\tilde{c}_3\operatorname{tr}\mathcal{N}^2,
	\end{align}
	where the definitions of the coefficients $ \tilde{c}_1,\tilde{c}_2,\tilde{c}_3 $ are the same with Fig.~8 of Ref.~\cite{PhysRevA.84.053616}.
	The GP equation is given by [c.f.: Eqs.~(\ref{eq:spinorgenGP00})-(\ref{eq:spinorgenGP03})]
	\begin{align}
		\mathrm{i}\partial_t\psi_j&=-\nabla^2\psi_j-\mu\psi_j+2\tilde{c}_0\rho\psi_j+\frac{2\tilde{c}_2}{7}\Theta(-1)^j\psi_{-j}^* \nonumber \\
		&\quad +\tilde{c}_1\left[ 2j\psi_jM_z+\sqrt{(3+j)(4-j)}\psi_{j-1}M_-+\sqrt{(3-j)(4+j)}\psi_{j+1}M_+ \right] \nonumber \\
		&\quad+\tilde{c}_3\biggl[ 2j^2\psi_jN_{zz}+(12-j^2)\psi_jN_{+-}+(2j-1)\sqrt{(3+j)(4-j)}\psi_{j-1}N_{z-}+(2j+1)\sqrt{(3-j)(4+j)}\psi_{j+1}N_{z+} \nonumber \\
		&\qquad\quad+\frac{1}{2}\sqrt{(3+j)(4-j)(2+j)(5-j)}\psi_{j-2}N_{--}+\frac{1}{2}\sqrt{(3-j)(4+j)(2-j)(5+j)}\psi_{j+2}N_{++}\biggr],
	\end{align}
	and the Bogoliubov equation is given by linearization of the GP equation. (We do not write down it explicitly here.) 
\subsubsection{Spin-3 F phase}\label{subsec:spin3Fphase}
	The state is given by
	\begin{align}
		\boldsymbol{\psi}=\sqrt{\rho_0}(0,1,0,0,0,0,0)^T,
	\end{align}
	and it becomes a solution to the GP equation with  $ \mu=2\rho_0(\tilde{c}_0+4\tilde{c}_1+48\tilde{c}_3) $. Since this state is inert \cite{PhysRevA.75.023625}, it always becomes a solution of the GP equation, and in particular, it becomes a ground state when $ \tilde{c}_1<0,\ \tilde{c}_2/|\tilde{c}_1|>28, $ and $ \tilde{c}_3/|\tilde{c}_1|>2/15 $ \cite{PhysRevLett.96.190405,PhysRevA.84.053616}. Since this state has a magnetization and preserves the $ U(1) $-symmetry, we expect one type-I and one type-II NGMs.\\
	\indent The Bogoliubov equation for  $ (\boldsymbol{u},\boldsymbol{v})^T=(u_3,\dots,u_{-3},v_3,\dots,v_{-3})^T $ is given by
	\begin{align}
		&(H_0+\sigma k^2)\begin{pmatrix}\boldsymbol{u} \\ \boldsymbol{v}\end{pmatrix}=\epsilon\begin{pmatrix}\boldsymbol{u} \\ \boldsymbol{v}\end{pmatrix},\label{eq:spin-3BogoF}\\ 
		&H_0=\begin{pmatrix}F&G \\ -G^* & -F^*\end{pmatrix},\quad \sigma=\begin{pmatrix} I_7 & \\ & -I_7 \end{pmatrix},
	\end{align}
	where $ F $ is a diagonal matrix 
	\begin{align}
		F=&\rho_0\operatorname{diag}\Bigl(5(2\tilde{c}_1+15\tilde{c}_3),2(\tilde{c}_0+4\tilde{c}_1+48\tilde{c}_3),3(2\tilde{c}_1+15\tilde{c}_3),4(15\tilde{c}_3-2\tilde{c}_1),-12\tilde{c}_1,\tfrac{4}{7}(\tilde{c}_2-28\tilde{c}_1),-20\tilde{c}_1\Bigr),
	\end{align}
	and nonzero components of $ G $ are
	\begin{align}
		&G_{13}=G_{31}=\sqrt{15}(2\tilde{c}_1+15\tilde{c}_3)\rho_0,\quad G_{22}=2(\tilde{c}_0+4\tilde{c}_1+48\tilde{c}_3)\rho_0=F_{22}.
	\end{align}
	Thus, Eq.~(\ref{eq:spin-3BogoF}) is divided into eleven blocks: $ (u_3,v_1),\ (u_2,v_2),\ (u_1,v_3),\ u_0,\ u_{-1},\ u_{-2}\ u_{-3},\ v_0,\ v_{-1},\ v_{-2}, $ and $  v_{-3} $. One type-I NGM is included in the block of $ (u_2,v_2) $ and type-II NGM with positive and negative dispersion relations are included in $ (u_1,v_3) $ and $ (u_3,v_1) $, respectively. All other modes are gapful. \\
	\indent To save space, we define unit vectors by
	\begin{align}
		\boldsymbol{e}_m:=\text{( $ u_m=1 $ and all other components are zero.)}
	\end{align}
	for $  m=3,\dots,-3 $. Note that the vector such that $ v_m=1 $ and all other components are zero can be written as $ \tau\boldsymbol{e}_m $, where $ \tau=\left(\begin{smallmatrix}&I_7 \\ I_7 &\end{smallmatrix}\right) $ . Then, the SSB-originated zero mode solutions with desired  $ \sigma $-orthogonal relations (\ref{eq:orthozero11})-(\ref{eq:orthozero12}) are given by
	\begin{align}
		\boldsymbol{y}_1&:=\frac{1}{\sqrt{\rho_0}}\begin{pmatrix}\boldsymbol{\psi} \\ -\boldsymbol{\psi}^* \end{pmatrix}=\boldsymbol{e}_2-\tau\boldsymbol{e}_2, \\
		\boldsymbol{x}_1&:=\frac{1}{2\sqrt{\rho_0}}\begin{pmatrix}F_-\boldsymbol{\psi} \\ -F_+^*\boldsymbol{\psi}^* \end{pmatrix}=\frac{\sqrt{10}}{2}\boldsymbol{e}_1-\frac{\sqrt{6}}{2}\tau\boldsymbol{e}_3.
	\end{align}
	Note that $ \boldsymbol{x}_1 $ has nonvanishing entries both in the $ \boldsymbol{u} $-part and $ \boldsymbol{v} $-part. This is due to the linear independence of $ F_x\boldsymbol{\psi} $ and $ F_y\boldsymbol{\psi} $, and it makes the coefficient of the quadratic dispersion relation to be greater than 1. The generalized eigenvector pairing with $ \boldsymbol{y}_1 $ is given by $ \boldsymbol{z}_1=\sigma\boldsymbol{y}_1 $ and satisfy
	\begin{align}
		H_0\boldsymbol{z}_1=2\kappa_1\boldsymbol{y}_1,\quad \kappa_1=F_{22}=2(\tilde{c}_0+4\tilde{c}_1+48\tilde{c}_3)\rho_0.
	\end{align}
	Thus the dispersion relation of the type-I Bogoliubov phonon is given by
	\begin{align}
		\epsilon = \pm\sqrt{2\kappa_1}k+O(k^2) = \pm2\sqrt{(\tilde{c}_0+4\tilde{c}_1+48\tilde{c}_3)\rho_0}k+O(k^2). \label{eq:spin-3typeI0}
	\end{align}
	On the other hand, the type-II dispersion relation is given by
	\begin{align}
		\epsilon=\frac{(\boldsymbol{x}_1,\sigma\boldsymbol{x}_1)_\sigma}{(\boldsymbol{x}_1,\boldsymbol{x}_1)_\sigma}k^2+O(k^4)=4k^2+O(k^4). \label{eq:spin-3typeII0}
	\end{align}
	Thus we have a steeper quadratic dispersion relation than that of a free particle $ \epsilon=k^2 $.\\ 
	\indent The exact dispersion relations can be obtained by solving the eigenvalue problem (\ref{eq:spin-3BogoF}) directly. The result is
	\begin{align}
		\epsilon=&\pm\sqrt{4(\tilde{c}_0+4\tilde{c}_1+48\tilde{c}_3)\rho_0k^2+k^4}, \label{eq:spin3FtypeI}\\
		&\pm(2\tilde{c}_1+15\tilde{c}_3)\rho_0\pm\sqrt{(2\tilde{c}_1+15\tilde{c}_3)^2\rho_0^2+8(2\tilde{c}_1+15\tilde{c}_3)\rho_0k^2+k^4}, \label{eq:spin3FtypeII} \\
		&\pm\left[k^2+4(15\tilde{c}_3-2\tilde{c}_1)\rho_0\right],\\
		&\pm(k^2-12\tilde{c}_1\rho_0),\\
		&\pm\left[k^2+\frac{4\rho_0}{7}(\tilde{c}_2-28\tilde{c}_1)\right],\\
		&\pm(k^2-20\tilde{c}_1\rho_0).
	\end{align}
	We can check that Eq. (\ref{eq:spin3FtypeI}) reproduces Eq. (\ref{eq:spin-3typeI0}), and Eq. (\ref{eq:spin3FtypeII}) with $(-,+)$ sign reproduces Eq. (\ref{eq:spin-3typeII0}).
\subsubsection{Spin-3 H phase}\label{subsec:spin3Hphase}
	This phase becomes the ground state when $ \tilde{c}_1>0,\ \frac{-2\tilde{c}_1}{5}<\tilde{c}_3<\frac{-2\tilde{c}_1}{15}, $ and $ \tilde{c}_2>\frac{252\tilde{c}_1(5\tilde{c}_3^2-2\tilde{c}_1\tilde{c}_3)}{4\tilde{c}_1^2+12\tilde{c}_1\tilde{c}_3+45\tilde{c}_3^2} $ \cite{PhysRevLett.96.190405,PhysRevA.84.053616}. The state is given by
	\begin{align}
		\boldsymbol{\psi}=\sqrt{\rho_0}(\sqrt{\tfrac{2+m}{5}},0,0,0,0,\sqrt{\tfrac{3-m}{5}},0)^T,
	\end{align}
	where  $ m $ represents the magnetization per density and $ -2<m<3 $ holds. Since this state has nonzero magnetization and a discrete $ C_5 $-symmetry, two type-I and one type-II NGMs appear. This state becomes a solution of the GP equation if
	\begin{align}
		\mu=\rho_0[2\tilde{c}_0+2\tilde{c}_1m^2+3\tilde{c}_3(m^2+4m+36)], \\
		m=-\frac{6\tilde{c}_3}{2\tilde{c}_1+3\tilde{c}_3} \quad\leftrightarrow\quad \tilde{c}_3=\frac{-2\tilde{c}_1m}{3(2+m)}. \label{eq:spin3Hm}
	\end{align}
	Henceforth we eliminate $ \tilde{c}_3 $ by using (\ref{eq:spin3Hm}). Though the H phase reduces to the F phase when $ m=-2 $, in the phase diagram,  $ m $ can take  $ \frac{1}{2}< m<3 $ because  $ \frac{-2\tilde{c}_1}{5}<\tilde{c}_3<\frac{-2\tilde{c}_1}{15} $. \\
	\indent  The Bogoliubov equation for  $ (\boldsymbol{u},\boldsymbol{v})^T=(u_3,\dots,u_{-3},v_3,\dots,v_{-3})^T $ is given by
	\begin{align}
		&(H_0+\sigma k^2)\begin{pmatrix}\boldsymbol{u} \\ \boldsymbol{v}\end{pmatrix}=\epsilon\begin{pmatrix}\boldsymbol{u} \\ \boldsymbol{v}\end{pmatrix}, \label{eq:spin-3bogoH} \\ 
		&H_0=\begin{pmatrix}F&G \\ -G^* & -F^*\end{pmatrix},\quad \sigma=\begin{pmatrix} I_7 & \\ & -I_7 \end{pmatrix}.
	\end{align}
	The matrices $ F $ and $ G $  have the form
	\begin{align}
		F=\begin{pmatrix} *&&&&&*& \\ &*&&&&&* \\ &&*&&&& \\ &&&*&&& \\ &&&&*&& \\ *&&&&&*& \\ &*&&&&&*  \end{pmatrix},\quad G=\begin{pmatrix} *&&&&&*& \\ &&&&*&& \\ &&&*&&& \\ &&*&&&& \\ &*&&&&&* \\ *&&&&&*& \\ &&&&*&&  \end{pmatrix}.
	\end{align}
	where nonvanishing entries are denoted by $ * $. Thus, Eq.~(\ref{eq:spin-3bogoH}) is divided into five blocks: $ (u_3,u_{-2},v_3,v_{-2}) $,\  $ (u_2,u_{-3},v_{-1}) $,\  $ (u_{-1},v_2,v_{-3}) $,\  $ (u_1,v_0) $ and $ (u_0,v_1) $. The explicit values are given by
	\begin{align}
		F_{11}&=G_{11}=\frac{2\rho_0}{5}[\tilde{c}_0(2+m)+6\tilde{c}_1(3-8m)], \\
		F_{16}&=F_{61}=G_{16}=G_{61}=\frac{2\rho_0\sqrt{(3-m)(2+m)}}{5}\left[ \tilde{c}_0-\frac{2\tilde{c}_1(6+19m)}{2+m} \right], \\
		F_{22}&=\frac{4\rho_0}{35}(7\tilde{c}_1+\tilde{c}_2)(3-m), \\
		F_{33}&=8\tilde{c}_1m\rho_0, \\
		F_{44}&=\frac{20\tilde{c}_1m^2\rho_0}{2+m}, \\
		F_{55}&=\frac{12\tilde{c}_1(1+m^2)\rho_0}{2+m}, \\
		F_{66}&=G_{66}=\frac{2(3-m)\rho_0}{5}\left[\tilde{c}_0+\frac{4\tilde{c}_1(2-7m)}{2+m}\right], \\
		F_{77}&=\frac{4\rho_0}{35}\left[\tilde{c}_2(2+m)+\frac{7\tilde{c}_1(9-66m-4m^2)}{2+m}\right], \\
		F_{27}&=F_{72}=\frac{4\rho_0\sqrt{(3-m)(2+m)}}{35}\left[ \frac{7\tilde{c}_1(3+14m)}{2+m}-\tilde{c}_2 \right], \\
		G_{25}&=G_{52}=\frac{4\sqrt{3}\rho_0\tilde{c}_1(1+3m)\sqrt{(3-m)(2+m)}}{\sqrt{5}(2+m)}, \\
		G_{34}&=G_{43}=-\frac{4\sqrt{2}\tilde{c}_1m\rho_0\sqrt{(3-m)(2+m)}}{2+m}, \\
		G_{57}&=G_{75}=\frac{4\sqrt{3}\rho_0\tilde{c}_1(3-m)(1-2m)}{\sqrt{5}(2+m)}.
	\end{align}
	 Two type-I excitations are included in the block of  $ (u_3,u_{-2},v_3,v_{-2}) $. Type-II excitations with positive and negative dispersion relations are included in the blocks of $ (u_2,u_{-3},v_{-1}) $ and $ (u_{-1},v_2,v_{-3}) $, respectively. All other modes are gapful. \\ 
	 \indent Let us first see the block of $ (u_3,u_{-2},v_3,v_{-2}) $, which has two type-I excitations. The Bogoliubov equation is given by
	 \begin{align}
	 	&H_0'\begin{pmatrix}u_3 \\ u_{-2} \\ v_3 \\ v_{-2} \end{pmatrix}=\epsilon\begin{pmatrix}u_3 \\ u_{-2} \\ v_3 \\ v_{-2} \end{pmatrix},\quad H_0'=\begin{pmatrix} F_{11}& F_{16} & G_{11} & G_{16} \\ F_{61} & F_{66} & G_{61} & G_{66} \\ -G_{11}^* & -G_{16}^* & -F_{11}^* & -F_{16}^* \\ -G_{61}^* & -G_{66}^* & -F_{61}^* & -F_{66}^*  \end{pmatrix}.
	 \end{align}
	 We can check that $ (H_0')^2=0 $, and hence Theorem~C.\ref{prp:K2} can be applied. So, we can determine $ \boldsymbol{y}_i $'s and $ \kappa_i $'s giving  the standard form (\ref{eq:UiHU0U}) by solving the eigenvalue problem of $ \sigma H_0' $. By solving it, we obtain
	 \begin{align}
	 	\boldsymbol{y}_1&:=\frac{1}{\sqrt{\rho_0}}\begin{pmatrix}\boldsymbol{\psi} \\ -\boldsymbol{\psi}^*\end{pmatrix}=\sqrt{\tfrac{2+m}{5}}(\boldsymbol{e}_3-\tau\boldsymbol{e}_3)+\sqrt{\tfrac{3-m}{5}}(\boldsymbol{e}_2-\tau\boldsymbol{e}_2), \\
	 	\boldsymbol{y}_2&:=\frac{1}{\sqrt{\rho_0}\sqrt{(3-m)(2+m)}}\left[ -m\begin{pmatrix}\boldsymbol{\psi} \\ -\boldsymbol{\psi}^*\end{pmatrix}+\begin{pmatrix}F_z\boldsymbol{\psi} \\ -F_z^*\boldsymbol{\psi}^*\end{pmatrix} \right]=\sqrt{\tfrac{3-m}{5}}(\boldsymbol{e}_3-\tau\boldsymbol{e}_3)-\sqrt{\tfrac{2+m}{5}}(\boldsymbol{e}_2-\tau\boldsymbol{e}_2),\\
	 	\boldsymbol{z}_1&:=\sigma\boldsymbol{y}_1=\sqrt{\tfrac{2+m}{5}}(\boldsymbol{e}_3+\tau\boldsymbol{e}_3)+\sqrt{\tfrac{3-m}{5}}(\boldsymbol{e}_2+\tau\boldsymbol{e}_2),\\
	 	\boldsymbol{z}_2&:=\sigma\boldsymbol{y}_2=\sqrt{\tfrac{3-m}{5}}(\boldsymbol{e}_3+\tau\boldsymbol{e}_3)-\sqrt{\tfrac{2+m}{5}}(\boldsymbol{e}_2+\tau\boldsymbol{e}_2).
	 \end{align}
	 and
	 \begin{align}
	 	&H_0\boldsymbol{z}_i=2\kappa_i\boldsymbol{y}_i,\quad i=1,2,\\
	 	&\kappa_1=2\rho_0\left[ \tilde{c}_0-\frac{2\tilde{c}_1m(m+18)}{2+m} \right],\\
	 	&\kappa_2=4\rho_0\tilde{c}_1(3-m).
	 \end{align}
	Thus, the dispersion relations are given by $ \epsilon =\sqrt{2\kappa_i}k+O(k^2),\ i=1,2 $.  While $ \boldsymbol{y}_1 $ is written by only using the phase-fluctuation zero mode, $ \boldsymbol{y}_2 $ has the form of linear combination of the phase and the spin fluctuations. So, the NGM arising from $ \boldsymbol{y}_2 $ has both density and spin fluctuations. Thus this mode is a sound-spin composite excitation. \\ 
	 \indent Next, let us see the block of $ (u_2,u_{-3},v_{-1}) $, which possesses a type-II NGM. A normalized finite-norm eigenvector constructed from two zero-mode solutions is given by
	 \begin{align}
	 	\boldsymbol{x}_1&:=\frac{1}{\sqrt{2m\rho_0}}\begin{pmatrix} F_-\boldsymbol{\psi}\\ -F_+^*\boldsymbol{\psi}^* \end{pmatrix}=\sqrt{\tfrac{3(2+m)}{5m}}\boldsymbol{e}_2+\sqrt{\tfrac{3(3-m)}{5m}}\boldsymbol{e}_{-3}-\sqrt{\tfrac{3-m}{m}}\tau\boldsymbol{e}_1.
	 \end{align}
	We can show $ (\boldsymbol{x}_1,\boldsymbol{x}_1)_\sigma=\operatorname{sgn}m $. So, it represents a normalized positive-norm eigenvector if $ m>0 $. The dispersion relation is given by
	 \begin{align}
	 	\epsilon = \frac{(\boldsymbol{x}_1,\sigma\boldsymbol{x}_1)_\sigma}{(\boldsymbol{x}_1,\boldsymbol{x}_1)_\sigma}k^2+O(k^4)=\frac{6-m}{m}k^2+O(k^4). \label{eq:spin3Htypeii}
	 \end{align}
	 Other two modes in the block of $ (u_2,u_{-3},v_{-1}) $ are gapful. Similarly, the block $ (u_{-1},v_2,v_{-3}) $ has the zero mode $ \tau\boldsymbol{x}_1^* $, and the corresponding type-II dispersion is given by $ \epsilon = -\frac{6-m}{m}k^2+O(k^4) $. \\
	\indent It is interesting to see what happens to the type-II NGM at $ m=0 $, though the state with $ -2<m<1/2 $ does not appear in the phase diagram. Since the magnetization vanishes, the WB matrix (\ref{eq:wbmatspinor}) vanishes and therefore we expect four type-I NGMs. When $ m=0 $, the expansion Eq.~(\ref{eq:spin3Htypeii}) becomes invalid. Instead, we have two type-I NGMs. Since the characteristic equation for the block of $ (u_2,u_{-3},v_{-1}) $ is cubic for $ \epsilon $, it is not smart to discuss the dispersion relation based on a lengthy exact expression. So, let us discuss the lowest order solution. If we ignore the terms $ \epsilon^\alpha k^\beta $ such that $ \alpha+\beta\ge 3 $, the characteristic equation for $ (u_2,u_{-3},v_{-1}) $ reduces to
	\begin{align}
		&\epsilon^2+2A m \epsilon+2A(m-6)k^2=0, \label{eq:spin3Htypechange}\\
		&A=\frac{2\tilde{c}_1\rho_0[14\tilde{c}_1m(3+4m)-\tilde{c}_2(1+m^2)]}{7\tilde{c}_1m(13+4m)-\tilde{c}_2(2+m)}.
	\end{align}
	When $ m \simeq  0 $,  $ A\simeq \tilde{c}_1\rho_0>0 $. Therefore, the gapless solution to the above equation becomes
	\begin{align}
		\epsilon&=-Am+\sqrt{A^2m^2+2A(6-m)k^2}\nonumber \\
		&\simeq  \begin{cases} \sqrt{12A}|k| & (m=0) \\ \frac{6-m}{m}k^2 & (m\ne 0). \end{cases}
	\end{align}
	Thus we can observe a type-I--type-II transition. When $ m=0 $, the gapful solution of Eq.~(\ref{eq:spin3Htypechange}) also changes to the gapless one. We again emphasize that the H phase is unstable if $ -2<m<1/2 $, and several other gapful excitations have Landau or dynamical instabilities in this region.
\section{The case of spacetime symmetry breaking}\label{sec:spacetime}
	\indent The general theory constructed in Sec.~\ref{sec:perturb} is restricted to the case where the state does not break a spacetime symmetry. In this section we consider two examples of spacetime symmetry breaking; the one is the Kelvin modes in a vortex, i.e., a spiral motion of a vortex and the other one is the ripplon in two-component BECs, i.e., the oscillation of a domain wall separating two immiscible condensates. We show that the main feature does not change even in the case of spacetime symmetry breaking. As with the result of Sec.~\ref{sec:perturb}, if a given zero mode solution is  $ \sigma $-orthogonal to all other zero mode solutions, the NGM originated from this zero mode has a linear dispersion, i.e., the NGM is of type-I. On the other hand, if there exists a pair such that their $ \sigma $-inner product is nonzero, then we can construct a finite-norm zero-mode solution from them and the dispersion of this NGM becomes quadratic, i.e., the NGM is of type-II. The coefficient of dispersion can be also calculated by the same method in Sec.~\ref{sec:perturb}. However, we also see that the coefficient of type-II dispersion relation diverges if the system size is sent to be infinite, which means that the naive perturbation method becomes invalid for infinite systems. We show a perspective to this issue in Subsec.~\ref{subsec:infinite}. 
\subsection{Kelvin modes in one component BECs}\label{subsec:Kelvin}
	\indent Let us consider the GP functional of a scalar condensate in three spatial dimensions:
	\begin{align}
		H=\int\mathrm{d}^3x \left(|\nabla\psi|^2-\mu|\psi|^2+c_0|\psi|^4\right)
	\end{align}
	The GP and the Bogoliubov equations are given by
	\begin{align}
		&\mathrm{i}\partial_t\psi=-\nabla^2\psi-\mu\psi+2c_0|\psi|^2\psi, \label{eq:GPst} \\
		&\mathrm{i}\partial_t\begin{pmatrix}u\\ v\end{pmatrix}=\begin{pmatrix} -\nabla^2-\mu+4c_0|\psi|^2 & 2c_0\psi^2 \\ -2c_0\psi^{2*} & \nabla^2+\mu-4c_0|\psi|^2 \end{pmatrix}\begin{pmatrix}u\\ v\end{pmatrix}.
	\end{align}
	Henceforth we consider a stationary vortex solution, and we assume that $ \psi $ is independent of  $ z,t $ and invariant under a $ z $-axis rotation. Let $ \psi $ be
	\begin{align}
		\psi(x,y)=f(r)\mathrm{e}^{\mathrm{i}n\theta},\quad (x,y)=(r\cos\theta,r\sin\theta).
	\end{align}
	Here $ n $ is an integer representing the charge of the vortex, and the non-negative function $ f(r) $ has an asymptotic behavior $ f(\infty)=\sqrt{\rho_0} $ at $ r=\infty $. From this boundary condition, the chemical potential is determined to be $ \mu=2c_0\rho_0 $, and the GP equation reduces to
	\begin{align}
		-f''(r)-\frac{f'(r)}{r}+\frac{n^2f(r)}{r^2}-2c_0f(r)\left(\rho_0-f(r)^2\right)=0.
	\end{align}
	The asymptotic expansion at $ r=\infty $ is given by
	\begin{align}
		f(r)=\sqrt{\rho_0}-\frac{n^2}{4c_0\sqrt{\rho_0}r^2}-\frac{8n^2+n^4}{32c_0^2\rho_0^{3/2}r^4}+\dotsb. \label{eq:frasymp}
	\end{align}
	Let us consider the Bogoliubov equation in the presence of this $ \psi $. We are interested in the solution of the mode propagating in the $ z $-direction and seek a solution of the form $ (u(x,y,z,t),v(x,y,z,t))=(u(x,y),v(x,y))\mathrm{e}^{\mathrm{i}(kz-\epsilon t)} $. The equation becomes
	\begin{gather}
		(H_0+\sigma k^2)\begin{pmatrix}u\\ v\end{pmatrix}=\epsilon\begin{pmatrix}u\\ v\end{pmatrix},\qquad \sigma=\begin{pmatrix}1 & \\ & -1\end{pmatrix}, \\
		H_0=\begin{pmatrix} -\partial_x^2-\partial_y^2-2c_0(\rho_0-2f^2) & 2c_0f^2\mathrm{e}^{2\mathrm{i}n\theta} \\  -2c_0f^2\mathrm{e}^{-2\mathrm{i}n\theta} & \partial_x^2+\partial_y^2+2c_0(\rho_0-2f^2) \end{pmatrix}.
	\end{gather}
	Following the same way with Sec.~\ref{sec:perturb}, we calculate an eigenvector of $ H=H_0+\sigma k^2 $ starting from an eigenvector of $ H_0 $. 
	We define the $ \sigma $-inner product for  $ w_1=(u_1(x,y),v_1(x,y))^T $ and $ w_2=(u_2(x,y),v_2(x,y))^T $ as
	\begin{align}
		(w_1,w_2)_\sigma:=\int\mathrm{d}x\mathrm{d}y\left( u_1^*u_2-v_1^*v_2 \right). \label{eq:kelvinsigmaprod}
	\end{align}
	If  $ \psi(x,y,z,t) $ is a solution of the GP equation (\ref{eq:GPst}),  $ \mathrm{e}^{\mathrm{i}\varphi}\psi(x+x_0,y+y_0,z,t) $ is also a solution. Differentiating both sides of Eq.~(\ref{eq:GPst}) by $ \varphi,x_0 $ and $ y_0 $, we obtain three SSB-originated zero-mode solutions for $ H_0 $: 
	\begin{align}
		w_{\text{phase}}=\begin{pmatrix} \psi \\ -\psi^* \end{pmatrix},\ w_{x\text{-trans}}=\begin{pmatrix} \partial_x\psi \\ \partial_x\psi^* \end{pmatrix},\ w_{y\text{-trans}}=\begin{pmatrix} \partial_y\psi \\ \partial_y\psi^* \end{pmatrix}. \label{eq:kelvinzeromode}
	\end{align}
	All these modes have zero norm, i.e., $ (w,w)_\sigma=0 $. The $ \sigma $-inner product between $ w_{\text{phase}} $ and  $ w_{x\text{-trans}} $ vanishes since $ |\psi|\rightarrow\sqrt{\rho_0} $ at infinity:
	\begin{align}
		(w_{\text{phase}},w_{x\text{-trans}})_\sigma=\int\mathrm{d}x\mathrm{d}y\partial_x|\psi|^2=0.
	\end{align}
	Similarly we also obtain $ (w_{\text{phase}},w_{y\text{-trans}})_\sigma=0 $. Thus,  $ w_{\text{phase}} $ is $ \sigma $-orthogonal to the other zero modes, and the NGM from this zero mode is of type-I. On the other hand,  $ w_{x\text{-trans}} $ and $ w_{y\text{-trans}} $ are not $ \sigma $-orthogonal, because
	\begin{align}
		&(w_{x\text{-trans}},w_{y\text{-trans}})_\sigma \nonumber \\
		=&\int\mathrm{d}x\mathrm{d}y\left[ \partial_x\psi^*\partial_y\psi-\partial_x\psi\partial_y\psi^* \right]=\int r\mathrm{d}r\mathrm{d}\theta \left[\frac{2n\mathrm{i}f(r)f'(r)}{r}\right]=2\pi\mathrm{i}n\rho_0.
	\end{align}
	Here, we have assumed $ f(0)=0 $. Thus, the $ \sigma $-inner product between the two zero modes originated from the translational-symmetry breaking gives the topological charge of the vortex.  Because of non-$ \sigma $-orthogonality,  the dispersion relation of the NGM from these two zero modes is expected to be of type-II, and this mode corresponds to the Kelvin mode. \\
	\indent In the present case, the Gram matrix is given by
	\begin{align}
		P&=\begin{pmatrix} (w_{\text{phase}},w_{\text{phase}})_\sigma \!& (w_{\text{phase}},w_{x\text{-trans}})_\sigma \!& (w_{\text{phase}},w_{y\text{-trans}})_\sigma \\ (w_{x\text{-trans}},w_{\text{phase}})_\sigma \!& (w_{x\text{-trans}},w_{x\text{-trans}})_\sigma \!& (w_{x\text{-trans}},w_{y\text{-trans}})_\sigma \\ (w_{y\text{-trans}},w_{\text{phase}})_\sigma \!& (w_{y\text{-trans}},w_{x\text{-trans}})_\sigma \!& (w_{y\text{-trans}},w_{y\text{-trans}})_\sigma \end{pmatrix}=\begin{pmatrix} 0&0&0 \\ 0&0&2\pi\mathrm{i}n\rho_0 \\ 0&-2\pi\mathrm{i}n\rho_0 & 0 \end{pmatrix}. \label{eq:kelvingram}
	\end{align} 
	So we obtain $ \frac{1}{2}\operatorname{rank}P=1 $, which implies that one type-II mode appears. \\
	\indent Let us derive the dispersion relation explicitly. Henceforth we assume $ n>0 $ without loss of generality. We can construct a positive-norm zero mode by
	\begin{align}
		w_0:=w_{x\text{-trans}}-\mathrm{i}w_{y\text{-trans}},
	\end{align}
	which has positive norm: $ (w_0,w_0)_\sigma=4\pi n\rho_0 $. (When $ n<0 $,  $ w_{x\text{-trans}}+\mathrm{i}w_{y\text{-trans}} $ has positive norm.) Let us solve the Bogoliubov equation for finite $ k $ perturbatively:
	\begin{align}
		(H_0+\sigma k^2)(w_0+w_2 k^2+w_4k^4+\dotsb)=(\epsilon_2 k^2+\epsilon_4k^4+\dotsb)(w_0+w_2 k^2+w_4k^4+\dotsb).
	\end{align}
	The equation for $ k^2 $-coefficient is given by $ H_0w_2+\sigma w_0=\epsilon_2 w_0 $. 
	Taking a  $ \sigma $-inner product between $ w_0 $ and this equation, we obtain
	\begin{align}
		\epsilon_2=\frac{(w_0,\sigma w_0)_\sigma}{(w_0,w_0)_\sigma}&=\frac{2\int\mathrm{d}x\mathrm{d}y(|\partial_x\psi|^2+|\partial_y\psi|^2)}{4\pi n \rho_0}=\frac{1}{n\rho_0}\int_0^\infty\mathrm{d}r\left[ \frac{n^2f(r)^2}{r}+rf'(r)^2 \right]. \label{eq:kelvindispersion0}
	\end{align}
	Since this integral diverges logarithmically, let us introduce a cutoff at $ r=R $. We then obtain $ \epsilon_2\simeq n\log R $, and the dispersion relation of the Kelvin mode is found to be
	\begin{align}
		\epsilon=(n\log R)k^2+\dotsb, \label{eq:kelvindispersion}
	\end{align}
	which is consistent with preceding works \cite{Kobayashi01022014}. It is worth noting that the calculation shown here does not need a concept of central extension of Lie algebra, which arises from a little sensitive mathematical treatment of the vortex core and is necessary if one wants to explain the emergence of type-II modes from  non-commutative nature of two generators \cite{Watanabe:2014pea}. \\
	\indent We can also obtain the dispersion relation of the NGM originated from $ w_{\text{phase}} $, which simply corresponds to the Bogoliubov phonon. The generalized eigenvector pairing with $ w_{\text{phase}} $  can be obtained by differentiation of the GP equation by parameters which are not originated from symmetry \cite{TakahashiPhysD}. In the present case, the differentiation by $ \rho_0 $ yields: 
	\begin{align}
		&H_0z_{\text{phase}}=2c_0w_{\text{phase}},\quad  z_{\text{phase}}:=\begin{pmatrix} \partial_{\rho_0}\psi \\ \partial_{\rho_0}\psi^* \end{pmatrix}, \label{eq:st5}
	\end{align}
	where the relation $ \partial_{\rho_0}\mu=2c_0 $ is used. Following the derivation of Subsec.~\ref{subsec:finitek}, we seek a solution for finite $ k $ by perturbative expansion:
	\begin{align}
		&(H_0+\sigma k^2)(w_{\text{phase}}+\alpha k z_{\text{phase}}+k^2 w_2+\dotsb)=(\beta k+\gamma k^2+\dotsb)(w_{\text{phase}}+\alpha k z_{\text{phase}}+k^2 w_2+\dotsb),
	\end{align}
	where $ \alpha,\beta,\gamma $ are constants to be determined. 
	From the equation of $ k^1 $-coefficient and the relation (\ref{eq:st5}), we obtain $ 2c_0\alpha=\beta $. Taking the $ \sigma $-inner product between $ w_0 $ and the equation of $ k^2 $-coefficient, we obtain
	\begin{align}
		& (w_{\text{phase}},\sigma w_{\text{phase}})_\sigma=\alpha\beta(w_{\text{phase}},z_{\text{phase}})_\sigma, \\
		&\leftrightarrow\quad  \alpha\beta=\frac{2\int\mathrm{d}x\mathrm{d}y |\psi|^2}{\int\mathrm{d}x\mathrm{d}y \partial_{\rho_0}|\psi|^2}.
	\end{align}
	In the last expression, both the numerator and the denominator diverge for infinite systems, but if we introduce a cutoff $ r=R $ , the ratio comes close to $ 2\rho_0 $ for sufficiently large $ R $, since $ |\psi|^2 \sim \rho_0 $ and $ \partial_{\rho_0}|\psi|^2\sim 1 $ hold far from the origin. Thus, we can set $ \alpha\beta=2\rho_0 $ and we obtain
	\begin{align}
		\alpha=\pm\sqrt{\frac{\rho_0}{c_0}},\quad  \beta=\pm2\sqrt{c_0\rho_0}.
	\end{align}
	Therefore, the perturbative expansions of the eigenstate and the dispersion relation are given by
	\begin{align}
		w&=w_{\text{phase}}\pm k\sqrt{\frac{\rho_0}{c_0}}z_{\text{phase}}+O(k^2), \\
		\epsilon&=\pm2\sqrt{c_0\rho_0}k+O(k^2),
	\end{align}
	respectively. Thus we obtain a type-I relation. This relation is the same with that of the Bogoliubov phonon in a uniform system [Eq.~(\ref{eq:scalarphonon})].
\subsection{Ripplons in two-component BECs}\label{subsec:ripplon}
	Let us consider the GP functional for two-component BECs in three spatial dimensions:
	\begin{align}
		H=\int\mathrm{d}^3x\left[ \sum_{i=1,2}\left(\frac{|\nabla\psi_i|^2}{2m_i}-\mu_i|\psi_i|^2\right)+\sum_{i,j=1,2}\left(g_{ij}|\psi_i|^2|\psi_j|^2\right)\right],
	\end{align}
	where $ g_{12}=g_{21} $ and $ g_{11}, g_{22} $ are positive. 
	If $ g_{12}>\sqrt{g_{11}g_{22}} $, the ground state is given by the state where two condensates $ \psi_1 $ and $ \psi_2 $ are separated. Let us consider a stationary domain-wall solution where  $ \psi_1 $ and $ \psi_2 $ are translationally invariant in the $ x $ and $ y $ directions and the domain wall exists at $ z=0 $. We set the boundary condition as
	\begin{align}
		\psi_1\rightarrow \begin{cases} \sqrt{\rho_{1}} & (z=+\infty) \\ 0 & (z=-\infty) \end{cases},\quad \psi_2\rightarrow \begin{cases} 0 & (z=+\infty) \\ \sqrt{\rho_{2}} & (z=-\infty). \end{cases}
	\end{align}
	Without loss of generality we can assume both $ \psi_1 $ and $ \psi_2 $ are real-valued. The GP equation with respect to the $ z $-axis is given by
	\begin{align}
		\frac{-1}{2m_1}\partial_z^2\psi_1-\mu_1\psi_1+2g_{11}|\psi_1|^2\psi_1+2g_{12}|\psi_2|^2\psi_1&=0, \\
		\frac{-1}{2m_2}\partial_z^2\psi_2-\mu_2\psi_2+2g_{22}|\psi_2|^2\psi_1+2g_{12}|\psi_1|^2\psi_2&=0. 
	\end{align}
	From the boundary conditions the chemical potentials are determined as
	\begin{align}
		\mu_i=2g_{ii}\rho_{i}, \quad i=1,2.
	\end{align}
	Henceforth, for simplicity, we consider the case where the parameters of  $ \psi_1 $ and $ \psi_2 $  are symmetric:
	\begin{align}
		g_{11}=g_{22}=g,\ 2m_1=2m_2=1,\ \rho_1=\rho_2=\rho_0.
	\end{align}
	If these parameters are different, the velocities of phonons in the right and left sides are unequal and more complicated reflection-refraction phenomena may occur.\\
	\indent As usual, the Bogoliubov equation is obtained by linearization of the GP equation. Here, we are interested in the solution propagating in the $ x $ and $ y $ directions. So we seek the solution of the form $ (u_i(x,y,z,t),v_i(x,y,z,t))=\mathrm{e}^{\mathrm{i}(k_xx+k_yy-\epsilon t)}(u_i(z),v_i(z)),\ i=1,2 $. Then, the Bogoliubov equation is given by
	\begin{align}
		\epsilon\begin{pmatrix} u_1 \\ u_2 \\ v_1 \\ v_2 \end{pmatrix}=(H_0+\sigma k^2)\begin{pmatrix} u_1 \\ u_2 \\ v_1 \\ v_2 \end{pmatrix},
	\end{align}
	where $ k=\sqrt{k_x^2+k_y^2},\ \sigma=\operatorname{diag}(1,1,-1,-1)$, and
	\begin{align}
		H_0&=T+S,\\
		T&=\operatorname{diag}(-\partial_z^2-2g\rho_0,-\partial_z^2-2g\rho_0,\partial_z^2+2g\rho_0,\partial_z^2+2g\rho_0),\\
		S&=2\begin{pmatrix} 2g|\psi_1|^2+g_{12}|\psi_2|^2 & g_{12}\psi_1\psi_2^* & g\psi_1^2 & g_{12}\psi_1\psi_2 \\ g_{12}\psi_1^*\psi_2 & 2g|\psi_2|^2+g_{12}|\psi_1|^2 & g_{12}\psi_1\psi_2 & g\psi_2^2 \\ -g\psi_1^{2*} & -g_{12}\psi_1^*\psi_2^* & -2g|\psi_1|^2-g_{12}|\psi_2|^2 & -g_{12}\psi_1^*\psi_2 \\ -g_{12}\psi_1^*\psi_2^* & -g\psi_2^{*2} & -g_{12}\psi_1\psi_2^* & -2g|\psi_2|^2-g_{12}|\psi_1|^2 \end{pmatrix}.
	\end{align}
	From the symmetry of the Hamiltonian, if $ (\psi_1(z),\psi_2(z)) $ is a solution of the GP equation,  $ (\psi_1(z+z_0)\mathrm{e}^{\mathrm{i}(\theta+\varphi)},\psi_2(z+z_0)\mathrm{e}^{\mathrm{i}(\theta-\varphi)}) $ is also a solution. Differentiating the GP equation by $ \theta,\ \varphi,\  $ and $ z_0 $, we obtain three SSB-originated zero-mode solutions for $ H_0 $:
	\begin{align}
		w_{\text{over}}=\begin{pmatrix} \psi_1 \\ \psi_2 \\ -\psi_1^* \\ -\psi_2^* \end{pmatrix},\ w_{\text{rel}}=\begin{pmatrix} \psi_1 \\ -\psi_2 \\ -\psi_1^* \\ \psi_2^* \end{pmatrix},\ w_{\text{trans}}=\begin{pmatrix} \partial_z\psi_1 \\ \partial_z\psi_2 \\ \partial_z\psi_1^* \\ \partial_z\psi_2^* \end{pmatrix} \label{eq:riplonzeromode}
	\end{align}
	Here ``over'' and ``rel'' mean the overall and relative phase factors. The generalized eigenvector pairing with $ w_{\text{over}} $ is found by differentiating the GP equation by $ \rho_0 $: 
	\begin{align}
		H_0z_{\text{over}}=2gw_{\text{over}},\quad z_{\text{over}}=\begin{pmatrix}\partial_{\rho_0}\psi_1 \\ \partial_{\rho_0}\psi_2 \\ \partial_{\rho_0}\psi_1^* \\ \partial_{\rho_0}\psi_2^* \end{pmatrix}.
	\end{align}
	In the present case, the $ \sigma $-inner product for two Bogoliubov wavefunctions $ w_1=(u_{11}(z),u_{12}(z),v_{11}(z),v_{12}(z))^T, w_2=(u_{21}(z),u_{22}(z),v_{21}(z),v_{22}(z))^T $ is defined as
	\begin{align}
		(w_1,w_2)_\sigma=\int\mathrm{d}z\left( u_{11}^*u_{21}+u_{12}^*u_{22}-v_{11}^*v_{21}-v_{12}^*v_{22} \right). \label{eq:ripsigmaprod}
	\end{align}
	We can check
	\begin{align}
		(w_{\text{over}},w_{\text{rel}})_\sigma=(w_{\text{over}},w_{\text{trans}})_\sigma=0,\ (w_{\text{rel}},w_{\text{trans}})_\sigma=2\rho_0.
	\end{align}
	Thus, $ w_{\text{over}} $ is $ \sigma $-orthogonal to the other two zero modes and it gives rise to a type-I NGM. On the other hand, $ w_{\text{rel}} $ and $ w_{\text{trans}} $ are not $ \sigma $-orthogonal, so these two modes become a seed of a type-II NGM. The Gram matrix becomes 
	\begin{align}
		P&=\begin{pmatrix} (w_{\text{over}},w_{\text{over}})_\sigma &(w_{\text{over}},w_{\text{rel}})_\sigma & (w_{\text{over}},w_{\text{trans}})_\sigma \\ (w_{\text{rel}},w_{\text{over}})_\sigma &(w_{\text{rel}},w_{\text{rel}})_\sigma & (w_{\text{rel}},w_{\text{trans}})_\sigma \\ (w_{\text{trans}},w_{\text{over}})_\sigma &(w_{\text{trans}},w_{\text{rel}})_\sigma & (w_{\text{trans}},w_{\text{trans}})_\sigma \end{pmatrix}=\begin{pmatrix} 0&0&0 \\ 0&0&2\rho_0 \\ 0&2\rho_0&0  \end{pmatrix}. \label{eq:gramripln}
	\end{align}
	So, we obtain $ \frac{1}{2}\operatorname{rank}P=1 $. \\
	\indent Let us determine a finite-norm eigenvector  $ w_0:=w_{\text{rel}}+cw_{\text{trans}} $ satisfying the following $ \sigma $-orthogonal relations [c.f.: Eqs. (\ref{eq:orthozero11}) and (\ref{eq:orthozero119})]:
	\begin{align}
		(w_0,\tau w_0^*)_\sigma=(w_0,\sigma\tau w_0^*)_\sigma=0,\quad \tau:=\begin{pmatrix} & I_2 \\ I_2 &  \end{pmatrix}.
	\end{align}
	$ (w_0,\tau w_0^*)_\sigma=0 $ is satisfied if $ c $ is real. From the second condition, we obtain
	\begin{align}
		c^2=\frac{(w_{\text{rel}},\sigma w_{\text{rel}})_{\sigma}}{(w_{\text{trans}},\sigma w_{\text{trans}})_{\sigma}}=\frac{\int\mathrm{d}z(|\psi_1|^2+|\psi_2|^2)}{\int\mathrm{d}z(|\partial_z\psi_1|^2+|\partial_z\psi_2|^2)}\sim \frac{2\rho_0L}{T_0},
	\end{align}
	where $ T_0=\int\mathrm{d}z(|\partial_z\psi_1|^2+|\partial_z\psi_2|^2) $ is a total kinetic energy and we have introduced a cutoff $ L $ for the integral of the numerator. (The interval of the system is set to $ [-L,L] $.) Using this $ w_0 $, we can carry out the perturbative calculation in the same way as Kelvin modes. The coefficient of quadratic dispersion is given by
	\begin{align}
		\epsilon_2=\frac{(w_0,\sigma w_0)_\sigma}{(w_0,w_0)_\sigma}=\frac{cT_0}{\rho_0}\sim \sqrt{\frac{2T_0L}{\rho_0}}. \label{eq:riplndisper0}
	\end{align}
	Thus, the dispersion relation of ripplons for a finite-size system is given by
	\begin{align}
		\epsilon = \sqrt{\frac{2T_0L}{\rho_0}} k^2+O(k^4). \label{eq:riplndisper}
	\end{align}
	The coefficient is proportional to the square root of the system length $ \sqrt{2L} $, which is consistent with the finite-size effect found by Takeuchi and Kasamatsu \cite{PhysRevA.88.043612}.
\subsection{Perspective on infinite systems}\label{subsec:infinite}
	So far we have seen that the type-II NGMs indeed have quadratic dispersion if the system size is finite. However, it is known that the dispersion of these NGMs in infinite systems is not given by an integer power of $ k $. For Kelvin modes, it is known that the dispersion is given by $ \epsilon \sim -k^2\log k $ \cite{Donnelly}. For the ripplon, while the dispersion becomes quadratic $ \epsilon \sim L^{1/2}k^2 $ in finite size systems, it becomes  $ \epsilon \sim k^{3/2} $ in infinite systems \cite{PhysRevA.88.043612}. Empirically, the correct dispersion relations in infinite systems can be obtained if we \textit{formally} replace the system length (or radius)  $ L $ (or $ R $) by $ k^{-1} $. In order to derive them, we need to modify the naive perturbation theory; if we appropriately take account of asymptotic behaviors of low-energy quasiparticle wavefunctions in large systems, we can obtain an interpolating formula which connects an integer-power dispersion in finite systems and a non-integer dispersion relations in infinite systems. These findings will be published elsewhere in future \cite{TTKKN2014}. 
\section{Summary and discussions}\label{sec:summary}
	In this last section, we provide a summary and discuss a few related and remaining issues.
\subsection{Summary}
	In this paper, we have constructed a theory to count NGMs with linear and quadratic dispersion relations in the framework of the Bogoliubov theory 
in systems with spontaneously broken internal and/or spacetime symmetries.  
In our theory, the classification of NGMs and the explicit calculation of dispersion relations are based on the following two core concepts:
\begin{enumerate}[(i)]
	\item $ \sigma $-inner products and  $ \sigma $-orthogonality --- non-positive-definite inner products between Bogoliubov quasiparticle wavefunctions.
	\item SSB-originated zero-mode solutions --- zero-energy solutions of the Bogoliubov equation derived by differentiation of the GP equation with respect to a parameter related to the symmetry.
\end{enumerate}
The concept (i) is introduced via Bogoliubov transformations, and the most general definition is given by Eq.~(\ref{eq:intro105}). For the case of internal symmetry breaking, we can use a simplified version with the omitted spatial integration (Subsec.~\ref{subsec:sigmabasis}). For the case of spacetime symmetry breaking, we can also omit the integration for the axis where the translational symmetry is preserved [Eqs.~(\ref{eq:kelvinsigmaprod}) and (\ref{eq:ripsigmaprod})]. The solution (ii) can be generally written as Eq.~(\ref{eq:intro110}), i.e., ``$(\text{a generator of the symmetry group of the system}) \times (\text{the order parameter}) $''. In the case of spacetime symmetry breaking, these solutions are simply given by spatial derivatives [Eqs.~(\ref{eq:kelvinzeromode}) and (\ref{eq:riplonzeromode})]. 

In terms of  the $ \sigma $-orthogonality of zero-mode solutions, our procedure to count type-I and type-II NGMs can be summarized as follows:
\begin{enumerate}
		\item Define  the $ \sigma $-inner product. (It is always possible if the system obeys the Hamiltonian mechanics.)
		\item Derive all zero-energy and zero-wavenumber solutions for the Bogoliubov equation. As for the SSB-originated zero-mode solutions, we can derive it by differentiation of the fundamental equation by the corresponding parameter. 
		\item If a given zero mode solution is $ \sigma $-orthogonal to all other zero-mode solutions, then the corresponding gapless mode is of type-I.
		\item If there exists a pair of  zero modes with a nonzero $ \sigma $-inner product,  then these two modes yield one type-II excitation.
	\end{enumerate}
	On the basis of this criterion, we can also construct a matrix which counts the number of type-II NGMs, namely, a Gram matrix $ P $ (Subsec.~\ref{subsec:gram} for internal symmetry breaking and Eqs.~(\ref{eq:kelvingram}) and (\ref{eq:gramripln}) for spacetime symmetry breaking). The number of type-II modes is then given by $ n_{\text{II}}=\frac{1}{2}\operatorname{rank}P $. 
	The counting method based on the $ \sigma $-orthogonality and the Gram matrix is more useful and powerful than that proposed in earlier works, because our method can easily include an additional zero modes, which are not originated from the SSB (see the example of quasi-NGMs of the spin-2 nematic phase in Subsec.~\ref{sec:spin2nematic}), and does not need a sensitive mathematical treatment for cores of topological defects in order to derive non-commutativity of translation operators (see Sec.~\ref{sec:spacetime}).\\ 
	\indent In addition to the above-mentioned main result, our paper also includes many new findings such as:
	\begin{enumerate}[(i)]
		\item The complete block-diagonalization of the WB matrix (Subsec.~\ref{sec:bdwbm}). Through this procedure, we have found that a pair of zero-modes becoming a seed of a type-II NGM is generally linearly independent, contrary to the original assumption by Nielsen and Chadha \cite{Nielsen:1975hm}.
		\item As a result of (i), if the pair of the zero modes are linearly independent, the generated type-II NGM has a dispersion relation with a coefficient larger than that of a free particle. Namely, if we write it $ \epsilon=A k^2 $, we can show $ A\ge 1 $ [Eq.~(\ref{eq:typeiidpndnt0})]. The simplest example is given by the spin-3 BEC F phase (Subsec.~\ref{subsec:spin3Fphase}).
		\item Several linear-algebraic theorems for finite-dimensional Bogoliubov equations and Bogoliubov transformations, which we refer to as B-hermitian and B-unitary matrices in this paper. (Sec.~\ref{sec:LAofsigmaprod}). In particular, we have revived Colpa's important result \cite{Colpa1986,Colpa1986II}, where the standard form of the B-hermitian matrices is given (Theorem \ref{prp:colpa001}). The standard form Eq.~(\ref{eq:UiHU0U}) based on this theorem completely describes how many type-I, type-II and gapful modes exist.
		\item A formulation of a perturbation theory by making full use of the concept of $ \sigma $-inner products (Sec.~\ref{sec:perturb}). The construction of this theory makes it possible to calculate the dispersion relation for a finite wavenumber $ k $ very systematically. For example, if the zero-mode solution of the type-II mode is given by $ \boldsymbol{x}=(\boldsymbol{u},\boldsymbol{v})^T $, the lowest-order result is given by
		\begin{align}
			\epsilon=\frac{(\boldsymbol{x},\sigma\boldsymbol{x})_\sigma}{(\boldsymbol{x},\boldsymbol{x})_\sigma}k^2=\frac{\boldsymbol{u}^\dagger\boldsymbol{u}+\boldsymbol{v}^\dagger\boldsymbol{v}}{\boldsymbol{u}^\dagger\boldsymbol{u}-\boldsymbol{v}^\dagger\boldsymbol{v}}k^2.
		\end{align}
		This kind of calculation appears in many parts of this paper, including the case of spacetime symmetry breaking [e.g., Eqs. (\ref{eq:kelvindispersion0}) and (\ref{eq:riplndisper0})].
	\end{enumerate}
	As for (ii), we mention that the relation between the intermediately-polarized phases and quantum fluctuations is recently discussed in Ref.~\cite{2014arXiv1408.1691}. We also mention that the type-I--type-II transition, which we have demonstrated in the unstable region of the spin-3 H phase (Subsec.~\ref{subsec:spin3Hphase}), is recently proposed in metastable spin texture states of spin-1 ferromagnetic BECs in a ring trap \cite{2014arXiv1408.4129}. \\ 
	\indent We finally would like to emphasize that the construction of the whole theory based on $ \sigma $-inner products and $ \sigma $-orthogonality is independent of symmetry discussions such as Lie algebras. In our formulation, the symmetry consideration is necessary only when we derive the SSB-originated zero-mode solutions at first (Subsec.~\ref{sec:symzero}), but once the zero-mode solutions are found, the rest of the theory can be constructed without using the concept of symmetry. In fact, as emphasized in Subsecs.~\ref{subsec:H0k0} and \ref{subsec:finitek}, the standard form of $ H_0 $ [Eq.~(\ref{eq:UiHU0U})] always exists even when the zero-energy eigenvector is not originated from the SSB, and once this standard form can be obtained, the perturbative calculation for finite $ k $ can be carried out without considering the physical origin of each mode.  
This is in contrast to the previous works \cite{Nielsen:1975hm,Nambu:2004yia,PhysRevD.84.125013,PhysRevLett.108.251602,Hidaka:2012ym} based on the Lie algebra of the Hamiltonian symmetry, in which one cannot take into account  accidental zero-mode solutions which do not have an SSB origin. 
Even though in this paper the analysis is restricted to the concrete multicomponent GP model [Eq.~(\ref{eq:GP00})], the procedure is general and can be used to Hamiltonian systems in general. 
It is an interesting future work to apply our method in other models.

\subsection{Discussions}\label{subsec:discussion}
\subsubsection{Non-positive-semidefinite cases}\label{subsec:bigjordan}
	In this paper, we have derived a standard form of B-hermitian matrices only when the positive-semidefinite assumption is satisfied [Theorem \ref{prp:colpa001} and Eq.~(\ref{eq:UiHU0U})]. In this case, the size of the largest Jordan block is 2. As mentioned in the introductory part of Sec.~\ref{sec:LAofsigmaprod}, the general B-hermitian matrices can have arbitrarily large Jordan blocks. However, we can show that if $ H_0 $ has a Jordan block whose size is greater than 2, the finite-wavenumber matrix $ H=H_0+\sigma k^2 $ always has a complex eigenvalue. Its derivation is given in \ref{app:jordan}. Thus, if we are only interested in the gapless modes with stable backgrounds, such cases are physically less important. 
\subsubsection{Explicit symmetry breaking and ``massive'' Nambu-Goldstone modes}\label{eq:subsecexmass}
	The ``massive'' NGMs in the presence of explicitly symmetry-breaking terms, e.g., an external magnetic field, are discussed in Refs.~\cite{PhysRevLett.110.011602,PhysRevLett.111.021601,Nicolis,2014arXiv1406.6271}. They are gapful, but their presence is still universally ensured by symmetry and the value of the gap is determined only by a symmetry discussion. These modes also can be treated in the framework of the Bogoliubov theory. See \ref{app:exmassive} for a detail. A well-known example is a spinor BEC with a magnetic field:
	\begin{align}
		h=\sum_{j=-F}^F|\nabla\psi_j|^2-\mu\rho +h_{\text{int}}-B M_z,
	\end{align}
	where the model is the same with that in Subsec.~\ref{sec:spinorgen} except for the last term. $ M_z $ is a $z$-component of the magnetization and $ B $ is a strength of the magnetic field. By this term, the symmetry of the system reduces from $ U(1)\times SO(3) $ to $ U(1)\times SO(2) $. As derived in Eqs.(\ref{eq:finiteSSB0}) and (\ref{eq:finiteSSB}), in addition to the ordinary zero-energy SSB-originated solutions
	\begin{align}
		\begin{pmatrix}\boldsymbol{u} \\ \boldsymbol{v}\end{pmatrix}=\begin{pmatrix} \boldsymbol{\psi} \\ -\boldsymbol{\psi}^* \end{pmatrix}, \begin{pmatrix} F_z\boldsymbol{\psi} \\ -F_z^*\boldsymbol{\psi}^* \end{pmatrix},\qquad \epsilon=0,
	\end{align}
	we obtain the SSB-originated \textit{finite-energy} solutions:
	\begin{alignat}{2}
		\begin{pmatrix}\boldsymbol{u} \\ \boldsymbol{v}\end{pmatrix}=\begin{pmatrix} F_\mp \boldsymbol{\psi} \\ -F_\pm^*\boldsymbol{\psi}^* \end{pmatrix},\qquad \epsilon=\pm B, \label{eq:finiteSSB22}
	\end{alignat}
	where  $ \epsilon=\pm B $ is a gap in the energy spectra. We thus obtain massive NGMs in the Bogoliubov theory. The dispersion relation for a finite wavenumber $ k $ can be also derived [Eq.~(\ref{eq:lastdance})]. They reduce to four zero-energy solutions (\ref{eq:spinorzeromode02}) when $ B=0 $. As discussed below, these solutions play an important role to explain a perfect tunneling of ``massive'' NGMs. 
\subsubsection{SSB-originated zero-modes as an origin of perfect tunneling of NGMs}\label{subsec:tunneling}
	The SSB-originated zero-mode solutions survive even when there exists an external potential and the order parameter is spatially non-uniform, unless the potential does not break a corresponding symmetry. In order to emphasize this, let us write them with the position variable $ \boldsymbol{r} $: 
	\begin{align}
		\begin{pmatrix} \boldsymbol{u}(\boldsymbol{r}) \\ \boldsymbol{v}(\boldsymbol{r}) \end{pmatrix} = \begin{pmatrix} Q_j\boldsymbol{\psi}(\boldsymbol{r}) \\ -Q_j^*\boldsymbol{\psi}(\boldsymbol{r})^* \end{pmatrix}, \label{eq:persp21}
	\end{align}
	where  $ Q_j $ is a generator of the symmetry group $ G $. For example, in the case of spinor BECs, $ Q_j=I,\ F_x,\ F_y, $ and $ F_z $. Note that the sign of $ \boldsymbol{v}(\boldsymbol{r}) $ is frequently taken in an opposite way in many papers. The ``massive'' NGMs, i.e., finite-energy solutions (\ref{eq:finiteSSB22}) in a magnetic field also exist, if the potential does not break the symmetry with respect to a $ z $-axis rotation. \\
	\indent The above zero-mode solutions in non-uniform systems have a close relation to the scattering properties of NGMs. Scattering problems of NGMs are extensively studied in Refs. \cite{Kovrizhin,Kagan,DanshitaYokoshiKurihara,doi:10.1143/JPSJ.77.013602,PhysRevA.78.063611,PhysRevA.83.033627,PhysRevA.83.053624,PhysRevA.84.013616}. It is known that NGMs show a perfect tunneling in the long-wavelength limit. As a first example, the tunneling properties of Bogoliubov phonons in scalar BECs were studied in Refs. \cite{Kovrizhin,Kagan,DanshitaYokoshiKurihara,doi:10.1143/JPSJ.77.013602}. In particular, the physical origin of this perfect tunneling was shown to be a coincidence between the condensate wavefunction and quasiparticle wavefunctions \cite{doi:10.1143/JPSJ.77.013602}. This solution is just Eq.~(\ref{eq:persp21}) with $ Q_j=I $. The similar coincidences are also found in the perfect tunneling of spin waves in spinor BECs \cite{PhysRevA.83.033627,PhysRevA.83.053624,PhysRevA.84.013616}, and they are reduced to the cases $ Q_j=F_x,\ F_y, $ and $ F_z $ of Eq.~(\ref{eq:persp21}). Moreover, the perfect tunneling of gapful modes in the presence of magnetic fields, e.g., the transverse spin waves in the current-carrying ferromagnetic BEC \cite{PhysRevA.83.033627} and the unsaturated magnetization phases \cite{PhysRevA.83.053624} can be also explained by a position-dependent version of the ``massive'' NGMs (\ref{eq:finiteSSB22}).  [Note that their wavenumbers are not necessarily equal to zero in the current-carrying case because the form of dispersion relation may change, though the universal existence of the solution (\ref{eq:finiteSSB22}) with the energy $ \epsilon=\pm B $ is unchanged.] \\
	\indent Thus, the SSB-originated zero- and finite-energy solutions provide an explanation for all perfect tunneling properties of NGMs known so far.

\section*{Acknowledgment}
We would like to thank Shun Uchino, Michikazu Kobayashi, and Masaya Kunimi for useful discussions. 
The work of MN is supported in part by Grant-in-Aid for Scientific
Research (No. 25400268) and by the ``Topological Quantum Phenomena'' 
Grant-in-Aid for Scientific Research on Innovative Areas (No. 25103720)  
from the Ministry of Education, Culture, Sports, Science and Technology 
(MEXT) of Japan.

\appendix
\section{Bogoliubov approximation in quantum field theory}\label{app:quantum}
	In this appendix we show the equivalence of the problem between the linear waves of classical field theory and  quantum field theory within the framework of Bogoliubov approximation. Let the 2nd-quantized Hamiltonian for the $ N $-component Bose gas be
	\begin{align}
		&\hat{\mathcal{H}}=\int \hat{h} \mathrm{d}x,\\
		&\hat{h}=-\sum_{i=1}^N\hat{\psi}_i^\dagger \nabla^2\hat{\psi}_i+F(\{ \hat{\psi}_l^\dagger, \hat{\psi}_l \}).
	\end{align}
	Here,  $ \hat{\psi}_1,\dots,\hat{\psi}_N $ are field operators satisfying the bosonic commutation relations $ [\hat{\psi}_i(x),\hat{\psi}_j(y)]=0 $ and $ [\hat{\psi}_i(x),\hat{\psi}_j^\dagger(y)]=\delta_{ij}\delta(x-y) $, and $ F(\{\psi_l^*,\psi_l\})=F(\psi_1^*,\dots,\psi_N^*,\psi_1,\dots,\psi_N) $ is a c-number polynomial function and $ F(\{ \hat{\psi}_l^\dagger, \hat{\psi}_l \}) $ is defined by substituting the field operators and sorting them in normal order. The spatial dimension is arbitrary and if it is $ d $,  $ \mathrm{d}x $ and $ \delta(x-y) $ should be read as $ \mathrm{d}x=\mathrm{d}x_1\dotsm\mathrm{d}x_d $ and $ \delta(x-y)=\delta(x_1-y_1)\dotsm\delta(x_d-y_d) $. 
	Let us assume that the Bose condensation occurs and each $ \hat{\psi}_i$ has a finite expectation value $ \braket{\hat{\psi}_i} $. We then write the field operator as the sum of the expectation value and the deviation from it: $ \hat{\psi}_i=\braket{\hat{\psi}_i}+\delta\hat{\psi}_i,\ \delta\hat{\psi}_i:=\hat{\psi}_i-\braket{\hat{\psi}_i} $. By definition $ \braket{\delta\hat{\psi}_i}=0 $. Substituting them to the Hamiltonian, we ignore higher order terms with respect of $ \delta\hat{\psi}_i $ and keep only quadratic terms with the assumption that these deviations are small. Writing the expectation value by hatless notation $ \psi_i=\braket{\hat{\psi}_i} $, the approximate Hamiltonian becomes
	\begin{align}
		\hat{h}&\simeq h_0+\hat{h}_1+\hat{h}_2,\\
		h_0&= -\sum_{i=1}^N\psi_i^* \nabla^2 \psi_i +F(\{\psi^*_l,\psi_l\}), \\
		\hat{h}_1&=\sum_{i=1}^N \left[\left( -\nabla^2\psi_i+\frac{\partial F}{\partial\psi_i^*} \right)\delta\hat{\psi}_i^\dagger+\left( -\nabla^2\psi_i^*+\frac{\partial F}{\partial\psi_i} \right)\delta\hat{\psi}_i\right], \label{eq:appd1} \\
		\hat{h}_2&=-\sum_{i=1}^N\delta\hat{\psi}_i^\dagger\nabla^2\delta\hat{\psi}_i+\sum_{i,j}\bigg[\frac{1}{2}\frac{\partial^2F}{\partial\psi_i\partial\psi_j}\delta\hat{\psi}_i\delta\hat{\psi}_j+\frac{\partial^2F}{\partial\psi_i^*\partial\psi_j}\delta\hat{\psi}_i^\dagger\delta\hat{\psi}_j+\frac{1}{2}\frac{\partial^2F}{\partial\psi_i^*\partial\psi_j^*}\delta\hat{\psi}_i^\dagger\delta\hat{\psi}_j^\dagger\bigg], \label{eq:appd2}
	\end{align}
	where the arguments of the partial derivatives of $ F $ in Eqs. (\ref{eq:appd1}) and (\ref{eq:appd2}) are merely the classical fields $ \{ \psi_l,\psi_l^* \} $, so these are c-number functions. Let us impose the extremum condition for classical fields:
	\begin{align}
		\frac{\delta \braket{\hat{\mathcal{H}}}}{\delta \psi_k(x)^*}=\Braket{\frac{\delta \hat{\mathcal{H}}}{\delta\psi_k(x)^*}}=0.
	\end{align}
	We then obtain the equation
	\begin{align}
		-\nabla^2\psi_k+\frac{\partial F}{\partial \psi_k^*}+\sum_{i,j}\left[\frac{1}{2}\frac{\partial^3F}{\partial\psi_k^*\partial\psi_i\partial\psi_j}\braket{\delta\hat{\psi}_i\delta\hat{\psi}_j}+\frac{\partial^3F}{\partial\psi_k^*\partial\psi_i^*\partial\psi_j}\braket{\delta\hat{\psi}_i^\dagger\delta\hat{\psi}_j}+\frac{1}{2}\frac{\partial^3F}{\partial\psi_k^*\partial\psi_i^*\partial\psi_j^*}\braket{\delta\hat{\psi}_i^\dagger\delta\hat{\psi}_j^\dagger}\right]=0. \label{eq:appgp1}
	\end{align}
	Compared to the classical GP equation (\ref{eq:GP01}), it contains the contribution from the expectation value of quasiparticles. The corresponding equation for the single component case (the case of scalar BEC) is found in Ref.~\cite{PhysRevB.53.9341}. Note that if we want to formulate the Hartree-Fock-Bogoliubov theory, we need to keep a little more kinds of terms for $ \delta\hat{\psi}_i $'s \cite{PethickSmith,PhysRevB.53.9341}. In the simplest Bogoliubov approximation, all expectation values of quasiparticles are ignored, and Eq.~(\ref{eq:appgp1}) is simply reduced to the GP equation (\ref{eq:GP01}). If we use the solution of the GP equation,  $ \hat{h}_1 $ vanishes automatically and the remaining work is to diagonalize $ \hat{h}_2 $ by the Bogoliubov transformation. Let us consider the Bogoliubov transformation
	\begin{align}
		\delta\hat{\psi}_i(x)&=\sum_n u_i^{(n)}(x)\hat{\alpha}_n+v_i^{(n)}(x)^*\hat{\alpha}_n^\dagger \label{eq:appdbog} \\
		\leftrightarrow\quad  \hat{\alpha}_n&=\sum_{i=1}^N\int\mathrm{d}x\left( u_i^{(n)*}(x)\hat{\psi}_i(x)-v_i^{(n)*}(x)\hat{\psi}_i^\dagger(x) \right),
	\end{align}
	where the subscript $ n $ is a label of quasiparticle eigenstates (\textit{not} to be confused with the number of the component). The operators $ \hat{\alpha}_n $ also satisfy the bosonic commutation relations: $ [\hat{\alpha}_m,\hat{\alpha}_n]=0 $ and $ [\hat{\alpha}_m,\hat{\alpha}_n^\dagger]=\delta_{mn} $. In order for these bosonic commutation relations to hold, the coefficient functions $ u_i^{(n)}(x) $ and $ v_i^{(n)}(x) $ must satisfy
	\begin{align}
		\sum_{i=1}^N\int\mathrm{d}x\left( u_i^{(m)*}(x)u_i^{(n)}(x)-v_i^{(m)*}(x)v_i^{(n)}(x) \right)&=\delta_{mn}, \label{eq:appdbu01} \\
		\sum_{i=1}^N\int\mathrm{d}x\left( u_i^{(m)}(x)v_i^{(n)}(x)-v_i^{(m)}(x)u_i^{(n)}(x) \right)&=0, \label{eq:appdbu02} \\
		\sum_n\left( u_i^{(n)}(x)u_j^{(n)}(y)^*-v_i^{(n)}(x)^*v_j^{(n)}(y) \right)&=\delta_{ij}\delta(x-y), \\
		\sum_n\left( u_i^{(n)}(x)v_j^{(n)}(y)^*-v_i^{(n)}(x)^*u_j^{(n)}(y) \right)&=0.
	\end{align}
	These relations can be regarded as an infinite-dimensional version of B-unitary condition discussed in  Sec.~\ref{sec:LAofsigmaprod}. To diagonalize $ \hat{h}_2 $, we choose $ (u_i^{(n)}(x),v_i^{(n)}(x)) $ to satisfy
	\begin{align}
		-\nabla^2u_i^{(n)}+\sum_j \left(F_{ij}u_j^{(n)}+G_{ij}v_j^{(n)}\right)&=\epsilon^{(n)}u_i^{(n)},  \label{eq:appd11} \\
		\nabla^2v_i^{(n)}-\sum_j \left(F_{ij}^*v_j^{(n)}+G_{ij}^*u_j^{(n)}\right)&=\epsilon^{(n)}v_i^{(n)},  \label{eq:appd12}
	\end{align}
	where  $ F_{ij} $ and $ G_{ij} $ are defined by Eq.~(\ref{eq:FijGij}) and satisfy Eq. (\ref{eq:FijGij2}), and the eigenvalue $ \epsilon^{(n)} $ is assumed to be real. 
	Let us substitute Eq.~(\ref{eq:appdbog}) into Eq.~(\ref{eq:appd2}) after rewriting the kinetic energy term in Eq.~(\ref{eq:appd2}) as 
	\begin{align}
		\sum_{i=1}^N\delta\hat{\psi}_i^\dagger\nabla^2\delta\hat{\psi}_i \rightarrow \frac{1}{2}\sum_{i=1}^N\left(\delta\hat{\psi}_i^\dagger\nabla^2\delta\hat{\psi}_i+(\nabla^2\delta\hat{\psi}_i^\dagger)\delta\hat{\psi}_i  \right)
	\end{align}
	by integration by parts. We then obtain
	\begin{align}
		\hat{h}_2&=\sum_{i,m,n}\frac{\epsilon^{(n)}+\epsilon^{(m)}}{2}\left( u_i^{(m)*}u_i^{(n)}-v_i^{(m)*}v_i^{(n)} \right)\hat{\alpha}_m^\dagger\hat{\alpha}_n-\sum_{i,n}\epsilon^{(n)}|v_i^{(n)}|^2 \nonumber \\
		&\quad+\frac{1}{2}\sum_{i,m,n}\epsilon^{(n)}\left( u_i^{(n)*}v_i^{(m)*}-v_i^{(n)*}u_i^{(m)*} \right)\hat{\alpha}_m^\dagger\hat{\alpha}_n^\dagger+\frac{1}{2}\sum_{i,m,n}\epsilon^{(n)}\left( u_i^{(n)}v_i^{(m)}-v_i^{(n)}u_i^{(m)} \right)\hat{\alpha}_m\hat{\alpha}_n.
	\end{align}
	Integrating this expression and using Eqs.~(\ref{eq:appdbu01}) and (\ref{eq:appdbu02}), we obtain
	\begin{align}
		\int\hat{h}_2\mathrm{d}x=\sum_n\epsilon^{(n)}\hat{\alpha}_n^\dagger\hat{\alpha}_n-\sum_{i,n}\epsilon^{(n)}\int|v_i^{(n)}|^2\mathrm{d}x.
	\end{align}
	Thus $ \hat{h}_2 $ can be diagonalized by the Bogoliubov transformation (\ref{eq:appdbog}) with Eqs. (\ref{eq:appd11}) and (\ref{eq:appd12}), which are equivalent to the linearized equations for a classical field [Eqs.~(\ref{eq:Bogo01}) and (\ref{eq:Bogo02})]. Therefore both theories share the same fundamental equations and the results in the main part of this paper are also applicable to quantum field theory within the framework of the Bogoliubov approximation.
\section{Equivalence between symplectic group and Bogoliubov transformation group}\label{app:classicalhamilton}
	The classical Hamiltonian mechanics can be formulated in terms of generalized position and momentum variables $ q $ and $ p $. We can rewrite it by the complex variable  $ \psi=(q+\mathrm{i}p)/\sqrt{2} $, which is convenient for GP or Ginzburg-Landau type equations. Here we briefly summarize the relation between both representations, and show that symplectic matrices and B-unitary matrices are equivalent up to trivial linear transformation. As stated in Sec.~\ref{sec:LAofsigmaprod}, B-unitary matrix corresponds to the Bogoliubov transformation for bosonic field operators. So the symplectic group and the Bogoliubov transformation group are equivalent. \\
	\indent Let $ H(\{q,p\})=H(q_1,\dots,q_n,p_1,\dots,p_n) $ be the Hamiltonian with $ n $ degree of freedom. The Hamilton equation is given by
	\begin{align}
		\dot{p}_j=-\frac{\partial H}{\partial q_j},\quad \dot{q}_j=\frac{\partial H}{\partial p_j}. \label{eq:apphami1}
	\end{align}
	Substituting $ q_j=q_j+\delta q_j, p_j=p_j+\delta p_j $ to the above, and ignoring higher order terms with respect to $ (\delta q_j,\delta p_j) $, we obtain the equation for the linearized small oscillations in the neighbor of a certain solution of Eq. (\ref{eq:apphami1}):
	\begin{align}
		\frac{\mathrm{d} }{\mathrm{d} t}\begin{pmatrix} \delta q \\ \delta p \end{pmatrix}=L \begin{pmatrix} \delta q \\ \delta p \end{pmatrix},\quad L=\begin{pmatrix} A & B \\ -C & -A^T \end{pmatrix}, \\
		\delta q:=(\delta q_1,\dots,\delta q_n)^T,\quad \delta p:=(\delta p_1,\dots,\delta p_n)^T,\\
		 A_{ij}=\frac{\partial^2 H}{\partial p_i \partial q_j},\quad B_{ij}=\frac{\partial^2H }{\partial p_i\partial p_j},\quad C_{ij}=\frac{\partial^2H }{\partial q_i\partial q_j}.
	\end{align}
	The matrix  $ L $ is sometimes called \textit{a ``hamiltonian'' matrix} in the literature of dynamical systems. This naming is rather confusing for condensed matter physicists, because a ``hamiltonian'' matrix is not hermitian! In order to avoid a confusion with hermitian matrices, we always add a double quotation mark. If we consider the small oscillation around a stationary solution, the eigenvalue of $ L $ describes the stability of the stationary point. The classification for normal forms of ``hamiltonian'' matrices is given in Arnold's book (Appendix 6 of Ref. \cite{Arnold}). \\ 
	\indent A symplectic matrix $ R $ is defined as a linear transformation for $ (p_j,q_j) $ which preserves the Hamilton equation, and must satisfy the following condition:
	\begin{align}
		R^*=R,\quad R^TJR=J,\quad J:=\begin{pmatrix} & I_n \\ -I_n &  \end{pmatrix}.
	\end{align}
	If we define $ (q',p')^T:=R(q,p)^T $ and $ H'(\{q',p'\}):=H(\{q,p\}) $, the new variables also satisfy the Hamilton equation. Note that the exponential of the ``hamiltonian'' matrix $ \mathrm{e}^{Lt} $ is symplectic.  So, using the above  $ J $, the ``hamiltonian'' matrix satisfies
	\begin{align}
		L^*=L,\quad L^T J+J L=0.
	\end{align}
	\indent Let us define complex variables by
	\begin{align}
		\begin{cases}
		q_j=\frac{\psi_j+\psi_j^*}{\sqrt{2}} \\
		p_j=\frac{\psi_j-\psi_j^*}{\mathrm{i}\sqrt{2}}
		\end{cases}
		\quad \leftrightarrow \quad
		\begin{cases}
		\psi_j=\frac{q_j+\mathrm{i}p_j}{\sqrt{2}} \\
		\psi_j^*=\frac{q_j-\mathrm{i}p_j}{\sqrt{2}}, 
		\end{cases}
	\end{align}
	and define a new Hamiltonian by $ \tilde{H}(\{\psi^*,\psi\})=H(\{\frac{\psi+\psi^*}{\sqrt{2}},\frac{\psi-\psi^*}{\sqrt{2}\mathrm{i}}\}) $. Then, the Hamilton equation is given by
	\begin{align}
		\mathrm{i}\dot{\psi}_j=\frac{\partial \tilde{H}}{\partial \psi_j^*},\quad -\mathrm{i}\dot{\psi}^*_j=\frac{\partial \tilde{H}}{\partial \psi_j},
	\end{align}
	and the linearized equation is
	\begin{align}
		&\mathrm{i}\frac{\mathrm{d} }{\mathrm{d} t}\begin{pmatrix} \delta \psi \\ \delta \psi^* \end{pmatrix}=\tilde{L} \begin{pmatrix} \delta\psi \\ \delta \psi^* \end{pmatrix},\quad \tilde{L}=\begin{pmatrix} F & G \\ -G^* & -F^* \end{pmatrix}, \\
		&\delta \psi:=(\delta \psi_1,\dots,\delta \psi_n)^T,\quad \delta \psi^*:=(\delta \psi_1^*,\dots,\delta \psi^*_n)^T,\\
		&F_{ij}=\frac{\partial^2 \tilde{H}}{\partial \psi_i^* \partial \psi_j},\quad G_{ij}=\frac{\partial^2\tilde{H}}{\partial \psi_i\partial \psi_j}.
	\end{align}
	Here, the matrix $ \tilde{L} $ is B-hermitian. The linearized variables $ (\delta q, \delta p) $ and $ (\delta \psi, \delta\psi^*) $ are related as
	\begin{align}
		\begin{pmatrix} \delta \psi \\ \delta \psi^* \end{pmatrix}=U_0 \begin{pmatrix}\delta q \\ \delta p \end{pmatrix},\quad U_0:=\frac{1}{\sqrt{2}}\begin{pmatrix} I_n & \mathrm{i}I_n \\ I_n & -\mathrm{i}I_n \end{pmatrix}.
	\end{align}
	Therefore, the ``hamiltonian'' matrix $ L $ and the B-hermitian matrix $ \tilde{L} $ satisfy
	\begin{align}
		\mathrm{i}L=U_0^{-1}\tilde{L}U_0.
	\end{align}
	Because of the imaginary number $ \mathrm{i} $, a pure imaginary eigenvalue of $ L $ corresponds to a real eigenvalue of $ \tilde{L} $. The correspondence between symplectic matrix  $ R $ and the B-unitary matrix $ U $ is given by
	\begin{align}
		R=U_0^{-1}UU_0.
	\end{align}
	From $ R^*=R $ and $ R^T=-JR^{-1}J $, we obtain the B-unitary conditions $ U^*=\tau U \tau $ and $ U^\dagger=\sigma U^{-1}\sigma $, respectively.
\section{Proofs of Theorems and Propositions in Sec.~\ref{sec:LAofsigmaprod}}\label{app:proof}
	In this appendix we provide the complete proofs for theorems and propositions given in Sec.~\ref{sec:LAofsigmaprod}. 
	\begin{proof}[Proof of the fundamental properties (i)-(iii) in Subsec.~\ref{subsec:sigmabasis}](ii): Let $ W $ be a subset of $ V $ such that its elements are $ \sigma $-orthogonal to all vectors in $ V $. We can easily show that $ W $ becomes a vector space, and therefore its property does not depend on a choice of basis. (i),(iii): Let $ \boldsymbol{w}_1,\dots,\boldsymbol{w}_r $ be a basis of $ V $. We define a $ (2N)\times r $ matrix by $ P=(\boldsymbol{w}_1,\dots,\boldsymbol{w}_r) $. Let us consider the Gram matrix with respect to the $ \sigma $-inner product $ P^\dagger \sigma P $, which gives a list of $ \sigma $-inner products in the current basis. Since $ P^\dagger \sigma P $ is hermitian, there exists an invertible matrix $ Q $ such that  $ Q^\dagger P^\dagger \sigma P Q =\operatorname{diag}(1,\dots,1,-1,\dots,-1,0,\dots,0) $. If we define a new basis by $ \boldsymbol{w}'_i=\sum_j\boldsymbol{w}_jQ_{ji} $, then $ \{\boldsymbol{w}_1',\dots,\boldsymbol{w}_r'\} $ becomes a  $ \sigma $-orthonormal system. Furthermore, by Sylvester's law of inertia, the number of  $ 1,-1 $ and $ 0 $ in $ Q^\dagger P^\dagger \sigma P Q $ does not depend on the diagonalizing matrix $ Q $. 
	\end{proof}
	To prove the properties (iv) and (v), we first prove the following proposition:
	\begin{prpappC} \label{prp:sigmabasis3} Let  $ \{ \boldsymbol{x}_1,\dots,\boldsymbol{x}_p,\boldsymbol{y}_1,\dots,\boldsymbol{y}_q,\boldsymbol{z}_1,\dots,\boldsymbol{z}_t \}$ be a  $ \sigma $-orthonormal basis such that $ \boldsymbol{x}_i $, $ \boldsymbol{y}_i $, and $ \boldsymbol{z}_i $ have positive, negative, and zero norm, respectively. Let us write $ \boldsymbol{x}_i=(\boldsymbol{u}_i,\boldsymbol{v}_i)^T,\ \boldsymbol{y}_i=(\boldsymbol{u}'_i,\boldsymbol{v}'_i)^T, $ and $ \boldsymbol{z}_i=(\boldsymbol{u}_i'',\boldsymbol{v}_i'')^T $, where  $ \boldsymbol{u}_i,\boldsymbol{u}'_i,\boldsymbol{u}''_i, \boldsymbol{v}_i,\boldsymbol{v}_i',\boldsymbol{v}_i''\in\mathbb{C}^N $. The following (a)-(d) hold:
	\begin{enumerate}[(a)]
		\item $\boldsymbol{u}_1,\dots,\boldsymbol{u}_p$ are linearly independent.
		\item $\boldsymbol{v}'_1,\dots,\boldsymbol{v}'_q$ are linearly independent.
		\item $\boldsymbol{u}''_1,\dots,\boldsymbol{u}''_t$ are linearly independent.
		\item $\boldsymbol{v}''_1,\dots,\boldsymbol{v}''_t$ are linearly independent.
	\end{enumerate}
	\end{prpappC}
	\begin{proof}[Proof of Proposition C.\ref{prp:sigmabasis3}]
	(a): The case of $ p=1 $ is trivial. Let $ p\ge 2 $ and assume the relation $ \boldsymbol{u}_p=\sum_{i=1}^{p-1}c_i\boldsymbol{u}_i $, where at least one $ c_i $ satisfies $ c_i\ne 0 $. Henceforth we abbreviate $ \sum_{i=1}^{p-1} $ as $ \sum $. By the Cauchy-Schwartz inequality, 
	\begin{align}
		\left|\boldsymbol{v}_p^\dagger\left(\sum c_i\boldsymbol{v}_i\right)\right|^2\le \left(\sum c_i\boldsymbol{v}_i\right)^\dagger\left(\sum c_i\boldsymbol{v}_i\right)\boldsymbol{v}_p^\dagger \boldsymbol{v}_p. \label{eq:basisproof00}
	\end{align}
	On the other hand, using the $ \sigma $-orthogonality $ \boldsymbol{u}_i^\dagger\boldsymbol{u}_j=\boldsymbol{v}_i^\dagger\boldsymbol{v}_j+\delta_{ij} $ and the first assumption $ \boldsymbol{u}_p=\sum c_i\boldsymbol{u}_i $, we obtain
	\begin{align}
		\left|\boldsymbol{v}_p^\dagger\left(\sum c_i\boldsymbol{v}_i\right)\right|^2=\left|\boldsymbol{u}_p^\dagger\left(\sum c_i\boldsymbol{u}_i\right)\right|^2=\left(\sum c_i\boldsymbol{u}_i\right)^\dagger\left(\sum c_i\boldsymbol{u}_i\right) \boldsymbol{u}_p^\dagger \boldsymbol{u}_p=\left[ \sum |c_i|^2+\left(\sum c_i\boldsymbol{v}_i\right)^\dagger\left(\sum c_i\boldsymbol{v}_i\right) \right]\left( 1+\boldsymbol{v}_p^\dagger\boldsymbol{v}_p \right). \label{eq:basisproof01}
	\end{align}
	Combining Eqs.~(\ref{eq:basisproof00}) and (\ref{eq:basisproof01}), we get $ (\sum|c_i|^2)(1+\boldsymbol{v}_p^\dagger\boldsymbol{v}_p)+(\sum c_i\boldsymbol{v}_i)^\dagger(\sum c_i\boldsymbol{v}_i)\le0 $, a contradiction. (b): The same as (a). (c): It is trivial if $ t=1 $. Let $ t\ge 2 $ and assume the relation  $ \boldsymbol{u}_t''=\sum_{i=1}^{t-1}c_i\boldsymbol{u}_i'' $, where at least one $ c_i $ satisfies  $ c_i\ne 0 $. Henceforth we abbreviate  $ \sum_{i=1}^{t-1} $ as $ \sum $. By a similar calculation to Eq.~(\ref{eq:basisproof01}), we obtain $ \left| (\boldsymbol{v}_t'')^\dagger\sum c_i\boldsymbol{v}_i'' \right|^2=\left(\sum c_i\boldsymbol{v}_i''\right)^\dagger\left(\sum c_i\boldsymbol{v}_i''\right)(\boldsymbol{v}_t'')^\dagger\boldsymbol{v}_t'' $, which is the case of the equality in the Cauchy-Schwartz inequality. Therefore, a relation $ \boldsymbol{v}_t''=\alpha\sum c_i\boldsymbol{v}_i'' $ with  $ \alpha \in\mathbb{C} $ exists. On the other hand, from the $ \sigma $-orthogonality,  $ (\boldsymbol{u}_t'')^\dagger\boldsymbol{u}_i''-(\boldsymbol{v}_t'')^\dagger\boldsymbol{v}_i''=0 $ for $ i=1,\dots,t-1 $ holds. Multiplying this relation by $ c_i $ and taking a sum with respect to $ i $, and using  $ \boldsymbol{u}_t''=\sum c_i\boldsymbol{u}_i'' $ and $ \boldsymbol{v}_t''=\alpha\sum c_i\boldsymbol{v}_i'' $, we obtain $ (1-\alpha)(\boldsymbol{v}_t'')^\dagger\boldsymbol{v}_t''=0 $. Since $ (\boldsymbol{v}_t'')^\dagger\boldsymbol{v}_t''=(\boldsymbol{u}_t'')^\dagger\boldsymbol{u}_t''\ne0 $, we get $ \alpha=1 $. But it implies $ \boldsymbol{z}_t=\sum c_i \boldsymbol{z}_i $, which contradicts the linear independence of $ \boldsymbol{z}_1,\dots,\boldsymbol{z}_t $. (d): The same as (c). 
	\end{proof}
	Then, the properties (iv) and (v) are proved as follows.
	\begin{proof}[Proof of the fundamental properties (iv) and (v) in Subsec.~\ref{subsec:sigmabasis}]
	(iv): It is obvious by Proposition C.\ref{prp:sigmabasis3}. (v):  Let  $ \boldsymbol{z} $ be a zero-norm vector in $ \mathbb{C}^{2N} $ and write it as $ \boldsymbol{z}=(\boldsymbol{u},\boldsymbol{v})^T $ with $  \boldsymbol{u},\boldsymbol{v}\in \mathbb{C}^N $. Since $ \boldsymbol{z}\ne\boldsymbol{0} $ and $ (\boldsymbol{z},\boldsymbol{z})_\sigma=0 $, both $ \boldsymbol{u} $ and $ \boldsymbol{v} $ are nonzero.  $ \mathbb{C}^{2N} $ has an element $ (\boldsymbol{u},\boldsymbol{0})^T $, and the $ \sigma $-inner product between this element and $ \boldsymbol{z} $ is nonzero. Thus, there cannot exist a zero-norm vector $ \boldsymbol{z} $ which is $ \sigma $-orthogonal to all vectors in $ \mathbb{C}^{2N} $. Therefore $ t=0 $ and $ p+q=2N $ follow. By (iv), however, only $ p=q=N $ is possible.
	\end{proof}
	\begin{proof}[Proof of Proposition\ref{prpsigmabasis2}] Let $ p<N $, and let us prove that we can make a positive-norm vector which is $ \sigma $-orthogonal to  $ \boldsymbol{x}_1,\dots,\boldsymbol{x}_p,\boldsymbol{y}_1,\dots,\boldsymbol{y}_q $. Let us write  $ \boldsymbol{x}_i=(\boldsymbol{u}_i,\boldsymbol{v}_i)^T,\ \boldsymbol{y}_i=(\tilde{\boldsymbol{u}}_i,\tilde{\boldsymbol{v}}_i)^T $ with $  \boldsymbol{u}_i,\boldsymbol{v}_i,\tilde{\boldsymbol{u}}_i,\tilde{\boldsymbol{v}}_i \in\mathbb{C}^N $. By Proposition C.\ref{prp:sigmabasis3}(a),  $ \boldsymbol{u}_1,\dots,\boldsymbol{u}_p $ are linearly independent. Since $ p<N $, we can take $ \boldsymbol{u}_{p+1}\in \mathbb{C}^N $ such that $ \boldsymbol{u}_{p+1} $ is orthogonal to all other $ \boldsymbol{u}_i $'s. Using it, we define $ \boldsymbol{w}=(\boldsymbol{u}_{p+1},\boldsymbol{0})^T\in\mathbb{C}^{2N} $, which obviously satisfies  $ (\boldsymbol{w},\boldsymbol{w})_\sigma>0 $ and $ (\boldsymbol{x}_i,\boldsymbol{w})_\sigma=0 $. Furthermore, we define $ \boldsymbol{w}'=\boldsymbol{w}+\sum_j(\boldsymbol{y}_j,\boldsymbol{w})_\sigma\boldsymbol{y}_j $. Then, $ \boldsymbol{w}' $ is $ \sigma $-orthogonal to $ \boldsymbol{x}_1,\dots,\boldsymbol{x}_p,\boldsymbol{y}_1,\dots,\boldsymbol{y}_q $ and has positive norm $ (\boldsymbol{w}',\boldsymbol{w}')_\sigma=(\boldsymbol{w},\boldsymbol{w})_\sigma+\sum_i |\tilde{\boldsymbol{u}}_i^\dagger \boldsymbol{u}_{p+1}|^2 >0$. By the same procedure, we can also make a new negative-norm vector if $ q<N $. Then, we can make a  $ \sigma $-orthonormal basis of $ \mathbb{C}^{2N} $ by repeating this procedure. For the B-orthonormal case, when the above-mentioned $ \boldsymbol{w}' $ is added to the new basis, $ \tau(\boldsymbol{w}')^* $ can be also added.
	\end{proof}
	\begin{proof}[Proof of Theorem~\ref{prplambdaKdecomp}]
	Assume that we find a positive-norm right eigenvector $ \boldsymbol{w}_1 $ with an eigenvalue $ \lambda_1 $. From the properties (viii) and (x) in Subsec.~\ref{subsec:bhbu},  $ \lambda_1 $ is real and $ \tau\boldsymbol{w}_1^* $ is a negative-norm right eigenvector with an eigenvalue $ -\lambda_1 $. By Proposition~\ref{prpsigmabasis2}, there exists a B-unitary matrix  $ U_1 $ such that the first and $ (N+1) $-th column are given by $ \boldsymbol{w}_1 $ and $ \tau\boldsymbol{w}_1^* $, respectively. We then obtain 
	\begin{align}
		U^{-1}_1HU_1=\begin{pmatrix} \lambda_1 &&& \\ & H_{11}' && H_{12}' \\ && -\lambda_1 & \\  & H_{21}' && H_{22}'  \end{pmatrix},
	\end{align}
	where  $ H'=\left( \begin{smallmatrix} H_{11}' & H_{12}' \\ H_{21}' & H_{22}' \end{smallmatrix} \right) $ is a B-hermitian matrix of size $ 2(N-1) $. By iteration, we can reduce the size of $ H $ as long as we find a new finite-norm eigenvector.\\ 
	\indent If there exists a degeneracy in some real eigenvalue $ \lambda $, we first take a  $ \sigma $-orthonormal basis for its eigenspace. (It is possible by the property (i) stated in Subsec.~\ref{subsec:sigmabasis}.) Then, as far as we find positive- and negative-norm vectors in the basis, we repeat the above-mentioned process. The rest zero-norm eigenvectors become a constituent of $ K $. The uniqueness follows from the properties (ii) and (iii)  of  $ \sigma $-orthonormal basis in Subsec.~\ref{subsec:sigmabasis}; the numbers of positive- and negative-norm vectors, $ p $ and $ q $, are unique and the subspace spanned by zero-norm eigenvectors does not depend on a choice of basis. 
	\end{proof}
	From this point forward, we give a few theorems necessary to prove Theorem~\ref{prp:colpa001}. The key lemma is given as follows. 
	\begin{lemappC}[Colpa \cite{Colpa1986}]\label{lem:colpa}Let $ K $ be a B-hermitian matrix such that all eigenvalues are zero and $ \sigma K $ is positive-semidefinite. Then, $ (\sigma K)^{1/2}\sigma (\sigma K)^{1/2}=0 $.
	\end{lemappC}
	This short lemma, appearing in the proof of Lemma B.2 of Ref.~\cite{Colpa1986}, seems to be the most important step to accomplish the construction of the whole theory.
	\begin{proof}Since $ \sigma K $ is a positive-semidefinite hermitian matrix, we can define $ (\sigma K)^{1/2} $ unambiguously. Using the general formula $ \det(\lambda I-AB)=\det(\lambda I-BA) $, we obtain $ \det(\lambda I-K)=\det(\lambda I-\sigma (\sigma K)^{1/2}(\sigma K)^{1/2})=\det(\lambda I-(\sigma K)^{1/2}\sigma (\sigma K)^{1/2}) $. By assumption, $ K $ has only zero eigenvalues, so $ (\sigma K)^{1/2}\sigma (\sigma K)^{1/2} $ also has only zero eigenvalues. However, since $ (\sigma K)^{1/2}\sigma (\sigma K)^{1/2} $ is hermitian, it must be a zero matrix.  
	\end{proof}
	\begin{thmappC}\label{prp:K2} Let $ K $ be a singular B-hermitian matrix of size $ 2n\times 2n $ and satisfy the same assumption with Lemma C.\ref{lem:colpa}. The following (i)-(iii) hold.
	\begin{enumerate}[(i)]
		\item  $ K^2=0 $. 
		\item  $ Let $  $ \boldsymbol{w} $ be an eigenvector of $ \sigma K $ with a positive eigenvalue $ 2\kappa $. Then, $ \sigma\boldsymbol{w} $ is an eigenvector of both $ \sigma K $ and $ K $ with zero eigenvalue.
		\item There exists a B-unitary matrix $ V $ such that
	\begin{align}
		V^{-1}KV=\begin{pmatrix} \tilde{K} & \tilde{K} \\ -\tilde{K} & -\tilde{K} \end{pmatrix},
	\end{align}
	where $ \tilde{K}=\operatorname{diag}(\kappa_1,\dots,\kappa_n) $, and  $ 2\kappa_i(>0) $ is an eigenvalue of $ \sigma K $. Here, $ V $ can be written as $ V=V_0\oplus V_0^* $ with an $ n\times n $ unitary matrix $ V_0 $. Thus, $ V $ is in fact both unitary and B-unitary.
	\end{enumerate}
	\end{thmappC}
	\begin{proof}(i): Multiplying the relation $ (\sigma K)^{1/2}\sigma (\sigma K)^{1/2}=0 $ by $ (\sigma K)^{1/2} $ from left and right, we obtain $ \sigma K^2=0 $. (ii): Multiplying the equation $ \sigma K\boldsymbol{w}=2\kappa \boldsymbol{w} $ by $ K\sigma $ from left and using (i), we obtain $ 2\kappa K\sigma\boldsymbol{w}=\boldsymbol{0} $, and $ \kappa\ne0 $ by assumption. (iii): Let us write positive eigenvalues of $ \sigma K $ as $ 2\kappa_1,\dots,2\kappa_l \ (l\le n) $ with distinguishing multiple roots, and let us write corresponding eigenvectors as $ \boldsymbol{w}_1,\dots,\boldsymbol{w}_l $. We can easily show that if $ \boldsymbol{w}_i $ is an eigenvector with an eigenvalue $ 2\kappa_i $,  $ \tau\boldsymbol{w}_i^* $ is also an eigenvector with the same eigenvalue. Using this symmetry, we can always choose the eigenvector to satisfy $ \boldsymbol{w}_i=\tau\boldsymbol{w}_i^* $ or  $ \boldsymbol{w}_i=-\tau\boldsymbol{w}_i^* $. So, we take each $ \boldsymbol{w}_i $ to satisfy $ \boldsymbol{w}_i=\tau\boldsymbol{w}_i^* $, which can be written as $ \boldsymbol{w}_i=\left( \begin{smallmatrix}\boldsymbol{u}_i \\ \boldsymbol{u}_i^*\end{smallmatrix} \right) $ with $ \boldsymbol{u}_i\in \mathbb{C}^n $. By (ii),  $ \sigma\boldsymbol{w}_i=\left( \begin{smallmatrix}\boldsymbol{u}_i \\ -\boldsymbol{u}_i^*\end{smallmatrix} \right) $ is an eigenvector of $ \sigma K $ and $ K $ with zero eigenvalue. We have now obtained $ 2l $ eigenvectors for $ \sigma K $. Since $ \sigma K $ is positive-semidefinite hermitian, and all positive eigenvalues are already exhausted, the rest eigenvectors have zero eigenvalue, and therefore, they are also eigenvectors of $ K $. Let us write them as $ \left( \begin{smallmatrix}\boldsymbol{u}_{l+1} \\ -\boldsymbol{u}_{l+1}^*\end{smallmatrix} \right),\dots,\left( \begin{smallmatrix}\boldsymbol{u}_{l+l'} \\ -\boldsymbol{u}_{l+l'}^*\end{smallmatrix} \right) $, where $ 2l+l'=2n $. If all eigenvectors shown so far are normalized with respect to hermitian inner product, the unitary matrix which diagonalizes $ \sigma K $ is given by
	\begin{align}
		P=\begin{pmatrix}\boldsymbol{u}_1  & \cdots & \boldsymbol{u}_l & \boldsymbol{u}_{1} &  \cdots & \boldsymbol{u}_{l+l'} \\  \boldsymbol{u}_1^*  &  \cdots  & \boldsymbol{u}_l^* & -\boldsymbol{u}_{1}^* & \cdots & -\boldsymbol{u}_{l+l'}^* \end{pmatrix}.
	\end{align}
	Since $ P $ is invertible, $P^\dagger \sigma P$ is also invertible. From the assumption that $ K $ is singular B-hermitian, its all eigenvectors $ \left( \begin{smallmatrix} \boldsymbol{u}_i \\ \boldsymbol{u}_i^* \end{smallmatrix} \right) \ (i=1,\dots,l+l') $ have zero norm and $ \sigma $-orthogonal to each other. Therefore, we obtain
	\begin{align}
		P^\dagger \sigma P=\begin{pmatrix} \tau_l & O_{2l \times l'} \\ O_{l'\times 2l} & O_{l'\times l'} \end{pmatrix},\ 
	\end{align} 
	but since $ P $ is invertible, $ l'=0 $ and $ l=n $. Let us rescale $ \boldsymbol{u}_i\rightarrow \sqrt{2}\boldsymbol{u}_i $, then $ \boldsymbol{u}_i^\dagger\boldsymbol{u}_i=1 $. Let us define
	\begin{align}
		V=\begin{pmatrix} \boldsymbol{u}_1 & \cdots & \boldsymbol{u}_n & \boldsymbol{0} & \cdots & \boldsymbol{0} \\ \boldsymbol{0} & \cdots & \boldsymbol{0} & \boldsymbol{u}_1^* & \cdots & \boldsymbol{u}_n^* \end{pmatrix},
	\end{align}
	which is both unitary and B-unitary. Using the relations $ K\left( \begin{smallmatrix}\boldsymbol{u}_i \\ \boldsymbol{u}_i^* \end{smallmatrix} \right)=2\kappa_i\left( \begin{smallmatrix}\boldsymbol{u}_i \\ -\boldsymbol{u}_i^* \end{smallmatrix} \right) $ and $ K\left( \begin{smallmatrix}\boldsymbol{u}_i \\ -\boldsymbol{u}_i^* \end{smallmatrix} \right)=\boldsymbol{0} $, we obtain the theorem.
	\end{proof}
	Theorem~\ref{prp:colpa001} is finally obtained as a corollary of Theorem~\ref{prplambdaKdecomp} and Theorem C.\ref{prp:K2}:
	\begin{proof}[Proof of Theorem~\ref{prp:colpa001}] One can soon verify that if $ \sigma H $ is positive-semidefinite,  $ \sigma K $ of the singular part in Theorem~\ref{prplambdaKdecomp} is also positive-semidefinite. So, one can apply Theorem C.\ref{prp:K2}. Let us define $ \tilde{V}=I_r\oplus V_0\oplus I_r\oplus V_0^* $, where $ V_0 $ is an $ (N-r)\times(N-r) $ unitary matrix such that $ V_0\oplus V_0^* $ gives the standard form of the singular part $ K $ as Theorem~C.\ref{prp:K2}. 
	If $ U^{-1}HU $ has the form of Eq.~(\ref{eq:prplambdaKdecomp00}), then $ \tilde{V}^{-1}U^{-1}HU\tilde{V} $ gives Eq.~(\ref{eq:colpathm002}).
	\end{proof}
\section{Existence of the basis satisfying Eqs.~(\ref{eq:orthozero11})-(\ref{eq:orthozero12})}\label{app:basischoice}
	In this appendix, for any B-hermitian matrix $ H_0 $ such that  $ \sigma H_0 $ is positive-semidefinite, we prove that there exists a basis for an eigenspace with zero eigenvalue satisfying the $ \sigma $-orthogonal and orthogonal relations Eqs.~(\ref{eq:orthozero11})-(\ref{eq:orthozero12}). 
	The existence of the block-diagonal form [Eq.~(\ref{eq:UiHU0U})] with an appropriate B-unitary matrix $ U $ is guaranteed by Theorem~\ref{prp:colpa001}. Therefore, if  $ \boldsymbol{x}_i $'s,  $ \tau\boldsymbol{x}_i^* $'s and $ \boldsymbol{y}_i $'s are positive-norm, negative-norm, and zero-norm eigenvectors with zero eigenvalue and $ \boldsymbol{z}_i $'s are generalized eigenvectors satisfying
	\begin{align}
		&H_0\boldsymbol{x}_i=\boldsymbol{0},\quad H_0\tau\boldsymbol{x}_i^*=\boldsymbol{0},\quad (i=1,\dots,s.),  \\
		&H_0\boldsymbol{y}_j=\boldsymbol{0},\quad H_0\boldsymbol{z}_j=2\kappa_j\boldsymbol{y}_j, \quad (j=1,\dots,r.),
	\end{align}
	we can always assume that the $ \sigma $-orthogonal relations
	\begin{align}
		&(\boldsymbol{x}_i,\boldsymbol{x}_j)_\sigma=-(\tau\boldsymbol{x}_i^*,\tau\boldsymbol{x}_j^*)_\sigma=\delta_{ij}, \label{eq:appex01} \\ 
		&(\boldsymbol{y}_i,\boldsymbol{y}_j)_\sigma=(\boldsymbol{y}_i,\boldsymbol{x}_j)_\sigma=(\boldsymbol{y}_i,\tau\boldsymbol{x}_j^*)_\sigma=(\boldsymbol{x}_i,\tau\boldsymbol{x}_j^*)_\sigma=0, \\
		&(\boldsymbol{z}_i,\boldsymbol{z}_j)_\sigma=0,\quad (\boldsymbol{y}_i,\boldsymbol{z}_j)_\sigma=2\delta_{ij} \label{eq:appex02}
	\end{align}
	are satisfied. So, what we should prove is that we can choose a basis for an eigenspace with zero eigenvalue which satisfies the orthogonal relations
	\begin{align}
		&(\boldsymbol{x}_i,\boldsymbol{x}_j)_{\mathbb{C}}=(\tau\boldsymbol{x}_i^*,\tau\boldsymbol{x}_j^*)_{\mathbb{C}}=\frac{1}{\mu_i}\delta_{ij},\quad (\boldsymbol{y}_i,\boldsymbol{y}_j)_{\mathbb{C}}=2\delta_{ij}, \\
		&(\boldsymbol{x}_i,\tau\boldsymbol{x}_j^*)_{\mathbb{C}}=(\boldsymbol{y}_i,\boldsymbol{x}_j)_{\mathbb{C}}=(\boldsymbol{y}_i,\tau\boldsymbol{x}_j^*)_{\mathbb{C}}=0.
	\end{align}
	with keeping the $ \sigma $-orthogonal relations (\ref{eq:appex01})-(\ref{eq:appex02}). (Here, only in this appendix, we use the notation of two kinds of product $ (\cdot,\cdot)_\sigma $ and $ (\cdot,\cdot)_{\mathbb{C}} $ in parallel for brevity.)
	\begin{proof}
	Since $ \sigma H_0 $ is positive-semidefinite and $ \boldsymbol{z}_i $ is not an eigenvector of $ H_0 $ with zero eigenvalue, $ (\boldsymbol{z}_i,H_0\boldsymbol{z}_i)_\sigma=2\kappa_i(\boldsymbol{z}_i,\boldsymbol{y}_i)_\sigma>0 $. So, both $ \kappa_i $ and $ (\boldsymbol{y}_i,\boldsymbol{z}_i)_\sigma $ can be set to be real and positive. If we define $ \boldsymbol{y}_i'=\sqrt{2\kappa_i}\boldsymbol{y}_i $ and $ \boldsymbol{z}_i'=\boldsymbol{z}_i/\sqrt{2\kappa_i} $, we obtain the relation $ H_0\boldsymbol{z}_i'=\boldsymbol{y}_i' $ with keeping the $ \sigma $-orthogonal relation $ (\boldsymbol{y}_i',\boldsymbol{z}_j')_\sigma=2\delta_{ij} $. If we write $ \boldsymbol{y}_i'=(\boldsymbol{\phi}_i,\boldsymbol{\phi}_i^*)^T,\ \boldsymbol{\phi}_i\in\mathbb{C}^N $, by Proposition C.\ref{prp:sigmabasis3}(c),  $ \boldsymbol{\phi}_1,\dots\boldsymbol{\phi}_r $ are linearly independent and the relation $ (\boldsymbol{y}_i',\boldsymbol{y}_j')_\sigma=\boldsymbol{\phi}_i^\dagger\boldsymbol{\phi}_j-\boldsymbol{\phi}_i^T\boldsymbol{\phi}_j^*\propto \operatorname{Im}\boldsymbol{\phi}_i^\dagger\boldsymbol{\phi}_j=0 $ holds. Therefore the $ r\times r $ Gram matrix $ P_{ij}=\boldsymbol{\phi}_i^\dagger\boldsymbol{\phi}_j $ is a positive-definite real symmetric matrix. Then, $ P $ can be diagonalized by a real orthogonal transformation  $ \boldsymbol{y}_i''=\sum_j\boldsymbol{y}_j'O_{ji} $, and  $ \boldsymbol{y}_i'' $ satisfies the orthogonal relation $ (\boldsymbol{y}_i'',\boldsymbol{y}_j'')_{\mathbb{C}}=4\kappa_i'\delta_{ij}, $ where $  2\kappa_i'>0 $ is an eigenvalue of $P$. If the generalized eigenvectors are also transformed by the same orthogonal matrix  $ \boldsymbol{z}_i''=\sum_j\boldsymbol{z}_j'O_{ji} $, the $ \sigma $-orthogonal relation $ (\boldsymbol{y}_i'',\boldsymbol{z}_j'')_\sigma=2\delta_{ij} $ and the relation $ H_0\boldsymbol{z}_i''=\boldsymbol{y}_i'' $ are preserved. Finally, defining $ \boldsymbol{y}_i'''=\boldsymbol{y}_i''/\sqrt{2\kappa_i'} $ and $ \boldsymbol{z}_i'''=\sqrt{2\kappa_i'}\boldsymbol{z}_i'' $, and eliminating the prime symbols from $ \kappa_i',\boldsymbol{y}_i''',\boldsymbol{z}_i''' $, we obtain the basis satisfying the desired orthogonal relation. Next, let us make a basis for $ \boldsymbol{x}_i $'s. By the Gram-Schmidt process, if we define $ \boldsymbol{x}'_i=\boldsymbol{x}_i-\sum_{l=1}^r\frac{(\boldsymbol{y}_l,\boldsymbol{x}_i)_{\mathbb{C}}}{(\boldsymbol{y}_l,\boldsymbol{y}_l)_{\mathbb{C}}}\boldsymbol{y}_l $, they satisfy $ (\boldsymbol{y}_j,\boldsymbol{x}_i')_{\mathbb{C}}=(\boldsymbol{y}_j,\tau(\boldsymbol{x}_i')^*)_{\mathbb{C}}=0 $ with keeping the $ \sigma $-orthogonal relations. Henceforth let us write  $ \boldsymbol{x}'_i $ as $ \boldsymbol{x}_i $ for simplicity. Let $ P=(\boldsymbol{x}_1,\dots,\boldsymbol{x}_s,\tau\boldsymbol{x}_1^*,\dots,\tau\boldsymbol{x}_s^*) $ and let us consider two kinds of Gram matrices, i.e.,  $ P^\dagger P $ for the normal inner product and  $ P^\dagger\sigma P $ for the $ \sigma $-inner product. Since the basis is now chosen as $ \sigma $-orthonormal, the relation $ P^\dagger\sigma P=\sigma_s $ holds. If the basis transformation $ P'=PU $ preserves $ (P')^\dagger\sigma P'=P^\dagger\sigma P $,  $ U $ must be a $ 2s\times 2s $ B-unitary matrix. Since $ P^\dagger P $ is positive-definite, by Theorem~\ref{prp:colpa009}, there exists a B-unitary matrix $ U $ such that $ U^\dagger P^\dagger PU $ is diagonal.
	\end{proof}
\section{Second and third order calculations for type-I and type-II NGMs}\label{app:typeI2nd}
	\indent In this appendix, as a complementary calculation of Sec.~\ref{sec:perturb}, we derive the second order term for type-I mode [Eqs. (\ref{eq:prtrb2storder01}) and (\ref{eq:prtrb2storder02})] and show the absence of the third-order term in type-II modes [Eqs.~(\ref{eq:typeiidsprsn}) and (\ref{eq:typeiidsprsn2})]. \\ 
	\indent Let $ \boldsymbol{\xi}_0 $ be an eigenvector such that $ H_0\boldsymbol{\xi}_0=\boldsymbol{0} $ and $ \epsilon_0=0 $. Then, a perturbative expansion up to third order is given by
	\begin{align}
		H_0\boldsymbol{\xi}_1&=\epsilon_1\boldsymbol{\xi}_0, \label{eq:app01} \\
		H_0\boldsymbol{\xi}_2+\sigma\boldsymbol{\xi}_0&=\epsilon_2\boldsymbol{\xi}_0+\epsilon_1\boldsymbol{\xi}_1, \label{eq:app02} \\
		H_0\boldsymbol{\xi}_3+\sigma\boldsymbol{\xi}_1&=\epsilon_3\boldsymbol{\xi}_0+\epsilon_2\boldsymbol{\xi}_1+\epsilon_1\boldsymbol{\xi}_2. \label{eq:app03}
	\end{align}
	\indent First, let us consider the type-I mode. Following the result of Subsec.~\ref{subsec:finitek}, we take $ \boldsymbol{\xi}_0=\boldsymbol{y}_j,\  \boldsymbol{\xi}_1=\pm\frac{1}{\sqrt{2\kappa_j}}\boldsymbol{z}_j, $ and $ \epsilon_1=\pm\sqrt{2\kappa_j} $. Then, Eq.~(\ref{eq:app01}) becomes an identity and Eqs.~(\ref{eq:app02}) and (\ref{eq:app03}) are given by
	\begin{align}
		H_0\boldsymbol{\xi}_2+\sigma\boldsymbol{y}_j&=\epsilon_2\boldsymbol{y}_j+\boldsymbol{z}_j, \label{eq:app04} \\
		H_0\boldsymbol{\xi}_3\pm\frac{1}{\sqrt{2\kappa_j}}\sigma\boldsymbol{z}_j&=\epsilon_3\boldsymbol{y}_j\pm\frac{\epsilon_2}{\sqrt{2\kappa_j}}\boldsymbol{z}_j\pm\sqrt{2\kappa_j}\boldsymbol{\xi}_2. \label{eq:app05}
	\end{align}
	The $ \sigma $-inner product between $ \boldsymbol{z}_j $ and Eq.~(\ref{eq:app04}) and that between  $ \boldsymbol{y}_j $ and Eq.~(\ref{eq:app05}) yield
	\begin{align}
		2\kappa_j(\boldsymbol{y}_j,\boldsymbol{\xi}_2)_\sigma+(\boldsymbol{z}_j,\sigma\boldsymbol{y}_j)_\sigma&=2\epsilon_2, \\
		(\boldsymbol{y}_j,\sigma\boldsymbol{z}_j)_\sigma&=2\epsilon_2+2\kappa_j(\boldsymbol{y}_j,\boldsymbol{\xi}_2)_\sigma.
	\end{align}
	From them, using the general property $ (\boldsymbol{y}_j,\sigma\boldsymbol{z}_j)_\sigma=(\boldsymbol{z}_j,\sigma\boldsymbol{y}_j)_\sigma^* $, we obtain
	\begin{align}
		\epsilon_2=\frac{\operatorname{Re}(\boldsymbol{z}_j,\sigma\boldsymbol{y}_j)_\sigma}{2},\quad (\boldsymbol{y}_j,\boldsymbol{\xi}_2)_\sigma=-\frac{\mathrm{i}\operatorname{Im}(\boldsymbol{z}_j,\sigma\boldsymbol{y}_j)_\sigma}{2\kappa_j}.
	\end{align}
	On the other hand, $ \boldsymbol{y}_j $ and $ \boldsymbol{z}_j $ generally have the form of $ \boldsymbol{y}_j=(\boldsymbol{\phi},-\boldsymbol{\phi}^*)^T $ and $ \boldsymbol{z}_j=(\boldsymbol{\eta},\boldsymbol{\eta}^*)^T $. Therefore $ (\boldsymbol{z}_j,\sigma\boldsymbol{y}_j)_\sigma=2\mathrm{i}\operatorname{Im}\boldsymbol{\eta}^\dagger\boldsymbol{\phi} $ is pure imaginary. Therefore
	\begin{align}
		\epsilon_2=0,\quad (\boldsymbol{y}_j,\boldsymbol{\xi}_2)_\sigma=-\frac{(\boldsymbol{z}_j,\sigma\boldsymbol{y}_j)_\sigma}{4\kappa_j}.
	\end{align}
	Thus we have proved the absence of the second-order energy for the type-I NGMs. Furthermore, taking the $ \sigma $-inner product between Eq.~(\ref{eq:app04}) and $ \boldsymbol{z}_i, \boldsymbol{w}_i, $ and $ \tau\boldsymbol{w}_i^* $, and using Eq.~(\ref{eq:prtrbxij2}), the expansion coefficients in Eq.~(\ref{eq:prtrbxij}) are given by
	\begin{align}
		d_i^{(2)}=-\frac{(\boldsymbol{z}_i,\sigma\boldsymbol{y}_j)_\sigma}{4\kappa_l},\quad \alpha_i^{(2)}=-\frac{(\boldsymbol{w}_i,\sigma\boldsymbol{y}_j)_\sigma}{\lambda_i}, \quad \beta_i^{(2)}=-\frac{(\tau\boldsymbol{w}_i^*,\sigma\boldsymbol{y}_j)_\sigma}{\lambda_i}.
	\end{align}
	Recalling that $ (\boldsymbol{X},\sigma\boldsymbol{Y})_\sigma\boldsymbol{X}=\boldsymbol{X}\boldsymbol{X}^\dagger\boldsymbol{Y} $,  $ \boldsymbol{\xi}_2 $ can be rewritten as
	\begin{align}
		\boldsymbol{\xi}_2=-\left[ \sum_{i=1}^r\frac{\boldsymbol{z}_i\boldsymbol{z}_i^\dagger}{4\kappa_i}+\sum_{i=1}^{m}\frac{\boldsymbol{w}_i\boldsymbol{w}_i^\dagger}{\lambda_i}+\sum_{i=1}^{m}\frac{\tau\boldsymbol{w}_i^*\boldsymbol{w}_i^T\tau}{\lambda_i} \right]\boldsymbol{y}_j,\label{eq:appthrid0I}
	\end{align}
	which just gives the second order term of Eq.~(\ref{eq:prtrb2storder02}). \\ 
	\indent Next, let us consider the type-II mode. Following the result of Subsec.~\ref{subsec:finitek}, we take  $ \boldsymbol{\xi}_0=\boldsymbol{x}_j,\ \boldsymbol{\xi}_1=\boldsymbol{0},\ \epsilon_1=0, $ and $ \epsilon_2=\frac{1}{\mu_j} $. Then, Eq.~(\ref{eq:app01}) becomes an identity and Eqs.~(\ref{eq:app02}) and (\ref{eq:app03}) are given by
	\begin{align}
		H_0\boldsymbol{\xi}_2+\sigma\boldsymbol{x}_j&=\frac{1}{\mu_j}\boldsymbol{x}_j,\\
		H_0\boldsymbol{\xi}_3&=\epsilon_3\boldsymbol{x}_j. \label{eq:app07}
	\end{align}
	By a similar calculation with the type-I case, we obtain the second order eigenvector as follows:
	\begin{align}
		\boldsymbol{\xi}_2=-\left[ \sum_{i=1}^r\frac{\boldsymbol{z}_i\boldsymbol{z}_i^\dagger}{4\kappa_i}+\sum_{i=1}^{m}\frac{\boldsymbol{w}_i\boldsymbol{w}_i^\dagger}{\lambda_i}+\sum_{i=1}^{m}\frac{\tau\boldsymbol{w}_i^*\boldsymbol{w}_i^T\tau}{\lambda_i} \right]\boldsymbol{x}_j,
	\end{align}
	which is consistent with the second order term of Eq.~(\ref{eq:typeiidsprsn2}). Taking the $ \sigma $-inner product between  $ \boldsymbol{x}_j $ and Eq.~(\ref{eq:app07}), we obtain
	\begin{align}
		\epsilon_3=\frac{(\boldsymbol{x}_j,H_0\boldsymbol{x}_j)_\sigma}{(\boldsymbol{x}_j,\boldsymbol{x}_j)_\sigma}=0.
	\end{align}
	Thus the third order energy for type-II modes generally vanishes. Then, Eq.~(\ref{eq:app07}) reduces to $ H_0\boldsymbol{\xi}_3=\boldsymbol{0} $. However, since $ \boldsymbol{\xi}_j $ with $ j\ge 1 $ does not contain the zeroth-order solution [See Eq.~(\ref{eq:prtrbxij})], we immediately have $ \boldsymbol{\xi}_3=\boldsymbol{0} $. 
\section{Perturbation theory for larger Jordan blocks}\label{app:jordan}
	Here we show the perturbation theory when $ H_0 $ has a Jordan block of size $ n\ge3 $. We can find a fractional dispersion such as $ \epsilon \propto k^{2/n} $ for finite $ k $, but there is at least one complex-valued coefficient, which means that the energy spectrum exhibits a dynamical instability. \\ 
	\indent For simplicity, we only consider the case of zero eigenvalue. If there exists a Jordan block of size $ n $, we can find the generalized eigenvectors satisfying the following relations:
	\begin{align}
		H_0\boldsymbol{w}_0=\boldsymbol{0},\quad H_0\boldsymbol{w}_1=\boldsymbol{w}_0,\quad\dots,\quad H_0\boldsymbol{w}_{n-1}=\boldsymbol{w}_{n-2}.
	\end{align}
	 By Theorem~\ref{prplambdaKdecomp}, such block must be singular B-hermitian, and hence $ \boldsymbol{w}_0 $ must have zero norm: $ (\boldsymbol{w}_0,\boldsymbol{w}_0)_\sigma=0 $. By an appropriate choice of the basis of the generalized eigenspace, we can always take $ \boldsymbol{w}_0,\dots,\boldsymbol{w}_{n-1} $ such that
	 \begin{align}
	 	&(\boldsymbol{w}_0,\boldsymbol{w}_{n-1})_\sigma=(\boldsymbol{w}_1,\boldsymbol{w}_{n-2})_\sigma=(\boldsymbol{w}_2,\boldsymbol{w}_{n-3})_\sigma=\dotsb=(\boldsymbol{w}_{n-1},\boldsymbol{w}_0)_\sigma\ne 0,\\
	 	&(\boldsymbol{w}_i,\boldsymbol{w}_j)_\sigma=0 \quad\text{with}\quad i+j\ne n-1. \label{eq:jordan01}
	 \end{align}
	Let us calculate an eigenvalue and an eigenvector of  $ H=H_0+\sigma k^2 $ perturbatively. The perturbative expansion works well if we expand the eigenvalue and the eigenvector as
	 \begin{align}
	 	\epsilon&=\sum_{j=1}^{n-1} \epsilon_jk^{2j/n}+O(k^2), \\
	 	\boldsymbol{\xi}&=\boldsymbol{w}_0+\sum_{m=1}^{n-1}k^{2m/n}\sum_{j=1}^m\alpha_{m,j}\boldsymbol{w}_j+k^2\boldsymbol{\xi}_2+O(k^{2(n+1)/n}).
	 \end{align}
	Substituting them into $ (H_0+\sigma k^2)\boldsymbol{\xi}=\epsilon\boldsymbol{\xi} $, the coefficients $ \alpha_{i,j} $ are iteratively determined and expressed in terms of $ \epsilon_j $'s. In particular, we obtain $ \alpha_{j,j}=\epsilon_1^j $. On the other hand, the equation for the coefficient of $ k^2 $ is given by
	\begin{align}
		H_0\boldsymbol{\xi}_2+\sigma\boldsymbol{w}_0=\sum_{m=1}^{n-1}\sum_{j=1}^m\epsilon_{n-m}\alpha_{m,j}\boldsymbol{w}_j.
	\end{align}
	Taking the $ \sigma $-inner product between this equation and $ \boldsymbol{w}_0 $ yields
	\begin{align}
		(\boldsymbol{w}_0,\sigma\boldsymbol{w}_0)_\sigma=\epsilon_1^n(\boldsymbol{w}_0,\boldsymbol{w}_{n-1})_\sigma,
	\end{align}
	where Eq.~(\ref{eq:jordan01}) and $ \alpha_{n-1,n-1}=\epsilon_1^{n-1} $ are used. We thus obtain
	\begin{align}
		\epsilon &= \left[\frac{(\boldsymbol{w}_0,\sigma\boldsymbol{w}_0)_\sigma}{(\boldsymbol{w}_0,\boldsymbol{w}_{n-1})_\sigma}k^2\right]^{1/n}+O(k^{4/n}), \label{eq:jordan04}
	\end{align}
	where we consider all possible $ n $-th roots, hence Eq.~(\ref{eq:jordan04}) represents  $ n $ different branches. If $ n\ge 3 $, Eq.~(\ref{eq:jordan04}) always includes at least one dispersion relation with complex coefficient. Thus we conclude that the system has a dynamical instability if  $ H_0 $ has a Jordan block of size $ n\ge 3 $. When $ n=1 $ and 2, it reduces to the type-II and type-I dispersion relations derived in Subsec.~\ref{subsec:finitek}, respectively. Therefore Eq.~(\ref{eq:jordan04}) includes all dispersion relations treated so far.\\  
	\indent Note that the origin of the fractional dispersion (\ref{eq:jordan04}) is completely different from that of ripplons $ \epsilon\sim k^{3/2} $ (Subsecs~\ref{subsec:ripplon} and \ref{subsec:infinite}), because the latter arises from an infinite-dimensional nature of B-hermitian operator and becomes exact only in the infinite-size limit. \section{``Massive'' Nambu-Goldstone modes in the Bogoliubov theory}\label{app:exmassive}
	This appendix is a complement of Subsec.~\ref{eq:subsecexmass}. We give a general result on ``massive'' NGMs \cite{PhysRevLett.110.011602,PhysRevLett.111.021601,Nicolis,2014arXiv1406.6271} and related SSB-originated \textit{finite-energy} solutions. \\ 
	\indent Let us consider the Hamiltonian density
	\begin{align}
		h=\sum_{i=1}^N\nabla \psi_i^* \nabla \psi_i +F(\boldsymbol{\psi}^*,\boldsymbol{\psi})-\mu_1M_1, \label{eq:magGP00}
	\end{align}
	where the model is the same with Eq.~(\ref{eq:GP00}) except for the last term $ \mu_1M_1 $. In the last term, $ \mu_1 $ is a real constant and $ M_1 $ is given by
	\begin{align}
		M_1=\boldsymbol{\psi}^\dagger Q_1\boldsymbol{\psi},
	\end{align}
	where $ Q_1 $ is a generator of the symmetry group $ G $, and hence hermitian. $ M_1 $ is a conserved quantity from Noether's conservation law. The GP equation is given by
	\begin{align}
		\mathrm{i}\partial_t \boldsymbol{\psi}=-\nabla^2\boldsymbol{\psi}+\frac{\partial F}{\partial\boldsymbol{\psi}^*}-\mu_1Q_1\boldsymbol{\psi},
	\end{align}
	where $ \boldsymbol{\psi}=(\psi_1,\dots,\psi_N)^T $ and $ \frac{\partial F}{\partial\boldsymbol{\psi}^*}=(\frac{\partial F}{\partial \psi_1^*},\dots,\frac{\partial F}{\partial \psi_N^*})^T $. The Bogoliubov equation is given by
	\begin{align}
		\mathrm{i}\partial_t\begin{pmatrix}\boldsymbol{u} \\ \boldsymbol{v}\end{pmatrix}=\begin{pmatrix} -\nabla^2+F-\mu_1Q_1 & G \\ -G^* & \nabla^2-F^*+\mu_1Q_1^* \end{pmatrix}\begin{pmatrix}\boldsymbol{u} \\ \boldsymbol{v}\end{pmatrix},
	\end{align}
	where the $ N\times N $ matrices $ F $ and $ G $ are the same with Eq.~(\ref{eq:FijGij}). Let us define
	\begin{align}
		\tilde{\boldsymbol{\psi}}=\mathrm{e}^{-\mathrm{i}\mu_1Q_1t}\boldsymbol{\psi},\quad \tilde{\boldsymbol{u}}=\mathrm{e}^{-\mathrm{i}\mu_1Q_1t}\boldsymbol{u},\quad \tilde{\boldsymbol{v}}=\mathrm{e}^{\mathrm{i}\mu_1Q_1^*t}\boldsymbol{v}.
	\end{align}
	Then, we can show that these tilde-added quantities satisfy the GP and Bogoliubov equations without the term $ -\mu_1M_1 $: 
	\begin{align}
		\mathrm{i}\partial_t \tilde{\boldsymbol{\psi}}&=-\nabla^2\tilde{\boldsymbol{\psi}}+\frac{\partial F}{\partial\tilde{\boldsymbol{\psi}}^*},\\
		\mathrm{i}\partial_t\begin{pmatrix}\tilde{\boldsymbol{u}} \\ \tilde{\boldsymbol{v}} \end{pmatrix}&=\begin{pmatrix} -\nabla^2+F & G \\ -G^* & \nabla^2-F^* \end{pmatrix}\begin{pmatrix}\tilde{\boldsymbol{u}} \\ \tilde{\boldsymbol{v}}\end{pmatrix},
	\end{align}
	where, in proving them, we must pay attention to the fact that the function $ F(\boldsymbol{\psi}^*,\boldsymbol{\psi}) $ satisfies the property $ F(\tilde{\boldsymbol{\psi}^*},\tilde{\boldsymbol{\psi}})=F(\boldsymbol{\psi}^*,\boldsymbol{\psi}) $, because $ U=\mathrm{e}^{-\mathrm{i}\mu_1Q_1t}\in G $. Thus, repeating the same argument in Subsec.~\ref{sec:symzero}, we can obtain SSB-originated zero-mode solutions for the Bogoliubov equation. If we go back to the tildeless notations, the solution can be written as
	\begin{align}
		\begin{pmatrix}\boldsymbol{u} \\ \boldsymbol{v}\end{pmatrix}=\begin{pmatrix} \mathrm{e}^{\mathrm{i}\mu_1Q_1 t}Q_j\mathrm{e}^{-\mathrm{i}\mu_1Q_1 t}\boldsymbol{\psi} \\ -\mathrm{e}^{-\mathrm{i}\mu_1Q_1^* t}Q_j^*\mathrm{e}^{\mathrm{i}\mu_1Q_1^* t}\boldsymbol{\psi}^* \end{pmatrix},\quad j=1,\dots,n=\dim G.
	\end{align}
	They are, however, \textit{time-dependent} solutions unless $ [Q_1,Q_j]=0 $. In order to discuss dispersion relations, we need to get information on stationary eigenstates. We can achieve it using the knowledge of Lie algebra. Every element in the Lie algebra can be classified into a Cartan subalgebra or raising and lowering operators. So either of the following two cases occur: 
	\begin{itemize}
		\item $[Q_1,Q_j]=0$, where $Q_j$ is an element of a Cartan subalgebra.
		\item $[Q_1,Q_\pm]=\pm \alpha Q_\pm$, where $Q_\pm=Q_j\pm\mathrm{i}Q_k$ is a raising and lowering operator, and $ \alpha $ is real and only determined by structure constants of the Lie algebra.
	\end{itemize}
	In the former case, we simply obtain $ \mathrm{e}^{-\mathrm{i}\mu_1Q_1 t}Q_j\mathrm{e}^{\mathrm{i}\mu_1Q_1 t}=Q_j $, so we obtain a zero-energy eigenvector. In the latter case, using the Baker-Campbell-Hausdorff formula, we obtain
	\begin{align}
		\mathrm{e}^{\mathrm{i}\mu_1Q_1 t}Q_\pm\mathrm{e}^{-\mathrm{i}\mu_1Q_1 t}=\mathrm{e}^{\pm\mathrm{i}\mu_1\alpha t}Q_\pm.
	\end{align}
	From them, we obtain an SSB-originated \textit{finite-energy} solution:
	\begin{align}
		\begin{pmatrix}\boldsymbol{u} \\ \boldsymbol{v} \end{pmatrix}=\begin{pmatrix} Q_\pm\boldsymbol{\psi} \\ -Q_{\mp}^*\boldsymbol{\psi}^* \end{pmatrix} \qquad \text{with an eigenvalue }\epsilon=\mp\mu_1\alpha.
	\end{align}
	\indent Let us examine the above result by a familiar example, i.e., the spinor BEC in the presence of magnetic field:
	\begin{align}
		h=\sum_{j=-F}^F|\nabla\psi_j|^2-\mu\rho +h_{\text{int}}-B M_z,
	\end{align}
	where, the model is the same with Subsec.~\ref{sec:spinorgen} except for the last term.  $ B $ represents a magnitude of the magnetic field. Note that the term $ -\mu\rho $ does not break any symmetry. In the present system $ \mu_1M_1=BM_z $, and the commutation relations are given by $ [F_z,F_z]=[F_z,I]=0 $ and $ [F_z,F_\pm]=\pm F_\pm $. Thus, the SSB-originated solutions are given by
	\begin{alignat}{2}
		\begin{pmatrix}\boldsymbol{u} \\ \boldsymbol{v}\end{pmatrix}=&\begin{pmatrix} \boldsymbol{\psi} \\ -\boldsymbol{\psi}^* \end{pmatrix} \quad \text{with } \epsilon=0, &&\qquad \begin{pmatrix} F_z\boldsymbol{\psi} \\ -F_z^*\boldsymbol{\psi}^* \end{pmatrix} \quad \text{with } \epsilon=0, \label{eq:finiteSSB0} \\
		&\begin{pmatrix} F_+\boldsymbol{\psi} \\ -F_-^*\boldsymbol{\psi}^* \end{pmatrix} \quad\text{with } \epsilon=-B, &&\qquad \begin{pmatrix} F_-\boldsymbol{\psi} \\ -F_+^*\boldsymbol{\psi}^* \end{pmatrix} \quad\text{with } \epsilon=+B. \label{eq:finiteSSB}
	\end{alignat}
	If we set $ B=0 $, they simply reproduce Eq.~(\ref{eq:spinorzeromode02}). We can also derive a dispersion relation for finite $ k $ using the perturbation theory in Sec.~\ref{sec:perturb}. For the last mode in Eq.~(\ref{eq:finiteSSB}), the second-order result is given by
	\begin{align}
		\epsilon = B+\frac{\boldsymbol{u}^\dagger\boldsymbol{u}+\boldsymbol{v}^\dagger\boldsymbol{v}}{\boldsymbol{u}^\dagger\boldsymbol{u}-\boldsymbol{v}^\dagger\boldsymbol{v}}k^2=B+\frac{\boldsymbol{\psi}^\dagger(F_+F_-+F_-F_+)\boldsymbol{\psi}}{\boldsymbol{\psi}^\dagger(F_+F_--F_-F_+)\boldsymbol{\psi}}k^2=B+\frac{N_{+-}}{M_z}k^2. \label{eq:lastdance}
	\end{align}
	Here $ N_{+-} $ is a component of a nematic tensor given by Eq.~(\ref{eq:Npm}).\\
	\indent We note that the above discussion is valid only when the symmetry-breaking term is given by a conserved quantity. For example, in spinor BECs, the quadratic Zeeman term $ q N_{zz} $, where $ N_{zz} $ is a $ (z,z) $-component of a nematic tensor (\ref{eq:Nzz}), is also important (e.g., See Ref. \cite{Kawaguchi:2012ii}.). If this term is added, the finite-energy solutions (\ref{eq:finiteSSB}) no longer exist. (On the other hand, the zero-energy solutions (\ref{eq:finiteSSB0}) always exist even in this case.)




\bibliographystyle{model1a-num-names}







\end{document}